\documentclass[conference]{IEEEtran}
\usepackage{amsmath,amssymb,amsthm}
\usepackage{xcolor}
\usepackage{comment}
\usepackage{xspace}
\usepackage{enumitem}
\usepackage[ruled,linesnumbered,noend]{algorithm2e}
\usepackage{subfigure}
\usepackage{multirow}
\usepackage{booktabs}
\usepackage{url}
\usepackage{makecell}
\usepackage{adjustbox}
\usepackage{dsfont}
\usepackage{bbm}
\usepackage{colortbl}

\usepackage{mathtools}
\usepackage[hidelinks]{hyperref}
\usepackage{xurl}

\definecolor{Gray}{gray}{0.9}

\newcommand{\naive}{\textt{D-MoCHy}\xspace}
\newcommand{\ata}{\textt{A2A}\xspace}
\newcommand{\family}{\textt{CODA}\xspace}
\newcommand{\adv}{\textt{CODA-A}\xspace}
\newcommand{\exact}{\textt{CODA-E}\xspace}
\newcommand{\mochyadv}{\textt{MoCHy-A+}\xspace}
\newcommand{\mochyex}{\textt{MoCHy-E}\xspace}

\newcommand\red[1]{\textcolor{red}{#1}}
\newcommand\orange[1]{\textcolor{orange}{#1}}
\newcommand\blue[1]{\textcolor{blue}{#1}}

\newcommand{\rome}[1]{\uppercase\expandafter{\romannumeral #1\relax}}

\newcommand{\smallsection}[1]{\noindent\underline{\smash{\textbf{#1:}}}}

\newcommand{\textt}[1]{\scalebox{1.0}[1.0]{\texttt{#1}}}

\newtheorem{property}{Property}

\newtheorem{lemma}{Lemma}

\newtheorem{problem}{Problem}

\newtheorem{proposition}{\textbf{Proposition}}

\SetKwComment{Comment}{$\triangleright$\ }{}

\SetCommentSty{mycommfont}
\SetAlFnt{\small}

\newcommand{\bus}[1]{\textbf{\underline{\smash{#1}}}}
\newcommand{\ours}{DHGs\xspace}
\newcommand{\our}{DHG\xspace}
\newcommand{\OURS}{\bus{CO}unting of \bus{D}irected Hypergr\bus{A}phlets\xspace}



\IEEEoverridecommandlockouts
    
\begin{document}

\title{Four-set Hypergraphlets for Characterization of Directed Hypergraphs}

\author{\IEEEauthorblockN{
Heechan Moon*\textsuperscript{1}, 
Hyunju Kim*\textsuperscript{1},
Sunwoo Kim\textsuperscript{1}, and
Kijung Shin\textsuperscript{1, 2}}
\IEEEauthorblockA{\textsuperscript{1}Kim Jaechul Graduate School of AI, KAIST, 
\textsuperscript{2}School of Electrical Engineering, KAIST \\ \{heechan9801, hyunju.kim, kswoo97, kijungs\}@kaist.ac.kr}}

\maketitle
\def\thefootnote{*}\footnotetext{These authors contributed equally to this work.}\def\thefootnote{\arabic{footnote}}

\begin{abstract}
    A directed hypergraph, which consists of nodes and hyperarcs, is a higher-order data structure that naturally models directional group interactions (e.g., chemical reactions of molecules).
Although there have been extensive studies on local structures of (directed) graphs in the real world, those of directed hypergraphs remain unexplored.  
In this work, we focus on measurements, findings, and applications related to local structures of directed hypergraphs, and they together contribute to a systematic understanding of various real-world systems interconnected by directed group interactions.

Our first contribution is to define $91$ \textit{directed hypergraphlets} (\ours), which disjointly categorize directed connections and overlaps among four node sets that compose two incident hyperarcs.
Our second contribution is to develop exact and approximate algorithms for counting the occurrences of each \our.
Our last contribution is to characterize $11$ real-world directed hypergraphs and individual hyperarcs in them using the occurrences of \ours, which reveals clear domain-based local structural patterns.
Our experiments demonstrate that our \our-based characterization gives up to $12\%$ and $33\%$ better performances on hypergraph clustering and hyperarc prediction, respectively, than baseline characterization methods.
Moreover, we show that \adv, which is our proposed approximate algorithm, is up to $32\times$ faster %
than its competitors with similar characterization quality.
    
\end{abstract}

\section{Introduction}
\label{sec:intro}

Many real-world systems are composed of group relations among three or more objects (e.g., co-authorships, group discussions, and market baskets \cite{benson2018simplicial}).
Many such group relations are also \textit{directional}.
One representative example is a chemical reaction, which can be interpreted as a directional relation between (a) a group of reactants and (b) a group of products. Other examples include email communications \cite{kim2022reciprocity}, paper citations~\cite{Tang:08KDD}, online Q/As \cite{stackexchange}, and bitcoin transactions~\cite{wu2021detecting}.

Such directional group relations are %
modeled as a \textit{directed hypergraph} (DH), which consists of nodes and hyperarcs. A \textit{hyperarc} is a directed hyperedge consisting of 
two node sets (head and tail sets).
See Figures~\ref{fig:example:data} and \ref{fig:example:DH} for an example chemical reaction %
dataset and a DH that models it.

Several studies have focused on real-world DHs, examining their structural properties, including connectivity~\cite{moyano2016strong} and reciprocity~\cite{kim2022reciprocity}. In addition, machine learning on DHs, such as classification~\cite{tran2020directed},  prediction~\cite{luo2022directed,gracious2023neural,yadati2020nhp}, generation \cite{kim2022reciprocity}, and question answering \cite{yadati2021knowledge}, has been explored.

In this work, we focus on local structures and contribute to a systematic understanding of real-world DHs by answering the following questions: (Q1) how can we characterize or measure the local structures of DHs? (Q2) what are the ``ingredients'' of real-world DHs and how can they be useful? (Q3) how can we rapidly characterize the local structures of a large-scale DH without even seeing the entire DH?

While local structures have been extensively studied for (directed) graphs~\cite{milo2004superfamilies,milenkovic2008uncovering,sarajlic2016graphlet} and recently for undirected hypergraphs~\cite{lee2020hypergraph,lotito2022higher,lee2021thyme+}, they remain still unexplored for DHs.
Especially, using the occurrences of \textit{graphlets} (i.e., induced subgraph isomorphism classes) or their extensions have been successful, with numerous applications, including graph classification \cite{milo2004superfamilies}, clustering \cite{benson2016higher}, and link prediction \cite{abuoda2020link,feng2020link}.
For undirected hypergraphs, Lee et al.~\cite{lee2020hypergraph} defined \textit{h-motifs} (a.k.a. hypergraphlets) based on the overlapping pattern of three node sets. However, in DHs, four sets (and also directions) are involved even in the smallest possible substructure (i.e., two hyperarcs), as shown in Figure~\ref{fig:example:DH}.

\begin{figure}[t!]
     \centering
     \subfigure[ \label{fig:example:data}Example dataset]{%
         \centering
         \includegraphics[width=0.3\linewidth]{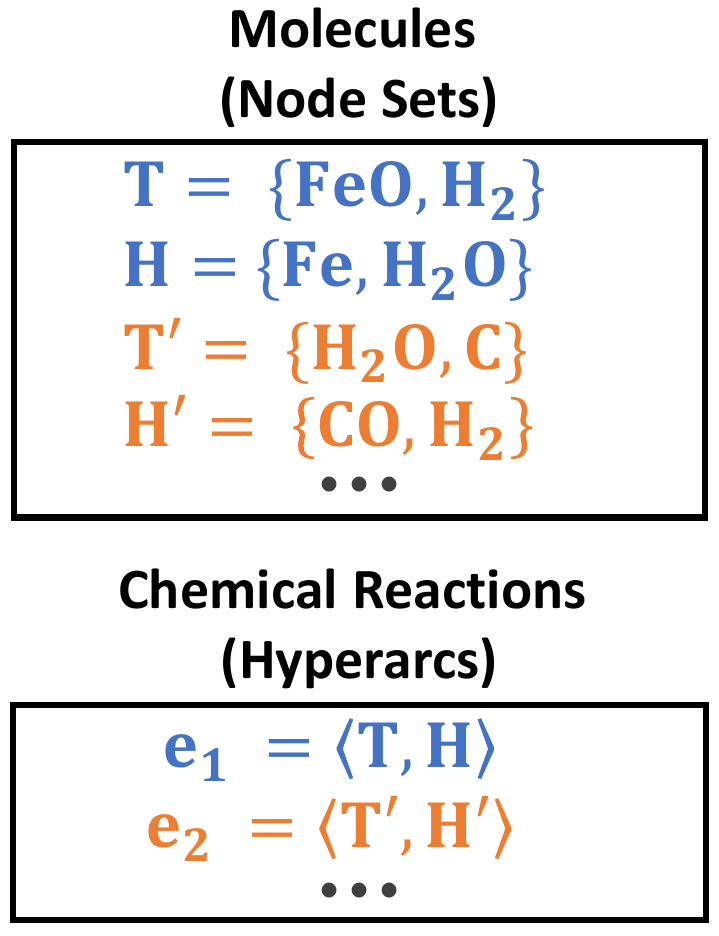}
     }
     \subfigure[ \label{fig:example:DH}Hypergraph]{%
         \centering
         \includegraphics[width=0.3\linewidth]{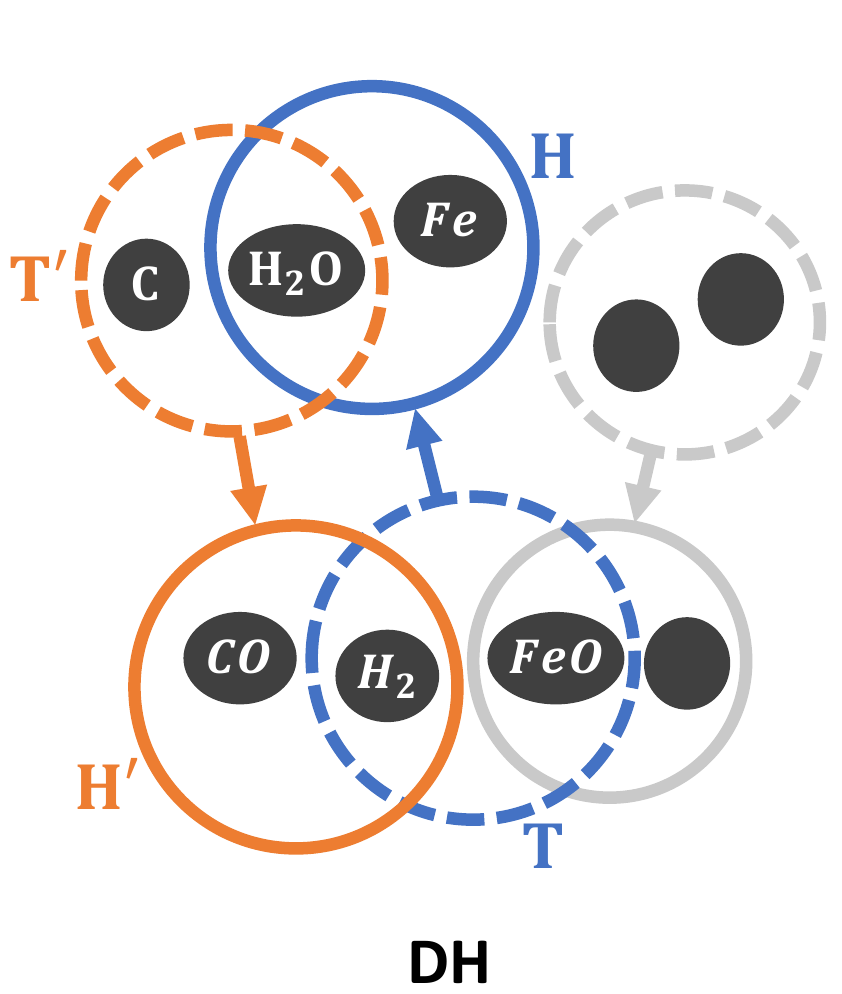}
     }
    \subfigure[\label{fig:example:DHG} DHG-74]{%
         \centering
         \includegraphics[width=0.3\linewidth]{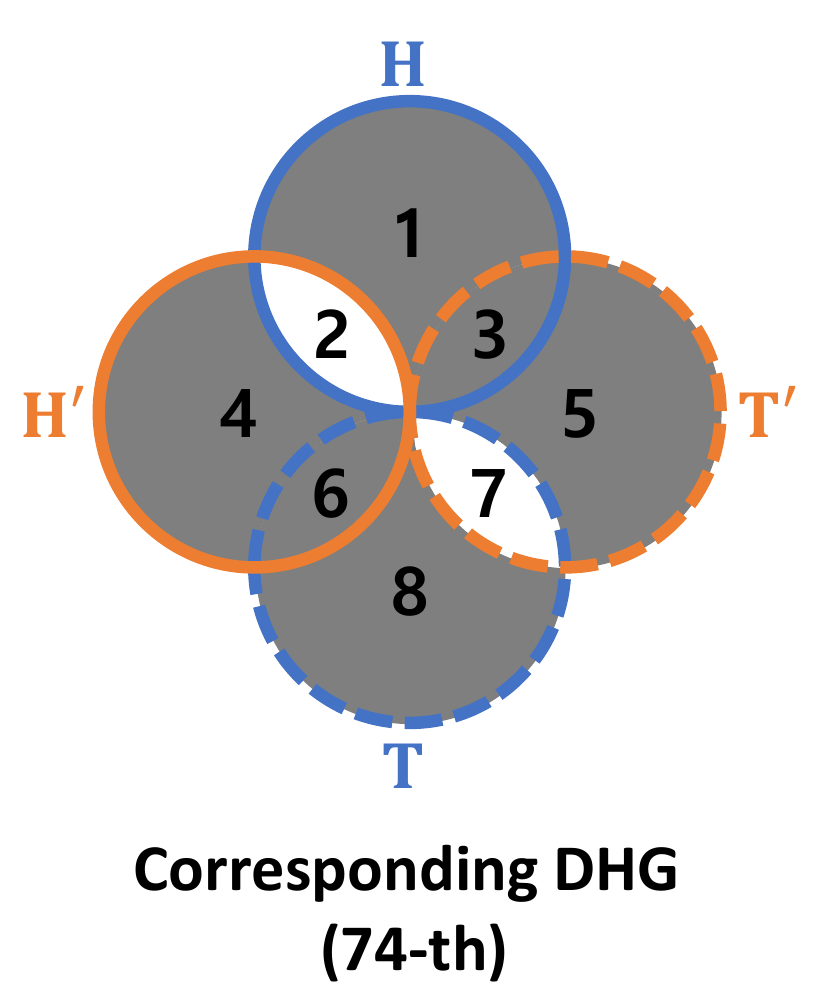}
 } \\
 \vspace{-1mm}
\caption{\label{fig:example} (a) An example chemical reaction dataset with four sets of molecules.
(b) A directed hypergraph that represents the dataset, (c) the directed hypergraphlet (\our) that corresponds to the \our instance depicted with \textbf{\orange{orange}} and \textbf{\blue{blue}} colors in (b).
Note that each of the eight regions in (c) is shaded if and only if at least one node exists in the corresponding region of the \our instance in (b).
For example, the region $4$ (i.e.,  $H'\setminus H \setminus T$) in (c) is shaded, and the corresponding region in (b) is not empty (spec., $\{\mathbf{CO}\}$).}
\end{figure}

Related to (Q1), we define $91$ \textit{directed hypergraphlets} (\ours). They disjointly categorize directed connections and overlaps among four node sets (i.e., two head sets and two tail sets) composing two incident hyperarcs. 
Specifically, they are defined by the patterns based on the emptiness of eight subsets of nodes, resulting from the overlaps of four sets containing directional information.
See Figures~\ref{fig:example:DH} and \ref{fig:example:DHG}
for an example \our and its instance.

Regarding (Q2), we characterize $11$ real-world %
DHs and individual hyperarcs in them using the occurrences of \ours, which reveals that DHs from the same domain share similar local structures.
We also show that using the outputs of our \our-based characterization method as input features results in up to $12\%$ and $33\%$ better performances on hypergraph clustering and hyperarc prediction, resp., than baseline methods, including h-motifs.

Our last contribution, related to (Q3), is the development of  \family (\OURS), a family of fast algorithms for counting the occurrences of each \our.
We propose an exact algorithm \exact and a sampling-based approximate algorithm \adv, proving their accuracy and complexity.
Notably, by prioritizing minority \ours without requiring preprocessing, \adv achieves similar characterization quality up to $32\times$ faster than its best competitors.

In short, our contributions are four-fold:
\begin{itemize}[leftmargin=*]
    \item \textbf{New Concepts.} To the best of our knowledge, we are the first to extend the notion of graphlets to directed hypergraphs (DHs), which lead to 91 directed hypergraphlets (\ours).
    \item \textbf{Discoveries.} Using \ours, we examine and  compare the local structures of $11$ real-world DHs. 
    \item \textbf{Applications.} We demonstrate successful applications of \ours in hypergraph clustering and hyperarc prediction. Especially, we numerically support the superiority of our \our-based characterization over other methods.
    \item \textbf{Fast Algorithms.} We develop fast and theoretically-sound algorithms for counting \our instances.
    We show that our approximate algorithm \adv significantly outperforms baseline approaches in terms of speed, space, and accuracy.
\end{itemize}
For \textbf{reproducibility}, we provide the code and datasets at~\cite{appendix}.

In Section~\ref{sec:related}, we review preliminaries and related studies. In Section~\ref{sec:concepts}, we introduce directed hypergraphlets with their applications.
In Section~\ref{sec:method}, we propose counting algorithms with theoretical analysis. In Section~\ref{sec:experiments}, we present experimental results. In Section~\ref{sec:conclusion}, we conclude our work.

\section{Preliminaries and Related Work}
\label{sec:related}
In this section, we introduce some preliminaries and related studies.
We list the frequently-used symbols in Table \ref{tab:notation}.
We let $[n]:=\{1,\cdots,n\}$ be the set of integers from $1$ to $n$.

\subsection{Preliminaries} 
\smallsection{Basic concepts} A \textit{directed hypergraph} $G=(V,E)$ is an ordered pair of a node set $V = \{v_{1}, \cdots , v_{\lvert V\vert }\}$ and a hyperarc set $E = \{e_{1}, \cdots e_{\lvert E \vert}\}$.
Each hyperarc $e_i = \langle T_i, H_i\rangle$ is defined as an ordered pair of a tail set $T_i \subseteq V$ and a head set $H_i \subseteq V$.
In this work, we assume that $G$ is free of duplicate hyperarcs and
self-loops (i.e., $T_i\cap H_i=\emptyset$, $\forall e_i\in E$), and the rationale for assuming self-loop-free hypergraphs is discussed in Section \ref{sec:definition}.
We generally omit the subscript $i$ unless needed to avoid ambiguity.
For each hyperarc $e=\langle T, H\rangle$, we use $\bar{e}:=T \cup H$ to denote the set of nodes contained in $T$ or $H$. %
We let $E_v:=\{e\in E : v\in \bar{e}\}$ be the set of hyperarcs containing $v$, and let the \textit{degree} $d_v:=|E_v|$ of $v$ be the size of the set.
Two distinct hyperarcs $e$ and $e'$ \textit{incident} if $\bar{e}\cap \bar{e}' \neq \emptyset$.
For each hyperarc $e\in E$, 
$N_e:=\{e' \in E\setminus\{e\}: \bar{e}\cap \bar{e}' \neq \emptyset \}$ denotes the set of hyperarcs incident to  $e$.

\begin{figure*}[t!]
    \vspace{-3mm}
    \centering
    \includegraphics[width=\linewidth]{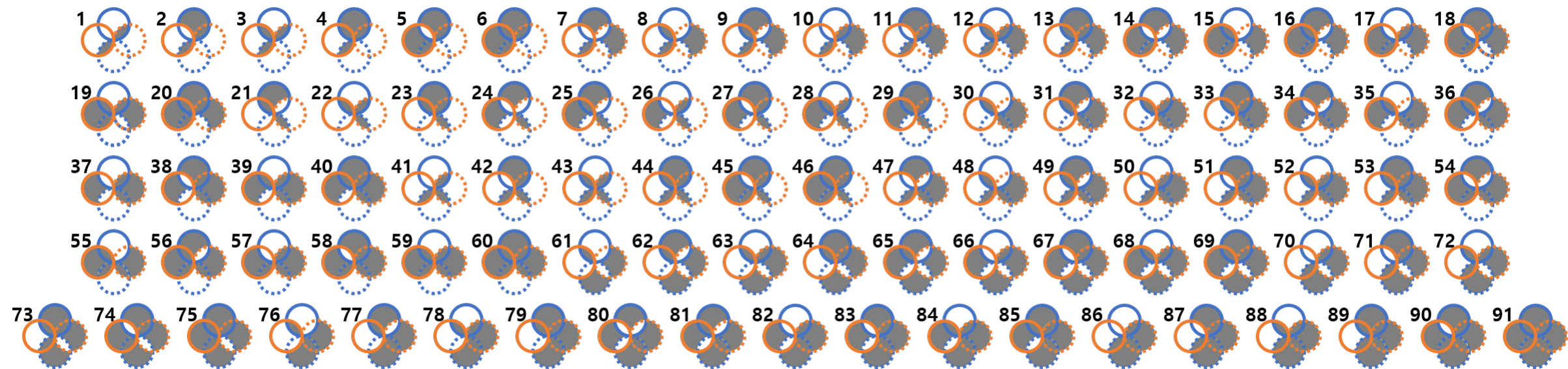}
     \\ \vspace{-2mm}
     \caption{
     91 directed hypergraphlets (DHGs). Each of the two hyperarcs (one is in blue and the other is in orange) is composed
of a tail set (dotted) and a head set (real-line). Each area is shaded if and only if there is at least one element (i.e., node) in it.
     }
     \label{fig:DHG}
\end{figure*}

\begin{table}[t]
    \centering
    \caption{\label{tab:notation} Frequently-used notations}
    \vspace{-2mm}
    \begin{adjustbox}{max width=\linewidth}
    \begin{tabular}{c|l}
        \toprule
        \textbf{Symbol} & \textbf{Definition} \\
        \midrule
        $[n]$ & $\{1,2,\cdots, n\}$ where $n\in \mathbb{N}$\\
         $\lvert A\vert$ & cardinality of a set $A$ i.e., number of elements in $A$ \\
        \midrule
        $G = (V,E)$ & directed hypergraph with nodes $V$ and hyperarcs $E$ \\ 
        $v$ & node in $G$ \\
        $e = \langle T, H\rangle $ & hyperarc with a tail set $T$  and a head set $H$ in $G$ \\
        $\bar{e}$ & set of nodes contained in $e$, i.e., $T\cup H$\\
        $E_v$ & hyperarcs that contain a node $v$ , i.e.,
        $\{e\in E:  v\in \bar{e}\}$ \\
        $d_v$ & degree of a node $v$, i.e., $|E_v|$\\
        $N_e$ & hyperarcs incident to a hyperarc $e$, i.e., $\{e'\in E\setminus\{e\}:  \bar{e}\cap \bar{e}'\neq \emptyset\}$ \\
        \midrule
        $\Omega$ & set of incident hyperarc pairs, i.e., line graph \\
        $f: \Omega\rightarrow [m]$ & function that outputs the \our index corresponding to $(e, e')\in \Omega$\\
        $\Omega_i$ & subset of $\Omega$ such that $f(e, e')=i$\\
        $C[i]$ & count of the instances of \our-$i$\\
        $m$ & number of \ours, i.e., $m=91$ \\
    \bottomrule
    \end{tabular}
    \end{adjustbox}
\end{table}

\subsection{Related work}

\smallsection{Applications and patterns of directed hypergraphs (DHs)}
DHs have been used in various domains to model chemical reactions \cite{ajemni2014modeling}, metabolisms \cite{pearcy2014hypergraph}, bitcoin transactions \cite{ranshous2017exchange}, citation networks \cite{yadati2021graph}, road networks \cite{luo2022directed}, and email communications \cite{kim2022reciprocity}, etc.
Several works explore machine learning on DHs, including  classification \cite{tran2020directed},  prediction \cite{luo2022directed,gracious2023neural,yadati2020nhp},  generation \cite{kim2022reciprocity}, and question answering \cite{yadati2021knowledge}.

The wide applicability of DHs has encouraged the study of the structural properties of real-world DHs, i.e., DHs that model real-world complex systems.
Kim et al. \cite{kim2022reciprocity} examined how reciprocal hyperarcs in real-world DHs are. To this end, they defined a principled measure of reciprocity of DHs, and using it, they revealed that hyperarcs in real-world DHs tend to be more reciprocal than those in random DHs. They also showed that the reciprocal structures could be reproduced by a simple preferential-attachment-like mechanism. %
Moyano et al. \cite{moyano2016strong} 
analyzed the strongly connected components (SCCs)\footnote{
An SCC in a directed hypergraph is a maximal subhypergraph where there exists a directed path between any two hyperedges in it.} of OWL ontologies after modeling them as DHs, and to this end, they proposed fast algorithms for finding SCCs in DHs.
In addition, a number of studies have focused on the structural properties of undirected hypergraphs \cite{benson2018simplicial,comrie2021hypergraph,benson2018sequences,lee2021hyperedges,do2020structural}. 

\smallsection{Graphlet analysis of (directed) graphs}
There exists a long line of research on local structures of (directed) graphs using \textit{graphlets} \cite{milo2004superfamilies,milenkovic2008uncovering,sarajlic2016graphlet}, which are induced isomorphism classes of subgraphs with a fixed number of nodes.
As extended to DHs in Section \ref{sec:method:characterize:hypergraph}, they (a) count the instances of each graphlet in a given graph and (b) measure its significance by comparing it with the count in randomized graphs.

The significances of all graphlets, which we call the characteristic profile (CP), summarize the local structures of the graph. 
CPs enable us to directly compare (e.g., measure the similarity between) the local structures of two graphs of different sizes.
In addition to characterization, graphlets have been proven helpful in various applications, including clustering \cite{benson2016higher}, link prediction \cite{abuoda2020link,feng2020link}, and malware detection \cite{gao2018android}.

Counting the instance of each graphlet is computationally challenging, and thus a large number of exact and approximate counting algorithms have been developed \cite{hovcevar2014combinatorial, ahmed2015efficient,pinar2017escape,bressan2017counting,bressan2018motif,bressan2019motivo}. 
It should be noticed that the computational challenge that these algorithms address comes from the increase of the size of each graphlet, and this is different from the challenge (see Section \ref{sec:weight}) in counting the instances of \ours, where the size (i.e., the number of hyperarcs) is fixed, and thus their ideas are not directly applicable to our problem.

\smallsection{Graphlet analysis of undirected hypergraphs}
Lotito et al. \cite{lotito2022higher} directly extended the notion of graphlets to undirected hypergraphs. They defined \textit{higher-order motifs}, which are isomorphic classes of connected sub-hypergraphs composed of $n$ (spec., $3$ or $4$) nodes. However, this notion is not appropriate when the size of a hyperedge is unbounded or at least can be greater than $n$.
Discarding the hyperedges of size bigger than $n$ or generating new hyperedges through the intersection with a set of $n$ nodes leads to distortion and loss of information that original hypergraphs contain.

In order to overcome such an issue, Lee et al. \cite{lee2020hypergraph} defined \textit{h-motifs} (a.k.a., hypergraph motifs) based on the emptiness of seven regions occurred by three overlapping hyperedges (i.e., node sets), independently of the sizes of hyperedges.
They also provided approximate counting algorithms of defined motifs, based on which the baseline counting algorithm in Section \ref{sec:approx} is designed.
The notion of h-motifs was extended to 3h-motifs \cite{lee2023hypergraph} and TH-motifs \cite{lee2021thyme+} by considering the cardinality (not just emptiness) of the aforementioned regions and the temporal order of hyperedges, respectively.
Below, we provide a more detailed explanation of the conceptual and algorithmic distinctions between our work and the studies conducted in \cite{lee2020hypergraph}, \cite{lee2021thyme+}, and \cite{lee2023hypergraph}.

\begin{itemize}[leftmargin=*]
    \item \textbf{Conceptual differences}: All h-motifs, TH-motifs, and 3h-motifs do not consider the directionality of group interactions, which is an important property of many real-world complex systems. The direct extension of h-motifs to directed hypergraphs (DHs) is not straightforward since h-motif cannot distinguish the role of hyperedges (i.e., sets considered in h-motifs), whereas each hyperarc of DHs consists of two node sets with different semantics (i.e., head and tail). Thus, even if h-motifs are extended to four hyperedges, what it describes is not equivalent to four sets derived from two hyperarcs. Loss of directional information in a naive extension of h-motifs to DHs leads to weak characterization power, as shown in Sections \ref{exp:domain} and \ref{sec:hypergraph-hyperarc-prediction}.
    
    \item \textbf{Algorithmic differences}: Our baseline method, \naive, extends \mochyadv, the most advanced algorithm in \cite{lee2020hypergraph}, to the considered problem, as both methods involve uniformly sampling a pair of incident hyperedges based on a line graph (defined in Section~\ref{sec:definition}).
    Our proposed algorithm, \adv, is different from \mochyadv in that it does not require the computation of a line graph, which is time-consuming and space-intensive.
    Moreover, to reliably estimate the counts of the instances of minority \ours, %
    \adv employs non-uniform sampling.
    Due to these differences, \adv significantly outperforms \naive, as shown in Section \ref{sec:exp:q3}.
    The algorithmic contributions of \cite{lee2021thyme+} and \cite{lee2023hypergraph} are about incorporating temporal information (e.g., sampling of time intervals)
    and dealing with subset cardinality. Thus, they are orthogonal to ours for static DHs.

\end{itemize}

\section{Proposed Concepts \& Applications}
\label{sec:concepts}

In this section, we define the notion of \textit{directed hypergraphlets} (\ours) and discuss how they can be used to characterize (local structures of) directed hypergraphs (DHs) and hyperarcs.

\subsection{Definition of Directed Hypergraphlets}
\label{sec:definition}

\smallsection{Desirable properties}
We aim to define the notion of graphlets (i.e., classes of equivalent subhypergraph) for directed hypergraphs %
to satisfy the following  properties:
\begin{property}[Expressiveness]
Both group and directional information should be captured, without any distortion or loss of information in DHs. Especially, each hyperarc in DHs should not be divided or impaired, regardless of its size.
\end{property}

Note that (a) using directed graphlets \cite{sarajlic2016graphlet} after reducing a DH into a directed graph and (b) using h-motifs
after ignoring directional information in DHs do not satisfy this property.

\smallsection{Definition}
We define directed hypergraphlets (\ours) so that they describe the patterns occurring between any two incident hyperarcs $e=\langle T, H\rangle$ and $e'=\langle T', H' \rangle$, or equivalently among the four node sets (i.e., $T$, $H$, $T'$, and $H'$) that compose the hyperarcs.
Specifically, 
in addition to whether each node set is a head set or a tail set (i.e., directional information) and which hyperarc(s) each node set belongs to (i.e., group information), we check the emptiness of each of the following eight subsets of the node sets: %
(1) $H \setminus H' \setminus T'$, (2) $H \cap H'$, (3) $H \cap T'$, (4) $H' \setminus H \setminus T$, (5) $T' \setminus H \setminus T$, (6) $H' \cap T$, (7) $T \cap T'$, and (8) $T \setminus H' \setminus T'$. 
It should be noticed that the eight subsets correspond to the eight areas in the Venn Diagram for four sets (see Figure \ref{fig:example:DHG} for an example) where, since we assume DHs that are free of self-loops, $T\cap H= \emptyset$ and $T'\cap H'= \emptyset$ always hold.
Out of $6144$ possible cases (the combination of $2^8$ patterns and $4!$ permutations of four sets),
91 distinct \ours in Figure \ref{fig:DHG} remain after excluding (a) duplicate ($H=H'$ and $T=T'$), (b) non-incident ($\bar{e}\cap \bar{e}'=\emptyset$), and (c) symmetric ones.
We call each $i$-th \our in  Figure \ref{fig:DHG} \our-$i$.

\smallsection{An example \our and its instance}
For the \our instance in Figure \ref{fig:example:DH}, the eight areas based on which \ours are defined are:
(1) $H \setminus H' \setminus T'=\{Fe\}$, (2) $H \cap H'=\emptyset$, (3) $H \cap T'=\{H_2O\}$, (4) $H' \setminus H \setminus T=\{CO\}$, (5) $T' \setminus H \setminus T=\{C\}$, (6) $H' \cap T=\{H_2\}$, (7) $T \cap T=\emptyset$, and (8) $T \setminus H' \setminus T'=\{FeO\}$. Thus, it corresponds to DHG-74 in Figures \ref{fig:example:DHG} and \ref{fig:DHG}.

\smallsection{Related concepts}
We call a pair of incident hyperarcs an \textit{instance} of (any of) \ours.
Then, we use $\Omega:=\{(e,e') \in {E \choose  2}: \bar{e} \cap \bar{e}' \neq \emptyset\}$ to denote the set of unordered pairs of incident hyperarc pairs (i.e., the set of instances of \ours), which is also known as the \textit{line graph}.
We use $m=91$ to denote the number of \ours, and we
let $f: \Omega\rightarrow [m]$ be the function that outputs the index of \our which each instance belongs to.
Lastly, $\forall i \in [m]$, we let $\Omega_i=\{(e, e')\in \Omega: f(e, e')=i\}$ be the instances of \our-$i$. %

\smallsection{Rationale for assuming self-loop-free DHs} %
When designing \ours, we assume that the head and tail sets of each hyperarc are disjoint and disregard self-loops, where the two sets intersect. This is based on the following reasons:
\begin{itemize}[leftmargin=*]
    \item \textbf{Interpretability:}  %
    Allowing self-loops in \ours results in an excessive number of possible \ours, with a total of \textbf{16,381}. This large number would make it challenging to interpret them intuitively, while the interpretability of graphlets plays a key role in their usefulness since complex. %
    \item \textbf{Simplicity:} It is a common practice in directed graph analysis to remove self-loops as they typically constitute a small proportion of edges and complicate the analysis \cite{milo2004superfamilies,milenkovic2008uncovering,sarajlic2016graphlet}. In the real-world datasets we consider, self-loops comprised an average of $18.7\%$ of the total hyperarcs.
    \item \textbf{Effectiveness:} Even without considering self-loops, \ours provide an effective characterization of directed hypergraphs, as validated experimentally through hypergraph clustering, hyperarc prediction, and temporal analysis in Sections \ref{exp:domain}, \ref{sec:exp_apply}, and \ref{sec:exp_discover}. 
\end{itemize}

\subsection{Application 1. Hypergraph Characterization}
\label{sec:method:characterize:hypergraph}

In this subsection, we describe how \ours can be used to characterize a directed hypergraph (DH), i.e., to represent (the local structures of) a DH as a fixed-size vector.
Note that, except for the definition of \ours, the below procedure has been employed for various types of graphs \cite{milo2004superfamilies,lee2020hypergraph}.

\smallsection{Randomized directed hypergraph}
For randomization, we use the configuration model \cite{chodrow2020configuration} 
that involves repeatedly shuffling the nodes in two randomly-chosen hyperedges.
To extend this model to DH, we apply shuffling separately to head sets and then tail sets.
It results in a randomized DH $G'$ that preserves the distributions of hyperarc sizes and node degrees in $G$.
Refer to Algorithm \ref{algo:configuration} in Appendix \ref{app:random} for details.

\smallsection{Significance of \ours}
The next step is to measure the significance of each \our in $G$ by comparing the instance count of each \our in $G$ with that in the $G'$.
Specifically, for each $i$, the \textit{significance} of \our-$i$ is defined as:
\begin{equation} \label{equ:significance}
    \mu^G_i := \frac{|\Omega^{G}_{i}|-|\Omega^{G'}_{i}|}{|\Omega^{G}_{i}|+|\Omega^{G'}_{i}|+\epsilon},
\end{equation}
where $|\Omega^{G}_{i}|$ and $|\Omega^{G'}_{i}|$ are the counts of the instances of \our-$i$ in $G$ and $G'$, respectively. 
We set $\epsilon$ to 1 throughout this paper, which is the hyperparameter for cases where $|\Omega^{G}_{i}|+|\Omega^{G'}_{i}|=0$.

\smallsection{Characteristic profile (CP)}
Lastly, we normalize the significance of all \ours to obtain the \textit{CP}, the $m$(=$91$)-dimensional vector where each $i$-th element is defined as:
\begin{equation} \label{equ:CP}
    CP_i^G := \frac{\mu_i^G}{\sqrt{\sum_{i \in [m]} {(\mu_i^G)}^2}}.
\end{equation}
The CP of a DH $G$ summarizes the local structural patterns of $G$ defined by \ours. It should be noticed that the L2-norm of CP is always $1$ regardless of the size of $G$, and thus  using CPs, we can directly compare %
the local structures of two DHs of different sizes.
We empirically demonstrate the effectiveness of this characterization method in Section \ref{exp:domain}.

\subsection{Application 2. Node \& Hyperarc Characterization}

\ours can also be used to characterize (i.e., represent as a fixed-size vector) individual nodes and hyperarcs. Then, the outputs can naturally be used as input features for machine learning tasks, such as node classification and hyperarc prediction, as we show the effectiveness in Section \ref{sec:exp_apply}. For characterizing individual nodes and hyperarcs, we use simply absolute counts since the output vectors are typically compared within a DH. Specifically, each hyperearc $e\in E$ in a DH $G$ is characterized as the $m$(=$91$)-dimensional vector where each $i$-th element is defined as $|\{(e',e'')\in \Omega_{i} : e=e' \text{ or } e=e'' \}|$.
Similarly, each node $v\in V$ in a DH $G$ is characterized as the $m$(=$91$)-dimensional vector where each $i$-th element is defined as $|\{(e',e'')\in \Omega_{i} : v\in \bar{e}'\cup \bar{e}'' \}|$.

\section{Proposed Counting Algorithms}
\label{sec:method}

We present \family (\OURS), a family of exact and approximate counting algorithms for \ours, which aim to address the following problem:
\begin{problem}~\label{sec:prob}
    For a given directed hypergraph (DH) $G$ and for every $i\in [m]$, to compute (or accurately estimate) the count of the occurrences of \our-$i$, (i.e., $|\Omega_i|$). %
\end{problem}
Specifically, we first suggest an exact algorithm and then two approximate algorithms. 
Mathematical analysis of their unbiasedness, variance, and complexity is also provided.

\subsection{Exact Algorithm. \exact}\label{sec:exact}  

\begin{algorithm}[t!]
    \small
    \caption{\exact: Exact Counting of \ours}\label{algo:exact}
    \SetKwInput{KwInput}{Input}
    \SetKwInput{KwOutput}{Output}
    \KwInput{(1) a directed hypergraph: $G=(V, E)$}
    \KwOutput{$C[i]$ for every $i\in [m]$}
    $\Omega\leftarrow \emptyset $\\
    $C[i] \leftarrow 0, \forall i \in [m]$ \label{algo:scan:init} 
    \\
    \ForEach{$e_j\in E$}{ \label{algo:exact:start}
        \ForEach{$v\in e_j$} { 
            \ForEach{$e_k$ such that $v\in e_k$ where $j<k$}{
                \If{$(e_j, e_k)\not\in \Omega$}{                        
                    $\Omega \leftarrow \Omega\cup \{(e_j, e_k)\}$ \label{algo:exact:end}\\
                    $C[f(e_j, e_k)]\leftarrow C[f(e_j, e_k)]+1$ \label{algo:exact:iso}\\
                }
            } 
        }
    }
    \Return{$C$} 
\end{algorithm}

\smallsection{Description}
Pseudocode of \exact for exactly counting the instances of \our-$i$ for every $i\in [m]$ in a given DH $G=(V,E)$ is given in Algorithm \ref{algo:exact}.
For each hyperarc $e_j\in E$, it scans every node $v\in e_j$ and finds every $e_k$ where $v\in e_k$ and $j<k$ so that every element in $\Omega$ is checked \textit{once}. %

\smallsection{Theoretical analysis} 
We \textbf{assume throughout the paper} that the input DH $G$ is stored in the adjacent list format using hash tables, and the number of entries in each hash table is maintained, without incurring additional time complexity.

\begin{lemma}[Time complexity of computing $f(e, e')$]\label{lem:exact_complexity}
    Given a pair of incident hyperarcs, $(e, e')$, the time complexity of computing $f(e, e')$ is $O(min(|\bar{e}|, |\bar{e}'|))$.
\end{lemma}
\begin{proof}
     Assume $|\bar{e}|=min(|\bar{e}|, |\bar{e}'|)$, without loss of generality.
    $f(e, e')=f(\langle T, H\rangle, \langle T', H'\rangle)$ can be computed from checking the emptiness of the following eight sets: (1) $H \setminus H' \setminus T'$, (2) $H \cap H'$, (3) $H \cap T'$, (4) $H' \setminus H \setminus T$, (5) $T' \setminus H \setminus T$, (6) $T \cap H'$, (7) $T \cap T'$, (8) $T \setminus H' \setminus T'$.
    Since the cardinalities of $H, T, H', T'$ are obtained in $O(1)$ time by assumption, we compute $H \cap H'$, $H \cap T'$, $T \cap H'$ and $T \cap T'$ in $O(|\bar{e}|)$ time by checking for each node in $H, T$ whether it also exists in $H'$, or $T'$.
    From the cardinalities of intersection sets, we obtain those of the four other sets in $O(1)$ time as follows:
    \begin{enumerate}
        \item $|H \setminus H' \setminus T'|=|H|-|H \cap H'|-|H \cap T'|$
        \item $|H' \setminus H \setminus T|=|H'|-|H \cap H'|-|T \cap H'|$
        \item $|T' \setminus H \setminus T|=|T'|-|H \cap T'|-|T \cap T'|$
        \item $|T \setminus H' \setminus T'|=|T|-|T \cap H'|-|T \cap T'|$
    \end{enumerate}
            \noindent %
            Hence, the time complexity of computing $f(e, e')$ is $O(|\bar{e}|)=O(\min(|\bar{e}|, |\bar{e}'|))$.
\end{proof}

\begin{proposition}[Time and space complexity of \exact] \label{thm:complexity_exact}
    The time complexity of Algorithm \ref{algo:exact} is $O(\sum_{e\in E} |N_e|\cdot |\bar{e}|)$. Its space complexity is $O(\sum_{e\in E}|\bar{e}|+|\Omega|)$.
\end{proposition}
\begin{proof}  Each $(e, e')\in \Omega$ is accessed $O(|\bar{e}\cap \bar{e}'|)$ times, and $f(e, e')$ is calculated exactly once taking $O(\min(|\bar{e}|, |\bar{e}'|))$ time (Lemma \ref{lem:exact_complexity}). Since $|\Omega|=O(\sum_{e\in E} |N_e|)$, the total complexity is $O(\sum_{(e, e')\in \Omega} \min(|\bar{e}|, |\bar{e}'|))\in O(\sum_{e\in E} |N_e|\cdot |\bar{e}|)$. Regarding space complexity, storing the input DH $G$ and $\Omega$ requires $O(\sum_{e\in E}|\bar{e}|)$ and $O(|\Omega|)$ space, respectively.
\end{proof}

\subsection{Approximate Counting 1. \naive}\label{sec:approx}

\smallsection{Description} 
As the first attempt, we consider \naive, described in Algorithm \ref{algo:naive}. 
\naive can be considered as a direct extension of \mochyadv, the most advanced algorithm presented in \cite{lee2020hypergraph} in that both methods involve uniformly sampling a pair of incident hyperedges using precomputed $\Omega$ (i.e., line graph).
It enumerates all pairs of incident hyperarcs to construct $\Omega$ and then repeats choosing one \ours instance uniformly at random $n$ times, where $n$ is given. Since each $(e, e')\in \Omega$ is chosen with probability $p(e, e')=\frac{1}{|\Omega|}$, an amount of $\frac{1}{n\cdot p(e, e')}$ is added to the (estimated) count of the instances of $f(e, e')$ on each independent trial for unbiasedness, as described below.

\begin{algorithm}[t!]
    \small
    \caption{\naive: Approx. Counting of \ours }\label{algo:naive}
    \SetKwInput{KwInput}{Input}
    \SetKwInput{KwOutput}{Output}
    \KwInput{(1) a directed hypergraph: $G=(V, E)$ \\
            \quad\quad\quad(2) \# of samples $n=q\cdot |E|$ for a given ratio $q$\\} %
    \KwOutput{$C[i]$ for every $i\in [m]$}
     $C[i] \leftarrow 0, \forall i \in [m]$ \\ %
     $\Omega\leftarrow $ the sample space constructed by Line \ref{algo:exact:start}-\ref{algo:exact:end} of Alg. \ref{algo:exact} \label{algo:naive:omega}\\
    \For{$1:n$}{
        Choose $(e, e')\in \Omega$ uniformly at random\\
        $C[f(e, e')]\leftarrow C[f(e, e')]+\frac{|\Omega|}{n}$\label{algo:naive:unb}\\
    }
     \Return{$C$} 
\end{algorithm}

\smallsection{Theoretical analysis} 
    Let $X^i = X_{1}^i$, $X_2^i$, $\dots$ , $X_n^i$ be independent and identically distributed random variables from $\Omega$ to $[m]$ defined as $X^i(e, e')=\frac{1}{n\cdot p(e, e')}\mathbbm{1}[f(e, e')=i]$ where $\mathbbm{1}$ is an indicator function. 
    We can interpret Line \ref{algo:naive:unb} of Algorithm \ref{algo:naive} (equivalently, Line \ref{algo:adv:unb} of Algorithm \ref{algo:adv}) of each $j$-th trial as the result of $X^i_j$.
    Then the output of Algorithm \ref{algo:naive} becomes $C[i]=\sum_{j=1}^n \sum_{(e, e')\in \Omega} X^{i}_j=n\cdot \sum_{(e, e')\in \Omega} X^i$. 
    Using these random variables, we present the unbiasedness, variance, time/space complexity, and sample concentration bound of \naive, below.
    
\begin{proposition}[Unbiasedness of \naive]\label{prop:naive:unb}
   Algorithm \ref{algo:naive} is unbiased, i.e., $\mathbb{E}[C[i]]=|\Omega_i|$, $\forall i\in [m]$. 
\end{proposition}
\begin{proof}
    For each $i\in [m]$, 
    {\small 
    \begin{align*}
        \mathbb{E}[C[i]]&= n\cdot \sum_{(e, e')\in \Omega} \mathbb{E}\left[X^{i}\right]  \because \text{linearity of expectation}
        \\&=n\cdot\sum_{(e, e')\in \Omega} \frac{\mathbbm{1}[f(e, e')=i]}{n\cdot p(e, e')}\cdot p(e, e') 
        \\&=n\cdot \frac{1}{n} \cdot|\{(e,e')\in \Omega: f(e, e')=i\}|=|\Omega_i|. \qedhere
    \end{align*}}
\end{proof}

\begin{proposition}[Variance of \naive]\label{prop:naive:var}
   For each $i\in [m]$, the variance of $C[i]$ obtained by Algorithm \ref{algo:naive} is \begin{align*}Var[C[i]]=\sum_{(e, e')\in \Omega_i}\frac{1}{n}\left(\frac{1}{p(e, e')}-1\right)=\frac{|\Omega_i|(|\Omega|-1)}{n}.\end{align*} 
\end{proposition} 
\begin{proof}
    Since random variables are i.i.d and only one hyperarc pair is considered to increment the count of the corresponding \our on each trial, 
    for all $i\in [m]$,
    {\small 
    \begin{align*}
        &Var[C[i]] =
        n\cdot \sum_{\mathclap{(e, e')\in \Omega} \ } Var[X^i] =n\cdot \sum_{\mathclap{(e, e')\in \Omega} \ } \left(\mathbb{E}[(X^i)^2]-(\mathbb{E}[X^i])^2\right)\\
        & = n\cdot \sum_{\mathclap{(e, e')\in \Omega} \ } \left(\frac{\mathbbm{1}[f(e, e')=i]}{n\cdot p(e, e')}\right)^2 p(e, e')- 
        n\cdot  \sum_{\mathclap{(e, e')\in \Omega} \ }\left(\frac{\mathbbm{1}[f(e, e')=i]}{n}\right)^2 \\&=\frac{|\Omega_i|(|\Omega|-1)}{n}. \qedhere
    \end{align*}}
\end{proof}

\begin{proposition}[Time and space complexity of \naive]\label{prop:naive:comp}
    The time complexity of Algorithm \ref{algo:naive} is $O(\sum_{(e, e')\in \Omega} |\bar{e}\cap \bar{e}'|+n\cdot \max_{e\in E}|\bar{e}|)$. Its space complexity is $O(\sum_{e\in E}|\bar{e}|+|\Omega|)$.
\end{proposition}
\begin{proof}
    Storing the DH $G$ requires $O(\sum_{e\in E}|\bar{e}|)$ space. Constructing $\Omega$ (Line \ref{algo:naive:omega}) takes $O(\sum_{(e, e')\in \Omega}|\bar{e}\cap \bar{e}'|)$ time and $O(|\Omega|)$ space. Sampling itself takes $O(1)$ time and space, but calculating $f(e, e')$ takes $O(\max_{(e, e')\in \Omega}\min(|\bar{e}|, |\bar{e}'|))\in O(\max_{e\in E} |\bar{e}|)$ time (Lemma \ref{lem:exact_complexity}) for each sample $(e, e')$.
\end{proof}

\begin{lemma}[Hoeffding's inequality \cite{hoeffding1994probability}]
\label{lem:hoeff}
Let $X_1, X_2, \dots, X_n$ be independent random variables with $a_j\leq X_j\leq b_j$ for all $j\in [n]$. Consider the sum of random variables $X=X_1+\dots+X_n$. Then for any $t>0$, we have 
\[
    \Pr[|X-\mu|\geq t]\leq 2\exp\left(-\frac{2t^2}{\sum_{j=1}^n (b_j-a_j)^2}\right).
\]
\end{lemma}

\begin{proposition}[Sample concentration bound of \naive]\label{prop:naive:con}
   For any $\epsilon, \delta>0$, if $n\geq \frac{1}{2\epsilon^2}\left(\frac{|\Omega|}{|\Omega_i|}\right)^2\ln (\frac{2}{\delta})$ and $|\Omega_i|>0$, $\Pr(|C[i]-|\Omega_i||\geq |\Omega_i|\cdot \epsilon)\leq \delta$, $\forall i\in [m]$.
\end{proposition} 

\begin{proof}
    Let $t:=|\Omega_i|\cdot\epsilon$. Since $\mathbb{E}[C[i]]=|\Omega_i|$ and $X^i_1$, $X^i_2$, $\cdots$, $X^i_n$ are independent random variables such that $0\leq X^i_j\leq \frac{1}{np(e, e')}=\frac{|\Omega|}{n}$ where $j\in [n]$, we can apply Hoeffding's inequality (Lemma \ref{lem:hoeff}):
    \begin{align*}
        \Pr[|C[i]-|\Omega_i||\geq |\Omega_i|\cdot \epsilon]
        &\leq 2\exp\left(-\frac{2\epsilon^2|\Omega_i|^2}{n(|\Omega|/n)^2}\right)
        \\&\leq 2\exp(-2\epsilon^2n|\Omega_i|^2/|\Omega|^2)\leq \delta. \qedhere
    \end{align*}
\end{proof}

\smallsection{Limitations} 
However, \naive has limits in several aspects. 
Most importantly, the construction of $\Omega$ is expensive in terms of time and space, %
since in many real-world datasets, $|\Omega|$ is significantly larger than $\sum_{e\in E}|\bar{e}|$, as shown in Table~\ref{tab:datasets}.

In addition, \naive is suboptimal for hypergraph characterization, described in Section \ref{sec:method:characterize:hypergraph}. \naive may fail to reliably estimate the counts of the instances of minority \ours, which are small in population, as discussed below. This may act as an obstacle to reliably measuring the CPs.

\vspace{-1mm}
\subsection{Approximate Counting 2. \adv}\label{sec:weight}

\smallsection{Motivation}
In order to overcome the limitations of \naive, regarding minority \ours, we aim to design an advanced approximate algorithm that samples instances of minority \ours with higher probability than \naive. However,  prioritizing the instances of minority \ours during sampling is not straightforward. %
For a solution, we found out through experiments that, in real-world DHs, (1) \ours with a large number of non-empty overlapping regions (i.e., $T\cap H$, $T\cap H'$, $T'\cap H$, and $T'\cap H'$) tend to be minorities, and (2) the number of non-empty overlapping regions in a DHG tend to become larger as $|\bar{e}\cap \bar{e}'|$ is bigger in each \ours instance $(e,e')\in \Omega$ (see Figure \ref{fig:overlap}). It concludes that an instance whose intersection is big is more likely to belong to a minority DHG.  
Therefore, we design \adv so that it samples an instance with a large intersection preferentially (spec., with probability proportional to the size of intersection), without computationally cost preprocessing, which is another limitation of \naive.

\smallsection{Description}
As a kind of weighted sampling, we propose \adv, which is faster and more specialized to sample the instances of minority \ours. Instead of constructing $\Omega$, \adv only requires the weight $w[v]=\binom{d(v)}{2}$ of each node $v$, as described in Algorithm \ref{algo:adv}.

\smallsection{Analysis}
At each trial, $(e, e')\in \Omega$ is sampled with probability $p(e, e')=\frac{|\bar{e}\cap \bar{e}'|}{\sum_{v\in V} w[v]}$, which sums up to 1 by the fact that $\sum_{v\in V} w[v] = \sum_{(e, e')\in \Omega} |\bar{e}\cap \bar{e}'|$. Despite this bias in sampling, the estimated counts from \adv remain unbiased. 
Below, we provide the unbiasedness,
variance, time/space complexity, and sample concentration
bound of \adv.

\begin{algorithm}[t!]
    \small
    \caption{\adv: Approx. Counting of \ours }\label{algo:adv}
    \SetKwInput{KwInput}{Input}
    \SetKwInput{KwOutput}{Output}
    \KwInput{(1) a directed hypergraph: $G=(V, E)$ \\
            \quad\quad\quad(2) \# of samples $n=q\cdot |E|$ for a given ratio $q$\\} %
    \KwOutput{$C[i]$ for every $i\in [m]$}
     $C[i] \leftarrow 0, \forall i \in [m]$ \\ %
     $w[v]\leftarrow \binom{d(v)}{2}, \forall v\in V$ \label{algo:adv:w}\\
    \For{$1:n$}{
        Choose $v$ from distribution $P(v)\propto w[v]$ \\   
        Choose $(e, e')\in \binom{E_v}{2}$ uniformly at random\\
        $C[f(e, e')]\leftarrow C[f(e, e')]+\frac{\sum_{v\in V} w[v]}{n\cdot |\bar{e}\cap \bar{e}'|}$ \label{algo:adv:unb}\\
    }
     \Return{$C$} 
\end{algorithm}

\begin{figure}[t]
    \vspace{-3mm}
     \centering
     \includegraphics[width=0.5\textwidth]{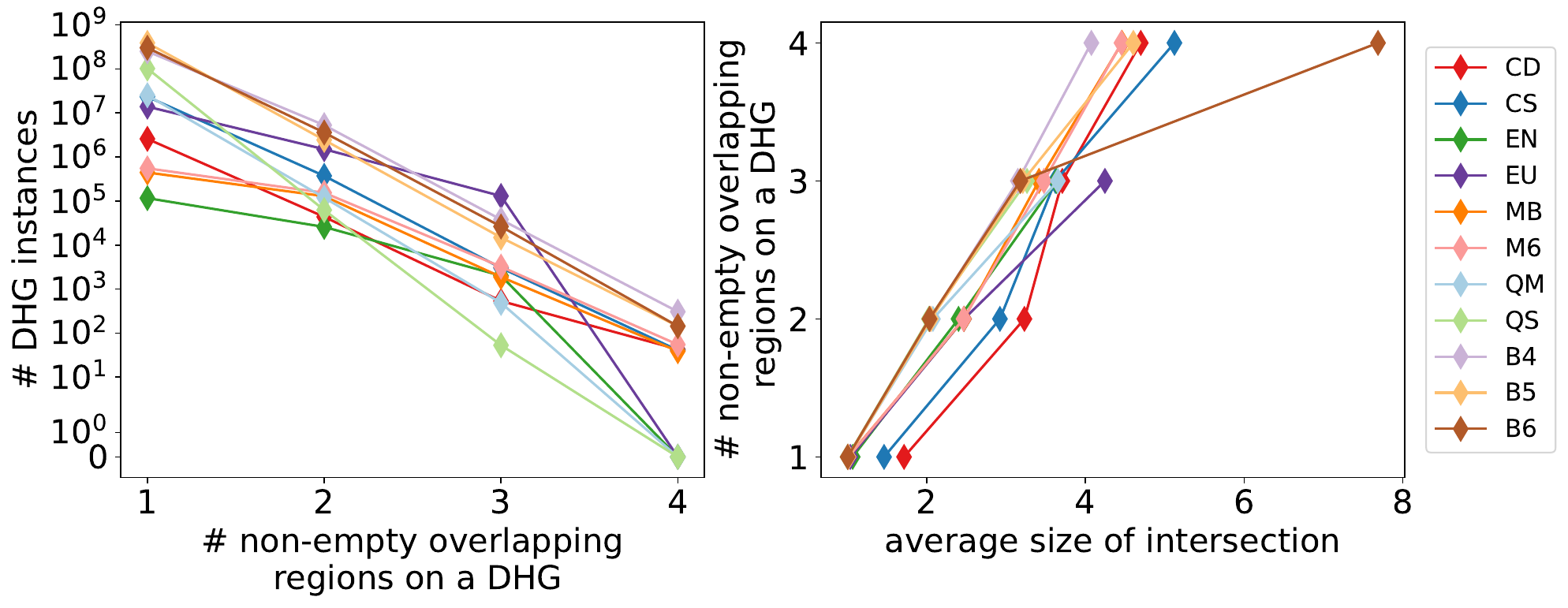} \\
     \vspace{-2mm}
     \caption{As the number of non-empty overlapping regions on a \our increases, the count of its instances tends to be smaller (left). As the size of the intersection (i.e., $|\bar{e}\cap \bar{e}'|$) increases, the number of non-empty overlapping regions on the corresponding \our tends to be larger (right). This implies that an instance whose size of intersection is large is more likely to belong to minority \ours.}
     \label{fig:overlap}
\end{figure}

\begin{proposition}[Unbiasedness of \adv]\label{prop:adv:unb}
   Algorithm \ref{algo:adv} is unbiased, i.e., $\mathbb{E}[C[i]]=|\Omega_i|$, $\forall i\in [m]$. 
\end{proposition}
\begin{proof}
Follow the flow of the proof of Proposition \ref{prop:naive:unb}.
\end{proof}

\begin{proposition}[Variance of \adv]\label{prop:adv:var}
   For each $i\in [m]$, the variance of $C[i]$ obtained from Algorithm \ref{algo:adv} is 
   {\small
   \begin{align*}
   Var[C[i]]=\sum_{\mathclap{(e, e')\in \Omega_i}}\frac{1}{n}\left(\frac{1}{p(e, e')}-1\right)
   =\sum_{\mathclap{(e, e')\in \Omega_i}}\frac{1}{n}\left(\frac{\sum_{v\in V}w[v]}{|\bar{e}\cap \bar{e}'|}-1\right).
   \end{align*}} 
\end{proposition} 
\begin{proof}
Follow the flow of the proof of Proposition \ref{prop:naive:var}.
\end{proof}

\begin{proposition}[Time and space complexity of \adv]
    \label{prop:adv:comp}
    The time complexity of Algorithm \ref{algo:adv} is $O(|V|+n\cdot(\log |V|+\max_{e\in E} |\bar{e}|))$. Its space complexity is $O(\sum_{e\in E}|\bar{e}|+|V|)$, bounded by $O(\sum_{e\in E}|\bar{e}|)$.
\end{proposition}
\begin{proof}
    Storing the DH $G$ requies $O(\sum_{e\in E}|\bar{e}|)$ space.
    Constructing the weight table $w$ (Line \ref{algo:adv:w}) takes $O(|V|)$ time and space by our assumptions in Section \ref{sec:exact}.
    Sampling $v$, $e$, and $e'$ takes       
    $O(\log |V|)$ time, and calculating $f(e, e')$ takes $O(\max_{(e, e')\in \Omega}\min(|\bar{e}|, |\bar{e}'|))\in O(\max_{e\in E} |\bar{e}|)$ time (Lemma \ref{lem:exact_complexity}), without increasing the space complexity.
\end{proof}

\begin{proposition}[Sample concentration bound of \adv]\label{prop:adv:con}
   For each $i\in [m]$, let $W=\sum_{v\in V} w[v]$ and $\gamma_i = \min_{(e,e')\in \Omega_i} |\bar{e}\cap \bar{e'}|$. Then for any $\epsilon, \delta>0$, if $n\geq \frac{1}{2\epsilon^2}\left(\frac{W}{\gamma_i|\Omega_i|}\right)^2\ln (\frac{2}{\delta})$ and $|\Omega_i|>0$, $\Pr(|C[i]-|\Omega_i||\geq |\Omega_i|\cdot \epsilon)\leq \delta$. \qedhere
\end{proposition} 

\begin{proof}
    Let $t:=|\Omega_i|\cdot\epsilon$. Since $\mathbb{E}[C[i]]=|\Omega_i|$ and $X^i_1$, $X^i_2$, $\cdots$, $X^i_n$ are independent random variables such that $0\leq X^i_j\leq \max\frac{1}{np(e, e')}=\frac{W}{n\gamma_i}$ where $j\in [n]$, we can apply Hoeffding's inequality (Lemma \ref{lem:hoeff}):
    \begin{align*}
        \Pr[|C[i]-|\Omega_i||\geq |\Omega_i|\cdot \epsilon]
        &\leq 2\exp\left(-\frac{2\epsilon^2|\Omega_i|^2}{n(W/n\gamma_i)^2}\right)
        \\& \leq 2\exp\left(-2\epsilon^2n\gamma_i^2|\Omega_i|^2/W^2\right)\leq \delta. 
    \end{align*}
    \qedhere
\end{proof}

\smallsection{Advantages}
The superiorities of \adv over \naive are summarized as follows.

\begin{itemize}[leftmargin=*]
    \item For pre-computation, while \naive requires $O(|\Omega|)$ time and space, \adv only takes $O(|V|)$ time and space.
    \item Based on real data statistics, it samples relatively more instances of minority \ours than \naive, leading to higher accuracy in estimating CP values, as shown in Section \ref{sec:exp:q3}. 
\end{itemize}

\section{Experiments}
\label{sec:experiments}

\begin{table}[t!]
    \vspace{-2mm}
    \centering
    \caption{\label{tab:datasets} Summary of 11 directed hypergraphs from 5 domains.
    }
    \vspace{-2mm}
    \scalebox{0.7}{
    \begin{tabular}{l|l|r|r|r|r}
        \toprule
        \textbf{Name} &\textbf{Full Name} & \textbf{$|V|$} & \textbf{$|E|$} & \textbf{$\sum_{e\in E} |\bar{e}|$} & \textbf{$|\Omega|$} \\
        \midrule
        \textt{MB} & \textt{metabolic-iAF1260b}  & $1,668$ & $2,064$ & $8,795$ & $574,769$ \\ 
        \textt{M6} & \textt{metabolic-iJO1366}  & $1,805$ & $2,233$ & $9,696$ & $709,151$ \\
        \midrule
        \textt{EN} & \textt{email-enron} &  $110$ & $1,447$ & $4,057$ & $143,980$ \\ 
        \textt{EU} & \textt{email-eu} &  $986$ & $34,485$ & $88,781$ & $15,512,713$ \\
        \midrule
        \textt{CD} & \textt{citation-data-science} &  $46,646$ & $38,144$ & $284,357$ & $2,619,227$ \\ 
        \textt{CS} & \textt{citation-software} & $94,886$ & $115,617$ & $827,693$ & $23,715,905$ \\ 
        \midrule
        \textt{QM} & \textt{qna-math} & $34,635$ & $83,425$ & $234,091$ & $25,268,639$ \\ 
        \textt{QS} & \textt{qna-server} & $163,508$ & $238,838$ & $659,319$ & $101,898,595$ \\ 
        \midrule
        \textt{B4} & \textt{bitcoin-2014} & $1,697,625$ & $1,164,119$ & $3,917,581$ & $255,367,684$ \\ 
        \textt{B5} & \textt{bitcoin-2015} & $1,961,886$ & $1,237,599$ & $4,317,923$ & $391,330,408$ \\ 
        \textt{B6} & \textt{bitcoin-2016} & $2,009,978$ & $1,293,604$ & $4,317,876$ & $304,737,356$ \\ 
        \bottomrule
    \end{tabular}
    }
\end{table}

We perform experiments on $11$ real-world directed hypergraphs (DHs), aiming to answer the following questions:
\begin{itemize}[leftmargin=*]
    \item \textbf{Q1. Hypergraph Characterization Power:} How accurately can we characterize DHs using \ours?
    Are there any patterns useful for classifying real-world DHs?
    \item \textbf{Q2. Hyperarc Characterization Power:} 
    How accurately can we characterize hyperarcs using \ours?
    Is the characterization method useful for hyperedge prediction?
    \item \textbf{Q3. Performance of Counting Algorithms:} How fast, memory-efficient, and accurate are \our counting algorithms? How do they depend on the number of samples?
    \item \textbf{Q4. Temporal Analysis:} What interesting temporal patterns do \ours reveal in real-world DHs?
\end{itemize}

\subsection{Experimental Setting}
\label{sec:settings}

\smallsection{Dataset} \label{exp:dataset}
We use $11$ real-world directed hypergraph datasets from $5$ different domains.
Some statistics of these datasets are given in Table \ref{tab:datasets}. %
Their semantics, sources, and preprocessing details are provided in Appendix D \cite{appendix}.

\smallsection{Competitors} \label{sec:hypergraph-competitors}
We evaluate the characterization power of \ours by comparing it to three baseline methods: %
\begin{itemize}[leftmargin=*]
    \item \textbf{H-motifs \cite{lee2020hypergraph} and 3h-motifs \cite{lee2023hypergraph} after hypergraph expansion.}
    A hypergraph expansion $H$ of a DH $G=(V, E)$ is an undirected hypergraph with a node set $V$ and a hyperedge set $\{\bar{e}: e \in E\}$.
    That is, the tail set and head set are merged to make an undirected hyperedge.
    We characterize $H$ as $26$ and $431$-dimensional vectors, using h-motifs and 3h-motifs, respectively, as suggested in \cite{lee2020hypergraph}, \cite{lee2023hypergraph}.
    \item \textbf{Triads with di-biclique expansion \cite{holland1977method}} A di-biclique expansion $D$ of a $G=(V, E)$ is a directed graph with a node set $V$ and an edge set $\bigcup_{e=\langle T, H\rangle\in E}\{(t, h): t\in T,\; h\in H\}$, where $t$ and $h$ are a node of the tail and head sets, respectively.
    We characterize $D$ as a $13$-dimensional vector, using triads (i.e., isomorphic classes of weakly connected 3-node directed subgraphs), as in Section~\ref{sec:method:characterize:hypergraph}.
    
\end{itemize}

\begin{figure}[t!]
    \vspace{-3mm}
     \centering
     \includegraphics[width=0.49\textwidth]{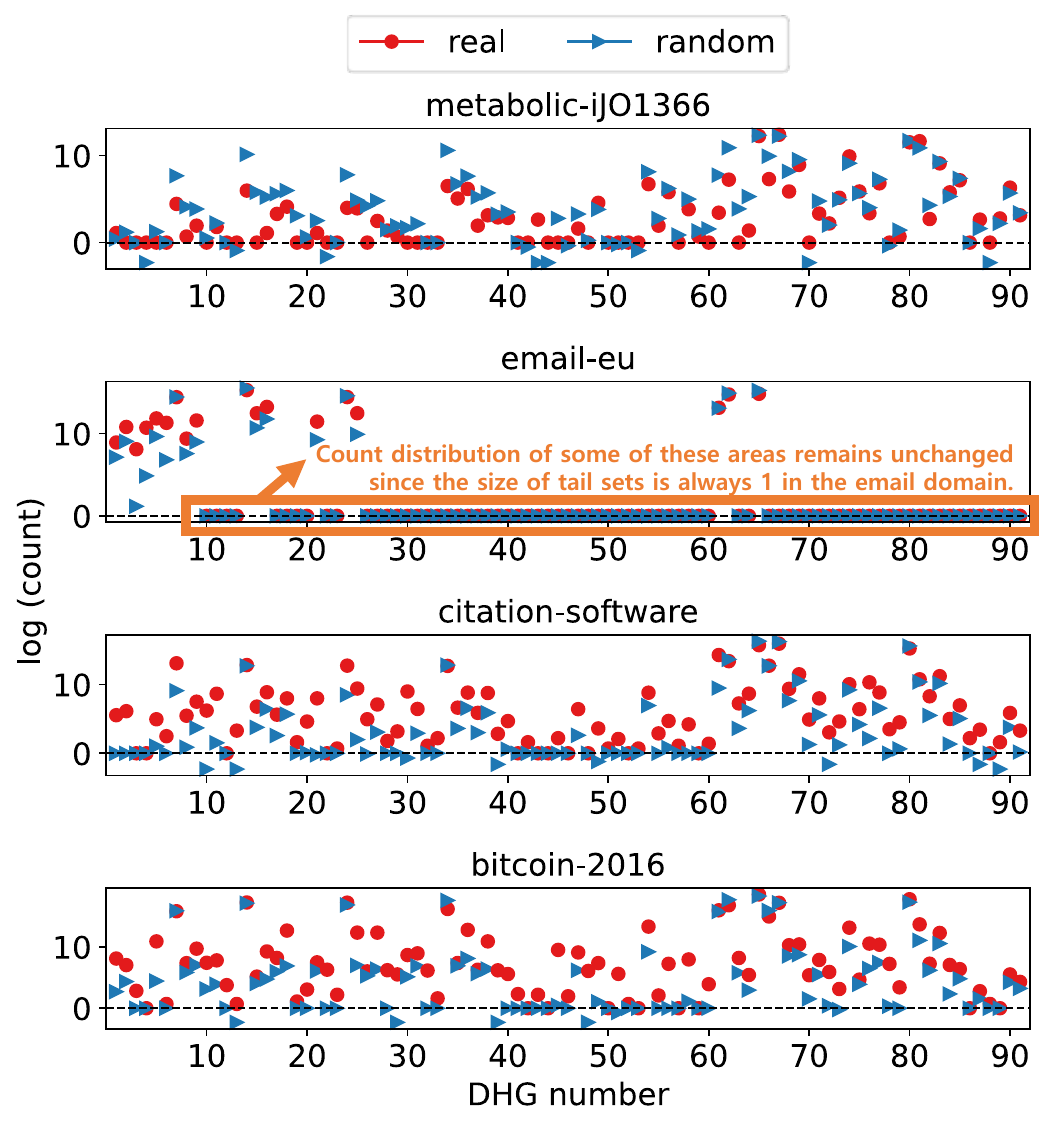} \\
     \vspace{-2mm}
     \caption{\label{fig:count_ratio_citation} Log counts of \ours in real-world and randomized directed hypergraphs (DHs). 
     The counts of \ours are clearly distinguished in real-world and randomized DHs.
     See Appendix E \cite{appendix} for full results on all DHs.}
\end{figure}

\begin{figure}[t!]
    \vspace{-3mm}
    \centering
    \includegraphics[width=\linewidth]{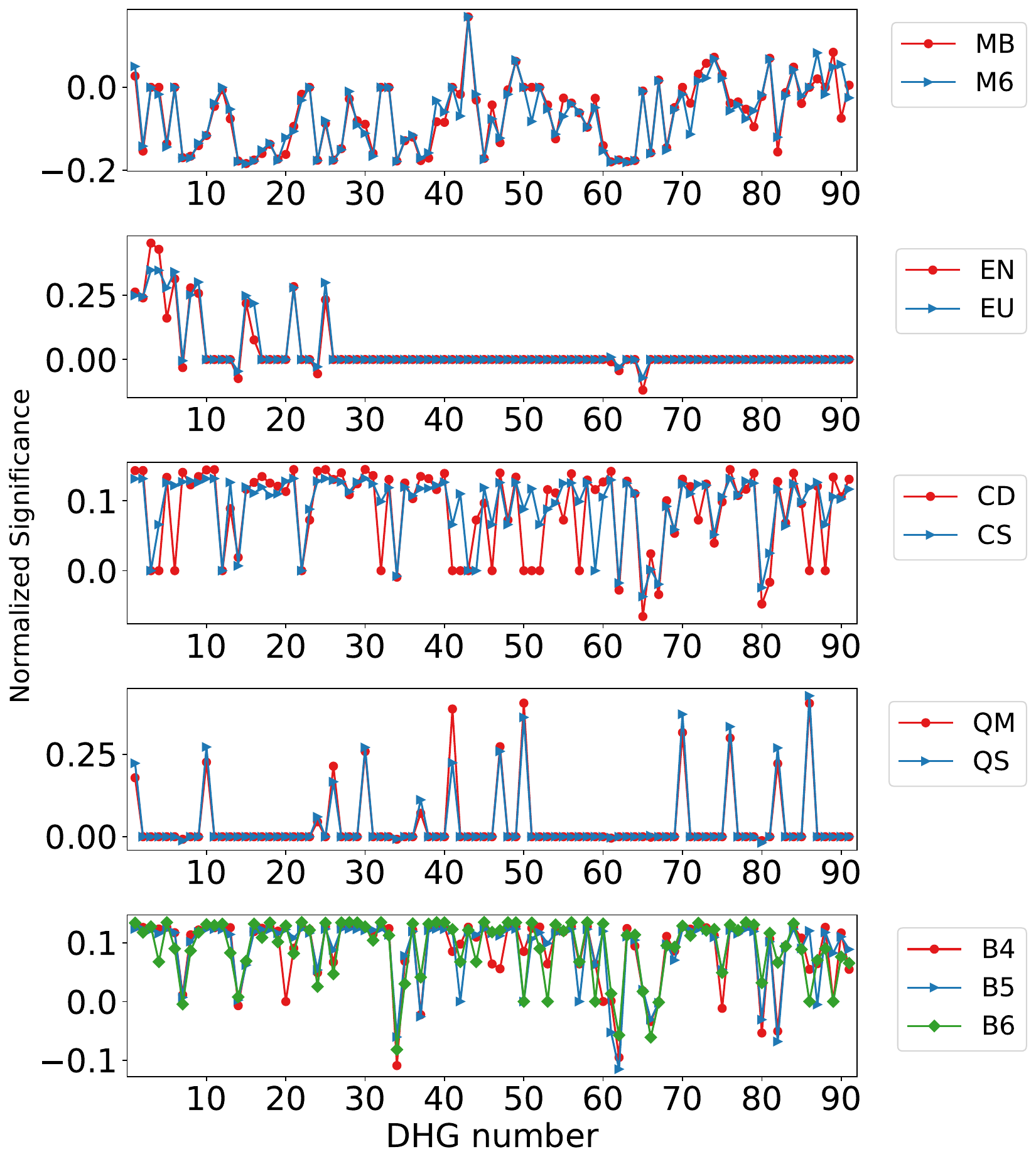} \\
    \vspace{-2mm}
     \caption{\label{fig:CPs_plot} Characteristic profiles (CPs) of \textt{metabolic}, \textt{email}, \textt{citation}, \textt{qna}, and \textt{bitcoin} datasets with ten randomized directed hypergraphs (DHs).
     The CPs of DHs from the same domain tend to be remarkably similar, while those from different domains tend to be different.
     }
\end{figure}

\smallsection{Machine and Implementation} 
We conduct all experiments on a machine with a 3.7GHz Intel i5-9600K CPU and 64GB RAM. 
For \ours, 
we implement all counting algorithms in C++.
For counting h/3h-motifs, we use the official C++ implementations \cite{lee2020hypergraph,lee2023hypergraph}. %
For %
counting triads, we implement a widely-used subquadratic algorithm~\cite{batagelj2001subquadratic} in C++.

\subsection{Q1. Hypergraph Characterization Power} \label{exp:domain}

We examine how accurately \ours characterize real-world directed hypergraphs (DHs), compared to existing methods.

\smallsection{Count distributions}
We analyze the occurrence distributions of \ours in real-world and randomized DHs.
For statistical significance, we generate ten randomized DHs and report the average counts.
Figure~\ref{fig:count_ratio_citation} shows that the counts of \ours in real-world DHs are distinct from those in randomized DHs.
Note that, in the email domain, the counts of some \ours (e.g., \ours-10, -11, and -12, etc.) are always 0, regardless of randomization,  since every tail set in this domain has size 1.

\smallsection{Domain-specific patterns (hypergraph clustering)}
For each real-world DH, we compute characteristic profiles (CPs) by using the average over ten randomized DHs as $|\Omega^{G'}_{i}|$ for each $i$ in Eq.~\eqref{equ:significance}. As shown in Figure~\ref{fig:CPs_plot}, CPs of DHs from the same domain tend to be similar, while those from the different domains tend to be different.
This tendency is confirmed numerically in Figure~\ref{fig:CPs}(a), where we calculate the similarity between CPs, using the Pearson correlation coefficient,
and plot their similarity matrix.
For instance, the similarity between \textt{MB} and \textt{M6} is 0.9515, while the average similarity between \textt{MB} and the other DHs except for \textt{M6} is -0.0017. 
Moreover, we perform spectral clustering using the similarity matrix from CPs as the input, and DHs are perfectly clustered based on their domains, as shown in Figure~\ref{fig:CPs}(e).

\begin{figure*}[t!]
     \centering
     \vspace{-3mm}
     \subfigure[\label{fig:CP-DHG} \ours]{%
         \centering
         \includegraphics[width=0.19\textwidth]{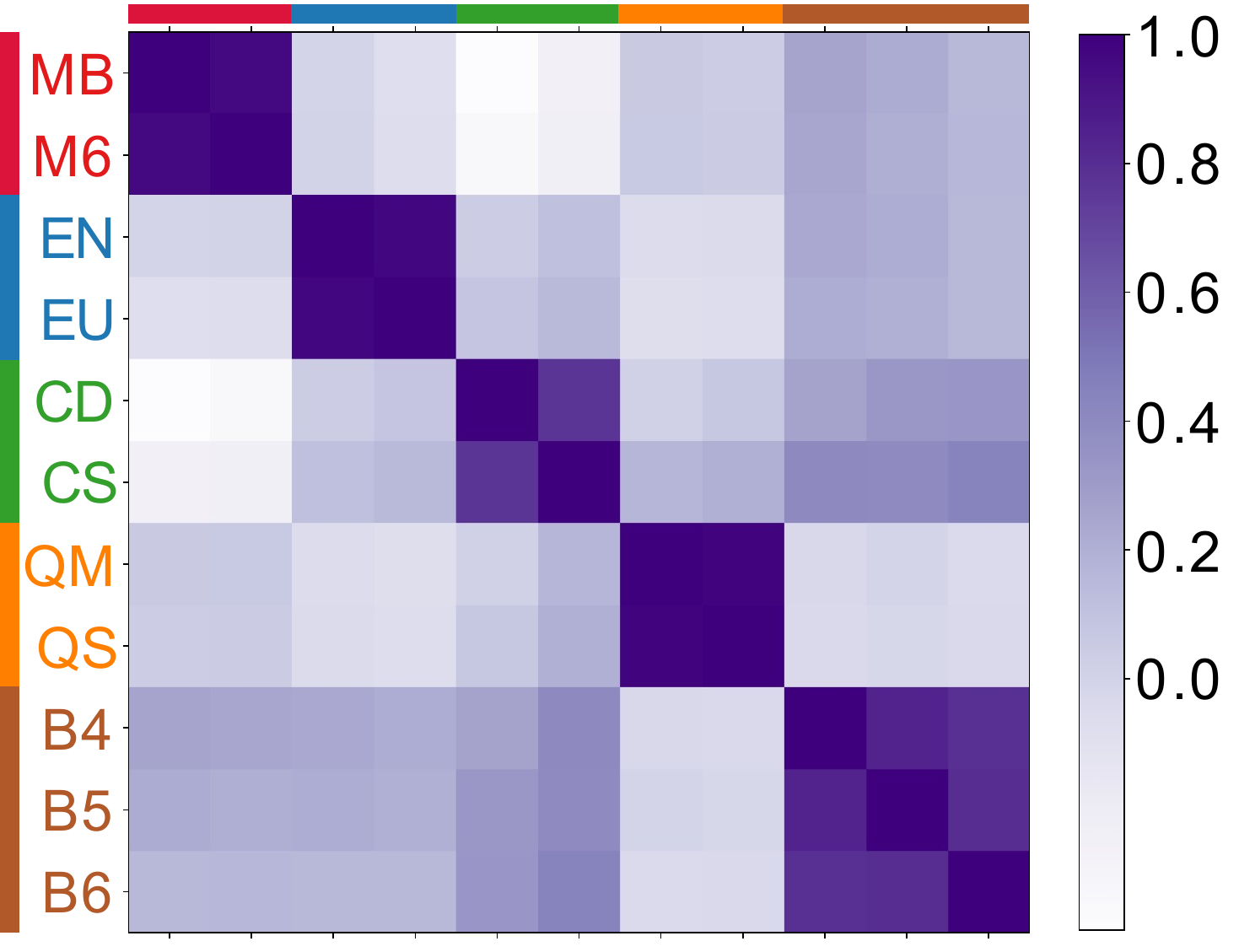}
     }
     \subfigure[\label{fig:CP-hypergraph} h-motifs]{%
         \centering
         \includegraphics[width=0.17\textwidth]{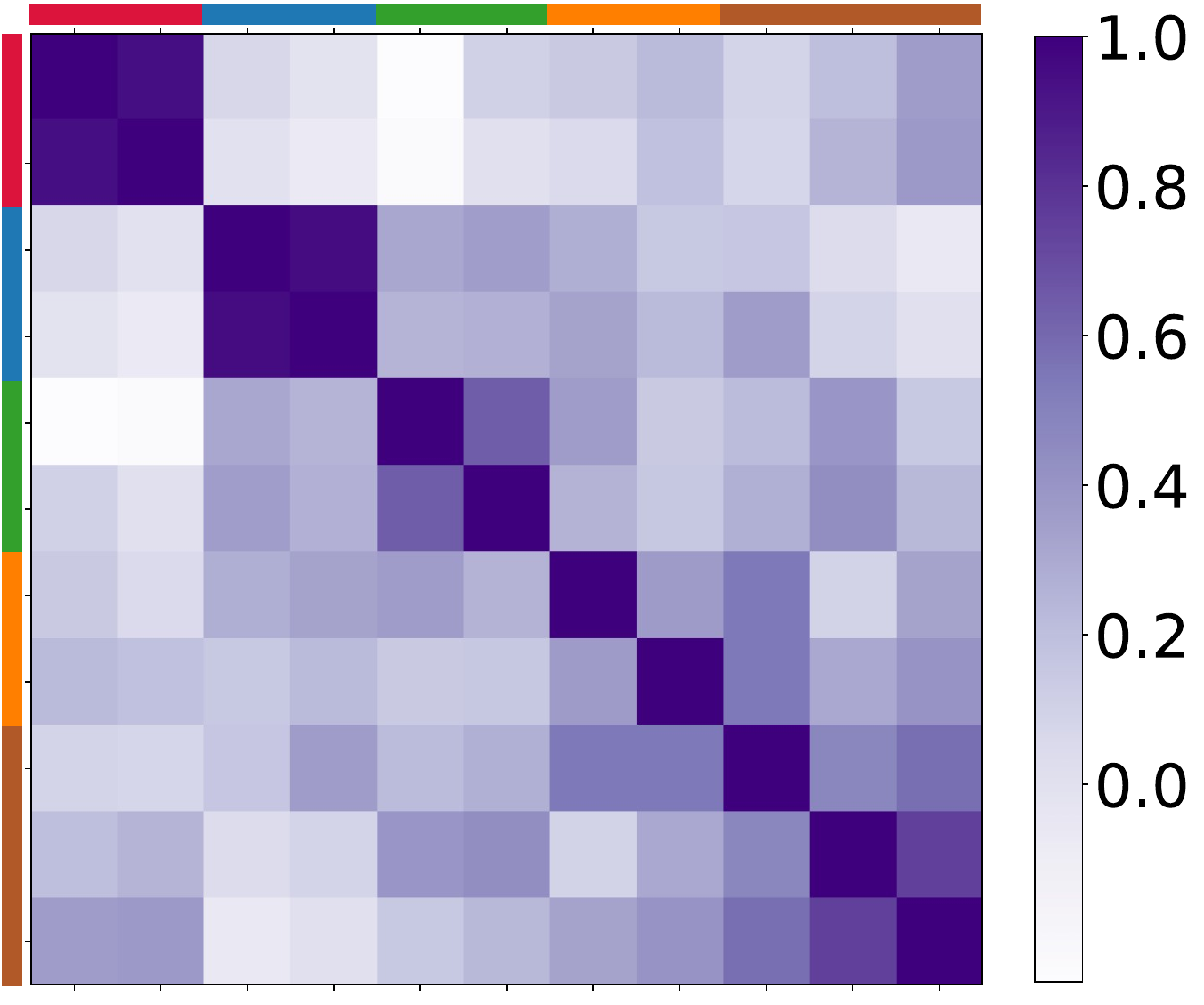}
     }
     \subfigure[\label{fig:CP-3-hypergraph} 3h-motifs]{%
         \centering
         \includegraphics[width=0.17\textwidth]{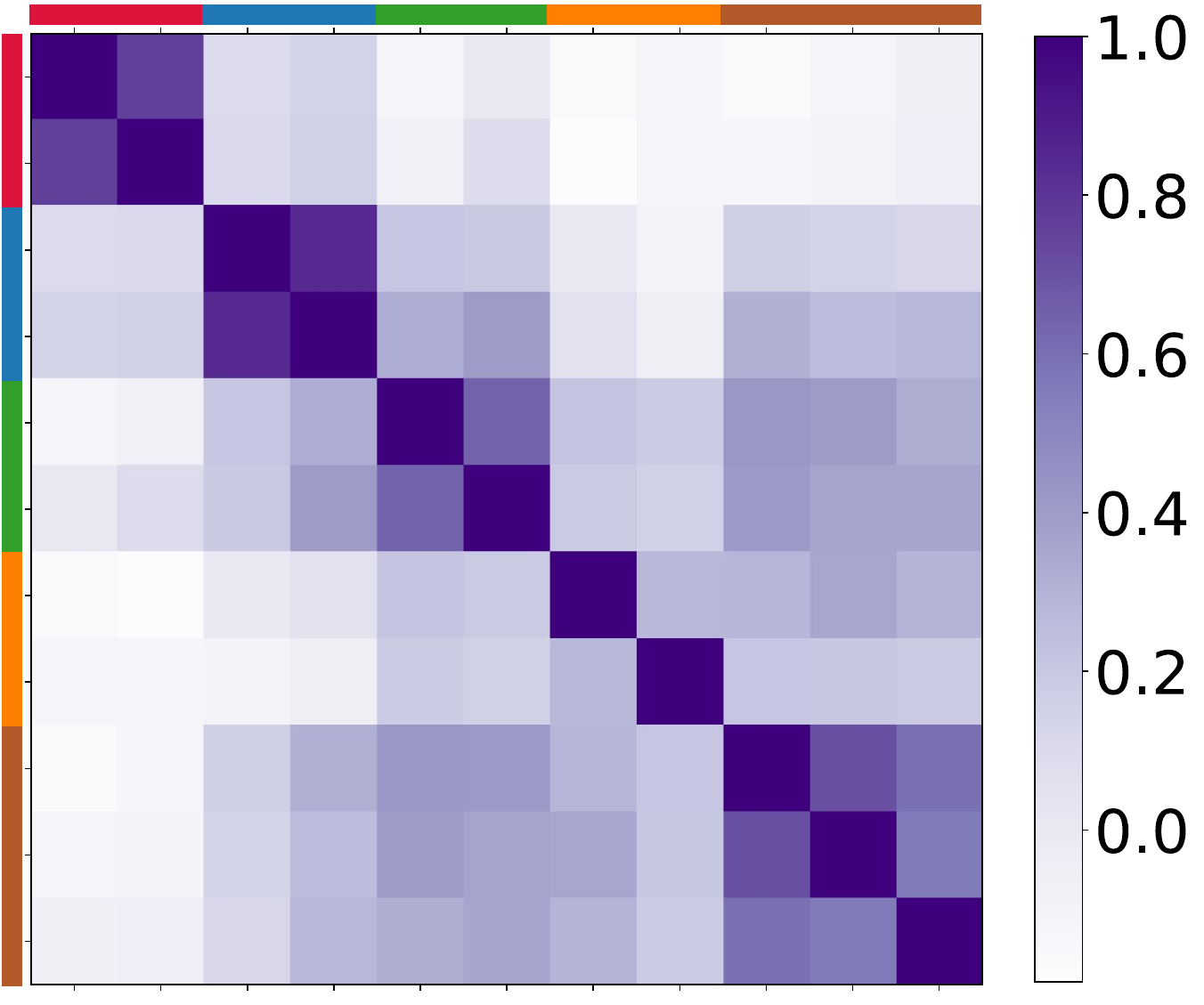}
     }
     \subfigure[\label{fig:CP-di-clique} triads]{%
         \centering
         \includegraphics[width=0.17\textwidth]{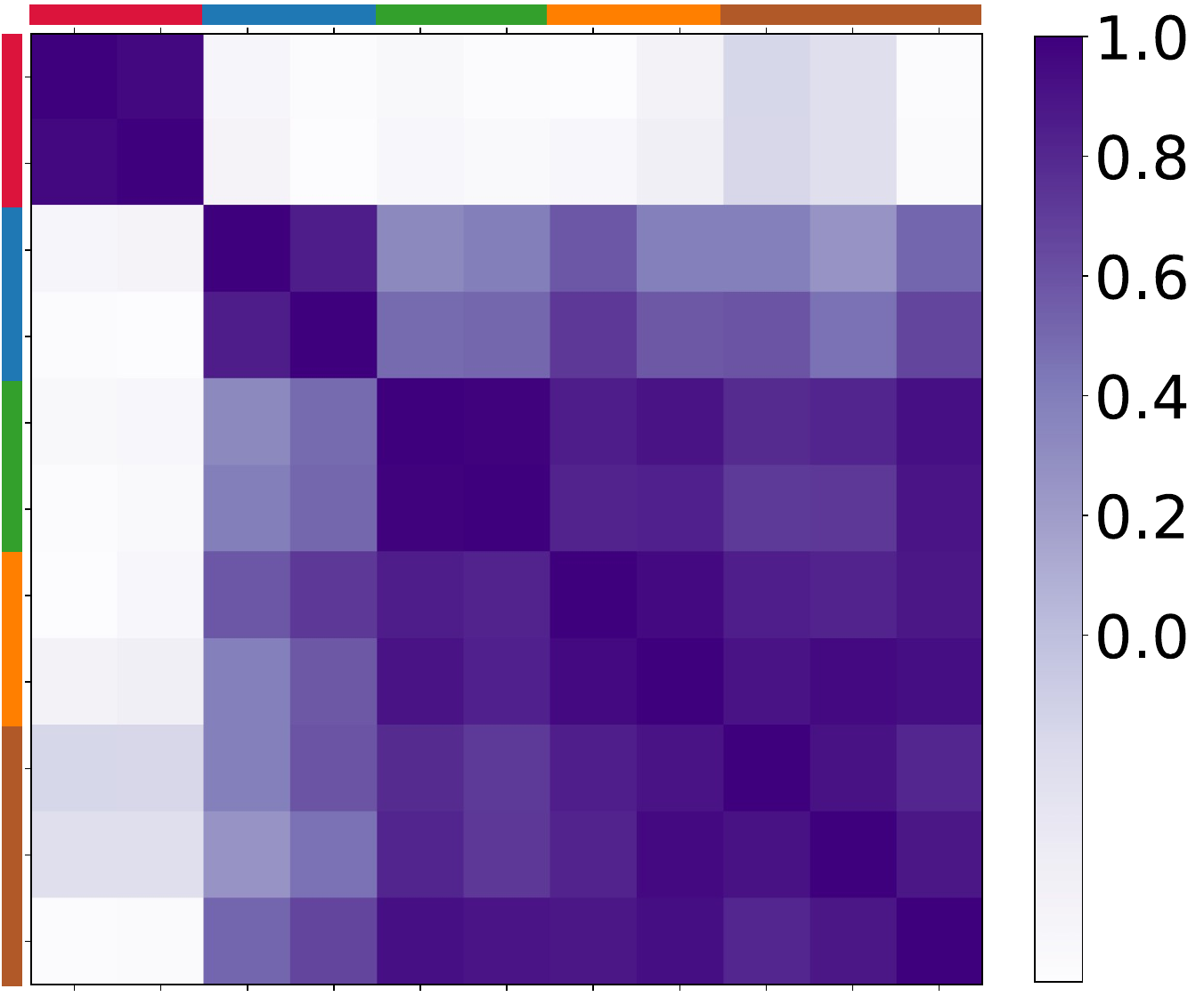}
     }
     \subfigure[\label{fig:CP-clustering} Clustering performance]{%
         \centering
         \includegraphics[width=0.18\textwidth]{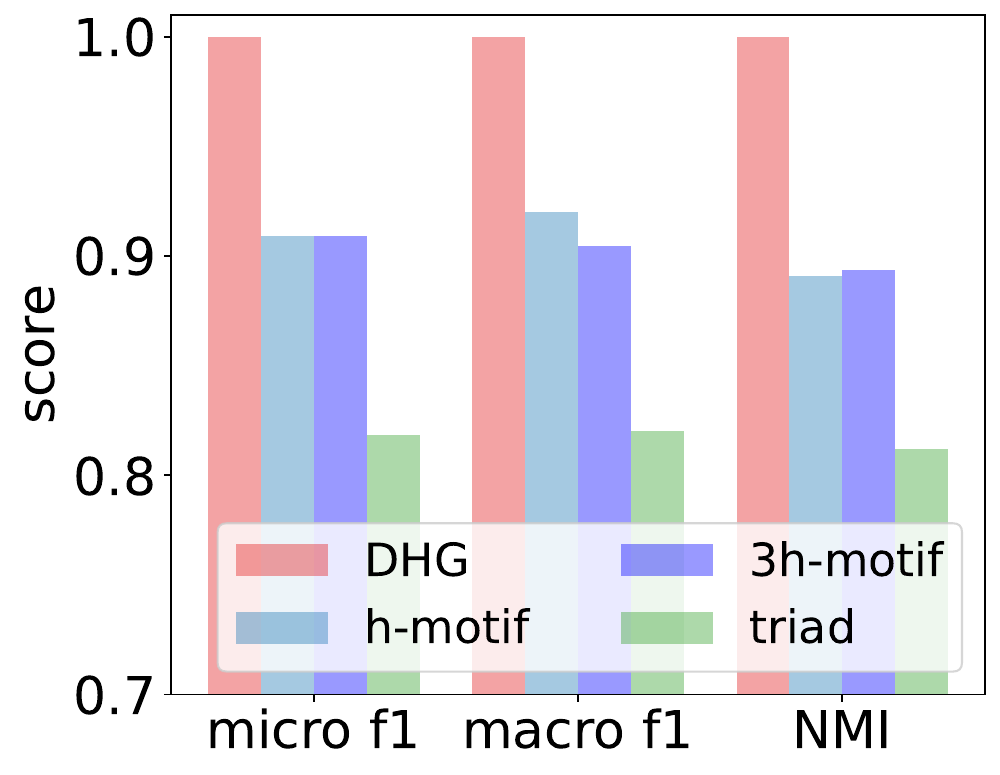}
     } \\
     \vspace{-2mm}
        \caption{\label{fig:CPs} Characteristic Profile (CP) similarity matrices and clustering performances when (a) \ours, (b) h-motifs, (c) 3h-motifs, or (d) triads are used.
        The CPs based on \ours more distinctly differentiate the datasets by domains and, in (e), lead to better clustering performance than the competitors.}
\end{figure*}

\smallsection{Comparison with other methods}
To demonstrate the effectiveness of \ours, we compare CPs based on it with CPs based on h-motifs, 3h-motifs, and triads (see in Section~\ref{sec:settings} for details).
\footnote{We count h-motifs and 3h-motifs using \mochyadv and traids using the exact subquadratic algorithm \cite{batagelj2001subquadratic}, respectively. We use \mochyadv since the exact algorithm, \mochyex, runs out of time in most data.}
Using the CPs, we compute the similarity matrix and perform clustering, as described above. Figures~\ref{fig:CP-DHG}-\ref{fig:CP-di-clique} show that using \ours leads to better domain-based differentiation than the others.
Moreover, Figure~\ref{fig:CP-clustering} shows that using the similarity matrix based on \ours results in up to $12\%$ better clustering performance (in terms of NMI) than the best competitor, respectively. 
These results imply that \ours precisely capture local structural patterns of real-world DHs.

\begin{table*}[t]
    \vspace{-3mm}
  \setlength\tabcolsep{1.8pt}
  \centering
  \caption{\label{tab:results_one_domain}
Hyperarc prediction performance (AUROC). We compare nine hyperarc feature vectors using seven classifiers. The best performances are highlighted in bold and the second-best performances are underlined.
Notably, using \our vectors leads to the best (up to $33\%$ better) AUROC in most settings, indicating that \ours extract informative hyperarc features.
See Appendix H \cite{appendix} for full results with standard deviations.
  }
\vspace{-2mm}
  \scalebox{0.90}{%
\begin{tabular}{l|cc|cc|cccc|c|cc|cc|cccc|c}
\toprule
Dataset & \multicolumn{9}{c|}{\textt{metabolic-iJO1366 (M6)}} & \multicolumn{9}{c}{\textt{email-eu (EU)}} \\
\midrule
Dimension & \multicolumn{2}{c|}{13} & \multicolumn{2}{c|}{26} & \multicolumn{4}{c|}{91} & 431 & \multicolumn{2}{c|}{13} & \multicolumn{2}{c|}{26} & \multicolumn{4}{c|}{91} & 431 \\
\midrule
Model & DHG-13 & triad & DHG-26 & h-motif & ~DHG~ & ~n2v~ & ~h2v~ & deep-h & 3h-motif & DHG-13 & triad & DHG-26 & h-motif & ~DHG~ & ~n2v~ & ~h2v~ & deep-h & 3h-motif  \\
\midrule
RF & \underline{0.824} & 0.730 & 0.817 & 0.794 & \textbf{0.834} & 0.541 & 0.546 & 0.803 & 0.820
& 0.941 & 0.901 & \underline{0.946} & 0.921 & \textbf{0.960} & 0.737 & 0.529 & 0.770 & 0.925 \\
MLP & \underline{0.782} & 0.646 & 0.761 & 0.705 & 0.687 & 0.548 & 0.559 & \textbf{0.804} & 0.682
& 0.960 & 0.911 & \textbf{0.963} & 0.909 & \underline{0.962} & 0.790 & 0.539 & 0.802 & 0.887 \\
XGB & 0.814 & 0.710 & 0.813 & 0.765 & \textbf{0.819} & 0.530 & 0.537 & 0.792 & \underline{0.818}
& 0.945 & 0.905 & \underline{0.949} & 0.917 & \textbf{0.961} & 0.747 & 0.557 & 0.777 & 0.925 \\
LGBM & \underline{0.818} & 0.687 & 0.813 & 0.747 & \textbf{0.822} & 0.515 & 0.541 & 0.794 & 0.812
& 0.948 & 0.908 & \underline{0.952} & 0.923 & \textbf{0.963} & 0.761 & 0.555 & 0.796 & 0.930 \\
HGNN & 0.571 & 0.552 & 0.587 & 0.551 & \textbf{0.601} & 0.509 & 0.526 & 0.534 & \underline{0.596}
& \underline{0.524} & 0.516 & 0.524 & 0.520 & \textbf{0.529} & 0.504 & 0.502 & 0.500 & 0.524 \\
FHGCN & 0.640 & 0.615 & 0.647 & 0.553 & \textbf{0.720} & 0.506 & 0.484 & 0.566 & \underline{0.692}
& 0.724 & \textbf{0.888} & 0.745 & 0.724 & \underline{0.849} & 0.520 & 0.550 & 0.614 & 0.768 \\
UGCN\uppercase\expandafter{\romannumeral2} & \textbf{0.713} & 0.647 & 0.692 & 0.663 & 0.663 & 0.502 & 0.503 & 0.642 & \underline{0.704}
& 0.864 & \textbf{0.912} & \underline{0.881} & 0.805 & 0.874 & 0.793 & 0.812 & 0.784 & 0.849 \\
\midrule
\rowcolor{Gray} Max & \underline{0.824} & 0.730 & 0.817 & 0.794 & \textbf{0.834} & 0.548 & 0.559 & 0.804 & 0.820
& 0.960 & 0.912 & \underline{\textbf{0.963}} & 0.923 & \underline{\textbf{0.963}} & 0.793 & 0.812 & 0.802 & 0.930 \\
\rowcolor{Gray} Avg. & \textbf{0.738} & 0.655 & 0.733 & 0.683 & \underline{0.735} & 0.522 & 0.528 & 0.705 & 0.732
& 0.843 & 0.849 & \underline{0.852} & 0.817 & \textbf{0.871} & 0.693 & 0.578 & 0.720 & 0.830 \\
\rowcolor{Gray} Rank Avg. & \underline{2.571} & 6.286 & 3.286 & 5.714 & \textbf{2.000} & 8.857 & 8.143 & 5.143 & 3.000
& 3.429 & 4.286 & \underline{2.429} & 5.429 & \textbf{1.571} & 8.000 & 8.286 & 7.571 & 4.000 \\
\midrule
\midrule
Dataset & \multicolumn{9}{c|}{\textt{citation-software (CS)}}  & \multicolumn{9}{c}{\textt{bitcoin-2016 (B6)}} \\
\midrule
Dimension & \multicolumn{2}{c|}{13} & \multicolumn{2}{c|}{26} & \multicolumn{4}{c|}{91} & 431 & \multicolumn{2}{c|}{13} & \multicolumn{2}{c|}{26} & \multicolumn{4}{c|}{91} & 431 \\
\midrule
Model & DHG-13 & triad & DHG-26 & h-motif & ~DHG~ & ~n2v~ & ~h2v~ & deep-h & 3h-motif & DHG-13 & triad & DHG-26 & h-motif & ~DHG~ & ~n2v~ & ~h2v~ & deep-h & 3h-motif  \\
\midrule
RF & 0.993 & 0.777 & \underline{0.994} & 0.945 & \textbf{0.999} & 0.611 & 0.502 & 0.739 & 0.986
& 0.971 & 0.766 & \underline{0.974} & \multirow{10}{*}{O.O.T.*} & \textbf{0.978} & 0.550 & 0.529 & 0.725 & \multirow{10}{*}{O.O.T.*} \\
MLP & \underline{\textbf{0.997}} & 0.775 & 0.997 & 0.930 & 0.996 & 0.671 & 0.544 & 0.806 & 0.964
& 0.909 & 0.655 & \underline{0.915} &  & \textbf{0.968} & 0.600 & 0.587 & 0.770 &  \\
XGB & 0.995 & 0.781 & \underline{0.996} & 0.939 & \textbf{0.999} & 0.620 & 0.522 & 0.786 & 0.991
& 0.964 & 0.801 & \underline{0.971} &  & \textbf{0.982} & 0.562 & 0.550 & 0.783 &  \\
LGBM & 0.995 & 0.792 & \underline{0.996} & 0.946 & \textbf{0.999} & 0.626 & 0.525 & 0.807 & 0.991
& 0.951 & 0.814 & \underline{0.963} &  & \textbf{0.982} & 0.563 & 0.558 & 0.788 &  \\
HGNN & 0.584 & 0.510 & \underline{0.589} & 0.509 & 0.587 & 0.562 & \textbf{0.590} & 0.508 & 0.551
& \underline{0.833} & 0.659 & \textbf{0.834} &  & 0.827 & 0.659 & 0.663 & 0.596 &  \\
FHGCN & 0.651 & 0.714 & 0.659 & 0.718 & \textbf{0.852} & 0.513 & 0.511 & 0.505 & \underline{0.732}
& 0.598 & \underline{0.712} & 0.597 &  & \textbf{0.825} & 0.508 & 0.507 & 0.505 &  \\
UGCN\uppercase\expandafter{\romannumeral2} & \underline{0.972} & 0.792 & \textbf{0.973} & 0.812 & 0.971 & 0.899 & 0.895 & 0.609 & 0.868
& \underline{\textbf{0.937}} & 0.847 & 0.937 &  & 0.931 & 0.856 & 0.801 & 0.704 &  \\
\midrule
\rowcolor{Gray} Max & \underline{0.997} & 0.792 & 0.997 & 0.946 & \textbf{0.999} & 0.899 & 0.895 & 0.807 & 0.991
& 0.971 & 0.847 & \underline{0.974} &  & \textbf{0.982} & 0.856 & 0.801 & 0.788 &  \\
\rowcolor{Gray} Avg. & 0.884 & 0.734 & \underline{0.886} & 0.828 & \textbf{0.914} & 0.643 & 0.584 & 0.680 & 0.869
& 0.880 & 0.751 & \underline{0.884} &  & \textbf{0.928} & 0.614 & 0.599 & 0.696 &  \\
\rowcolor{Gray} Rank Avg. & 3.286 & 6.571 & \underline{2.143} & 5.429 & \textbf{1.857} & 6.857 & 7.143 & 7.429 & 4.286
& 2.571 & 4.286 & \underline{2.143} &  & \textbf{1.571} & 5.429 & 6.286 & 5.714 &  \\
\bottomrule
\multicolumn{10}{l}{* O.O.T: out-of-time ($>1$ day).}
\end{tabular}
}%
\end{table*}

\subsection{Q2. Hyperarc Characterization Power}\label{sec:exp_apply}

We evaluate the effectiveness of \ours in characterizing hyperarcs by applying them to hyperarc prediction tasks.

\smallsection{Problem settings} \label{sec:hypergraph-hyperarc-prediction} 
We formulate hyperarc prediction as a classification task as in \cite{lee2020hypergraph,patil2020negative}.
Similar to constructing the randomized directed hypergraphs (DHs), we generate fake synthetic hyperarcs, up to ten percent of the total number of original hyperarcs, and add them to the original DHs.
After splitting the fake hyperarcs into the train and test sets with an 8:2 ratio, we uniformly sample the same number of real ones as the fake ones in the train and test sets, respectively.
The model is trained to classify real and fake hyperarcs with the feature vectors obtained from the contaminated DHs.

\smallsection{Feature vectors} \label{sec:hypergraph-feature-vectors}
To the best of our knowledge, \ours are the only method specifically designed for extracting feature vectors from DHs. Thus, we adapt alternative methods designed for different purposes for comparison with \ours. %
Specifically, we compare the following hyperarc feature vectors:
\begin{itemize}[leftmargin=*]
    \item \textbf{\our $\in \mathbb{R}^{91}$:}
    The exact count of the instances of each \our that each hyperarc is involved in.
    \item \textbf{h-motif $\in \mathbb{R}^{26}$} and \textbf{3h-motif $\in \mathbb{R}^{431}$:}
    The exact count of the instances of each h/3h-motif that each hyperarc is involved in after hypergraph expansion (see Section~\ref{sec:hypergraph-competitors}).
    \item \textbf{triad $\in \mathbb{R}^{13}$:}
    The exact count of the instances of each triad census that each hyperarc is involved in after di-biclique expansion (see Section~\ref{sec:hypergraph-competitors}).
    \item \textbf{DHG-26 $\in \mathbb{R}^{26}$} and \textbf{DHG-13 $\in \mathbb{R}^{13}$:}  \our after dimension reduction using PCA\footnote{https://scikit-learn.org/stable/modules/generated/sklearn.decomposition.PCA} to match the dimensions of h-motif and triad, respectively.
    \item \textbf{node2vec (n2v) $\in \mathbb{R}^{91}$:}
    The node2vec \cite{grover2016node2vec} embedding of each hyperarc after di-biclique expansion (see Section~\ref{sec:hypergraph-competitors}). %
    \item \textbf{hyper2vec (h2v) \& deep hyperedges (deep-h) $\in \mathbb{R}^{91}$:} The hyper2vec \cite{huang2019hyper2vec} and deep hyperedges \cite{payne2019deep} embeddings of each hyperarc after hypergraph expansion (see Section~\ref{sec:hypergraph-competitors}).
\end{itemize}
Note that, for the last two, we generate 91-dimensional vectors to match the dimension of CPs based on \ours.

\smallsection{Evaluation protocol}
We train ten classifiers (Logistic Regression, Random Forest, Decision Tree, K-Nearest Neighbor, Multi-Layer Perceptron \cite{scikit-learn}, eXtreme Gradient Boosting \cite{chen2016xgboost}, Light Gradient-Boosting Machine \cite{ke2017lightgbm}, HGNN \cite{feng2019hypergraph}, FastHyperGCN \cite{NEURIPS2019_1efa39bc}, and UniGCN\uppercase\expandafter{\romannumeral2} \cite{ijcai21-UniGNN}) using each of the feature vectors above.
We assess the performance using the area under the ROC curve (AUROC) and accuracy.
We generate ten different contaminated DHs from each dataset, and train and test the model ten times on each generated DH. See Appendix G \cite{appendix} for details, including hyperparameter settings.

\smallsection{Results}
Table~\ref{tab:results_one_domain} presents the AUROC results (average over 100 trials) on four datasets using the seven best classifiers. Notably, the use of \our vectors leads to the best results in most datasets with most classifiers, achieving up to  $33\%$ better AUROC on the \textt{B6} dataset than other features. 
Even with reduced dimensions, \our vectors consistently outperform the competitors with the same dimension.
These results indicate that \ours extract informative features of hyperarcs.
For full results on all datasets with standard deviations and all classifiers, see Appendix H \cite{appendix}.

\subsection{Q3. Performance of Counting Algorithms}\label{sec:exp:q3}

We evaluate the speed and accuracy of \our counting algorithms.
In addition to the proposed ones (i.e., \exact and \adv), we consider two approximate competitors: \naive (Section~\ref{sec:approx}) and \ata.\footnote{\ata samples a single hyperarc $e$ from $E$ uniformly at random first, and then another incident hyperarc $e'\in N_{e}$ uniformly at random. We present its detailed procedure, unbiasedness, variance, and complexity in Appendix A. }
We vary the ratio $q$ of the number of samples to the number of hyperarcs in $\{0.2, 1, 5, 10, 50\}$ when $|V|<10^4$ (\textt{EN}, \textt{EU}, \textt{MB}, and \textt{M6}) and in $\{0.2, 0.4, 0.6, 0.8, 1.0\}$ otherwise (\textt{CD}, \textt{CS}, \textt{QM}, \textt{QS}, \textt{B4}, \textt{B5}, and \textt{B6}). 
For each algorithm and each $q$ value, we report the running time and accuracy averaged over $20$ independent trials.
We employ two accuracy measures, as described below.

\smallsection{Measure 1. error measure for \our counts}
As a performance measure, we use $err^G$, which measures the difference between the exact and estimated counts of \our instances as:
\begin{align*}
 err^G:=\frac{\sum_{i \in [m]} |C^G[i]-|\Omega^G_i||}{\sum_{i \in [m]}|\Omega^G_i|},
\end{align*}
where $C^G[i]$ is an estimate of $|\Omega_i^G|$ obtained by one among \naive, \ata and \adv. 
Note that a lower $err^G$ value indicates better performance.

As shown in Figure~\ref{app:fig:error}, the $err^G$ values of \adv, which is our proposed approximate algorithm, are significantly lower than those of \naive and \ata on all datasets, when their running time is similar. 
When the $err^G$ values are similar, \adv is up to $40\times$ faster than \naive and \ata.

\smallsection{Measure 2. CP similarity}
We also use $cos^G$, defined as the cosine similarity between the characteristic profiles (CPs) based on exact and estimated occurrences of \ours.
The estimates are obtained by one among \naive, \ata, and \adv.
A higher $cos^G$ value indicates better performance.

Figure~\ref{fig:cosine} shows that %
\adv is up to $32\times$ faster than \naive and \ata on all datasets when their $cos^G$ values are similar. When their running times are similar, \adv gives the best $cos^G$ value in almost all cases.

\begin{figure}[t!]
    \centering
    \vspace{-1mm}
    \includegraphics[width=\linewidth]{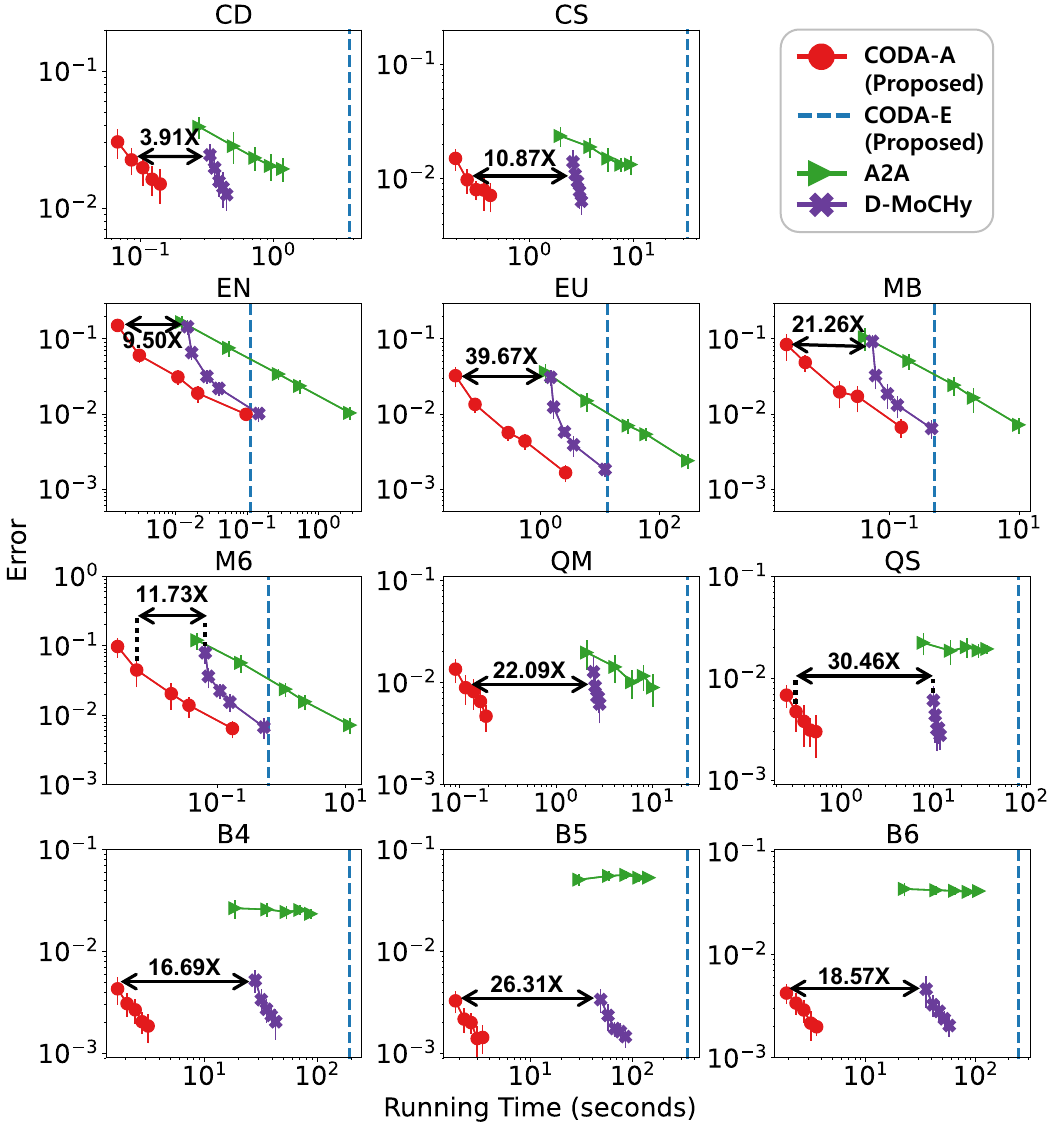}
    \\
     \vspace{-2mm}
     \caption{\label{app:fig:error} Estimation error $err^{G}$ (a lower value indicates better performance) and running time of each algorithm on all datasets. Note that our proposed approximate algorithm, \adv, consistently
     gives the best trade-off among the approximate ones. 
     Error bars indicate $\pm 1$ standard deviation.
     For \exact, which is an exact algorithm, $err^{G}$ is always 0.
     }
\end{figure}

\begin{figure}[t!]
    \centering
    \vspace{-1mm}
    \includegraphics[width=\linewidth]{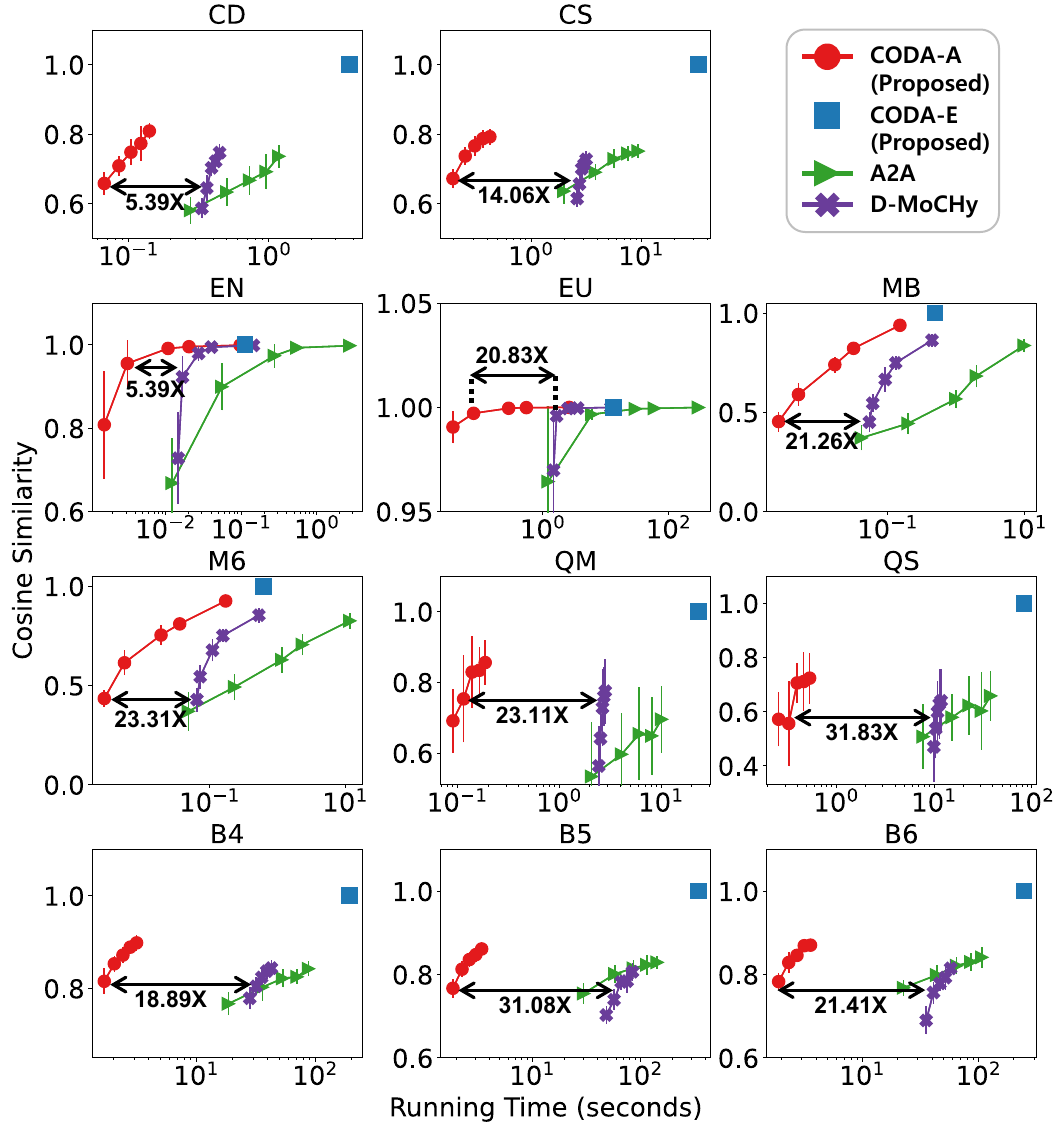} \\
     \vspace{-2mm}
     \caption{\label{fig:cosine} CP similarity $cos^G$ (a higher value indicates better performance) and running time of each algorithm on all datasets. Note that our proposed approximate algorithm, \adv, consistently
     gives the best trade-off among the approximate ones.
     Error bars indicate $\pm 1$ standard deviation.
     } %
\end{figure}

\begin{figure}[t!]
    \centering
    \vspace{-1mm}
    \includegraphics[width=\linewidth]{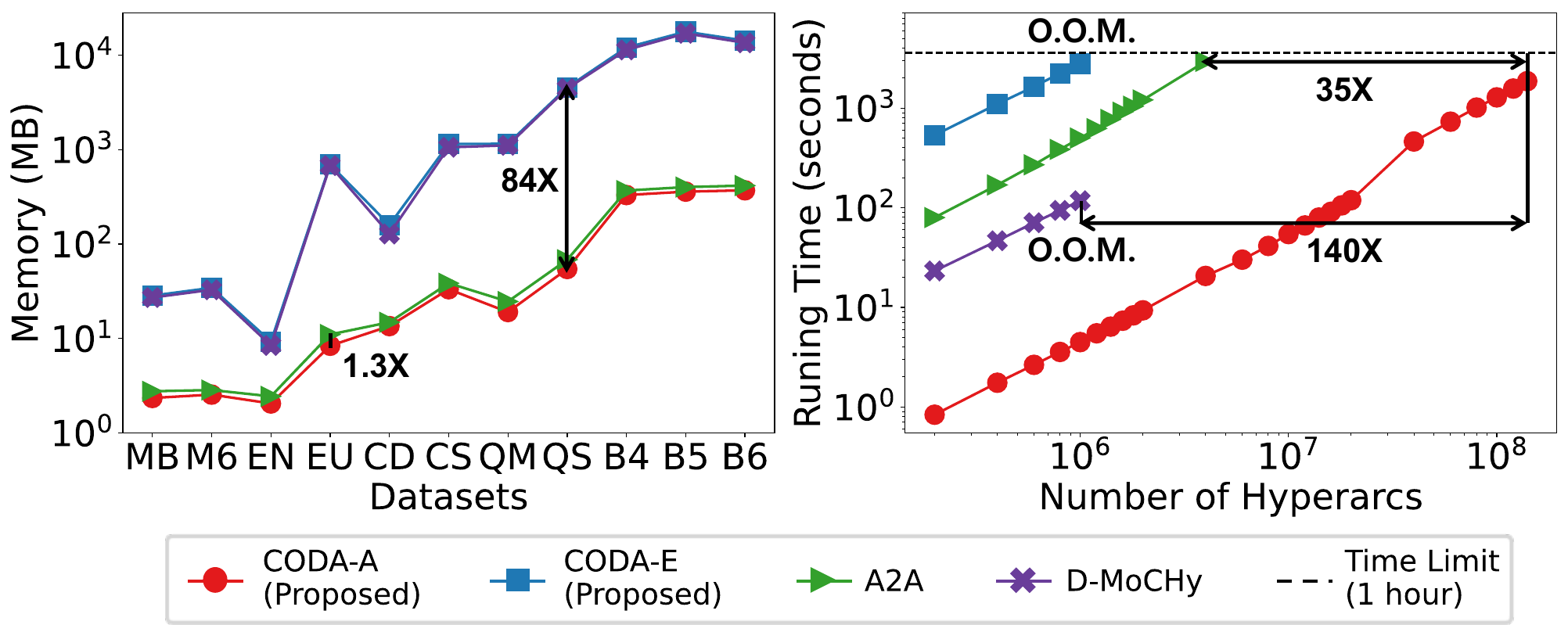}
    \\
     \caption{\label{app:fig:memory_and_scalability} (Left) Memory requirements of each algorithm on eleven datasets. Notably, \adv and \ata require significantly less memory compared to \naive and \exact. (Right) Running times of each algorithm on large synthetic datasets. While \naive and \exact runs out of memory on DHs with more than $10^6$ hyperarcs, \adv is capable of handling DHs that are up to $140\times$ larger than those processed by \naive and \exact.}
\end{figure}

\smallsection{Memory requirements}
The left side of Figure~\ref{app:fig:memory_and_scalability} shows the average memory requirement over ten trials for all \our counting algorithms across eleven datasets. Importantly, these memory requirements are independent of the number of samples.
\adv requires significantly less memory than \naive and \exact.
This result aligns with our theoretical findings that \naive and \exact require $O(|\Omega|)$ additional space compared to \adv (Propositions~\ref{thm:complexity_exact}, \ref{prop:naive:comp}, and \ref{prop:adv:comp}).
\ata, despite its weaker speed-accuracy trade-off, requires a similar amount of memory with \adv. Their space complexities are the same (see Proposition~\ref{prop:a2a} in Appendix~\ref{app:ata}).

\smallsection{Scalability}
We evaluate the scalability of \adv using synthetic datasets generated by a hypergraph generator.
Refer to Algorithm \ref{algo:naive_random} in Appendix~\ref{app:generate} for details.
We set a maximum hyperarc size $k$ to $40$ and a ratio of the number of hyperarcs to the number of nodes $r$ to $20$. 
We vary the number of nodes $N$ from $10^4$ to $10^7$.
For sampling, we fix the ratio $q$ of the number of samples to the number of hyperarcs to $1.0$.
Using created synthetic DHs with varying numbers of hyperarcs, we measure the average running time over ten trials with a $1$ hour time limit.
The right side of Figure~\ref{app:fig:memory_and_scalability} shows that \naive and \exact run out of memory when dealing with over $10^6$ hyperarcs, whereas \adv exhibits remarkable scalability, capable of handling $140\times$ more hyperarcs than both methods.

In summary, \adv, which is our proposed approximate algorithm,  gives the best trade-off between time and accuracy consistently among the approximate algorithms, being consistent with our time-complexity analysis (Propositions~\ref{prop:naive:comp} and \ref{prop:adv:comp}). Moreover, \adv requires less memory, enabling its application to large datasets without out-of-memory issues.

\begin{figure}[t!]
    \centering
    \includegraphics[width=\linewidth]{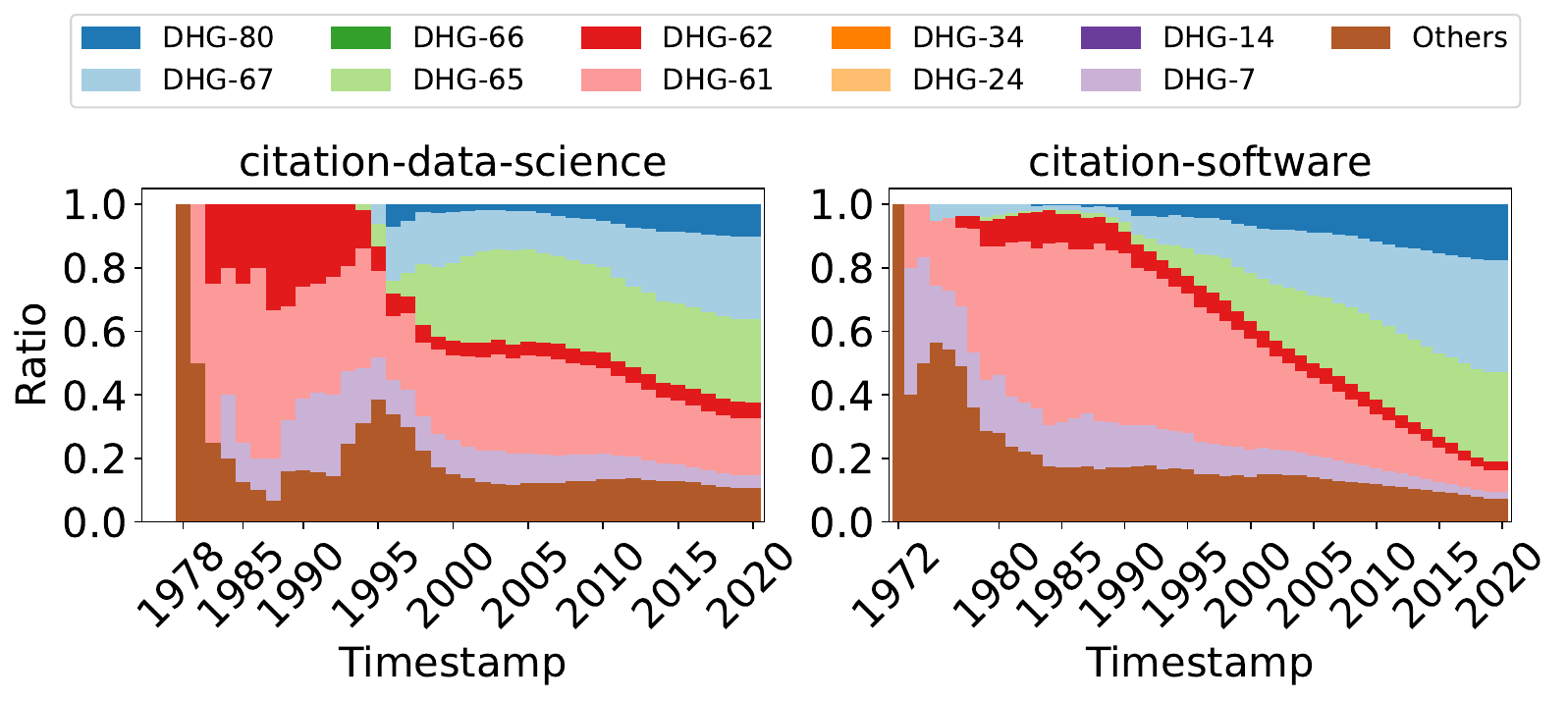}
    \vspace{1mm}
    \includegraphics[width=\linewidth]{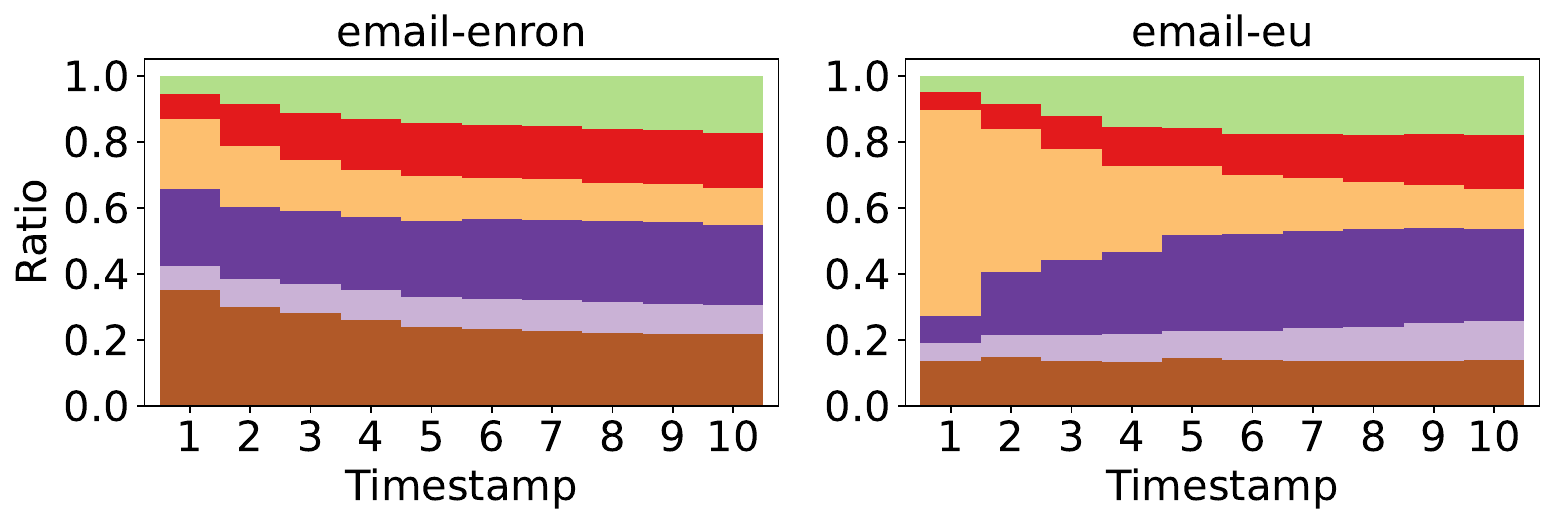}
    \vspace{1mm}
    \includegraphics[width=\linewidth]{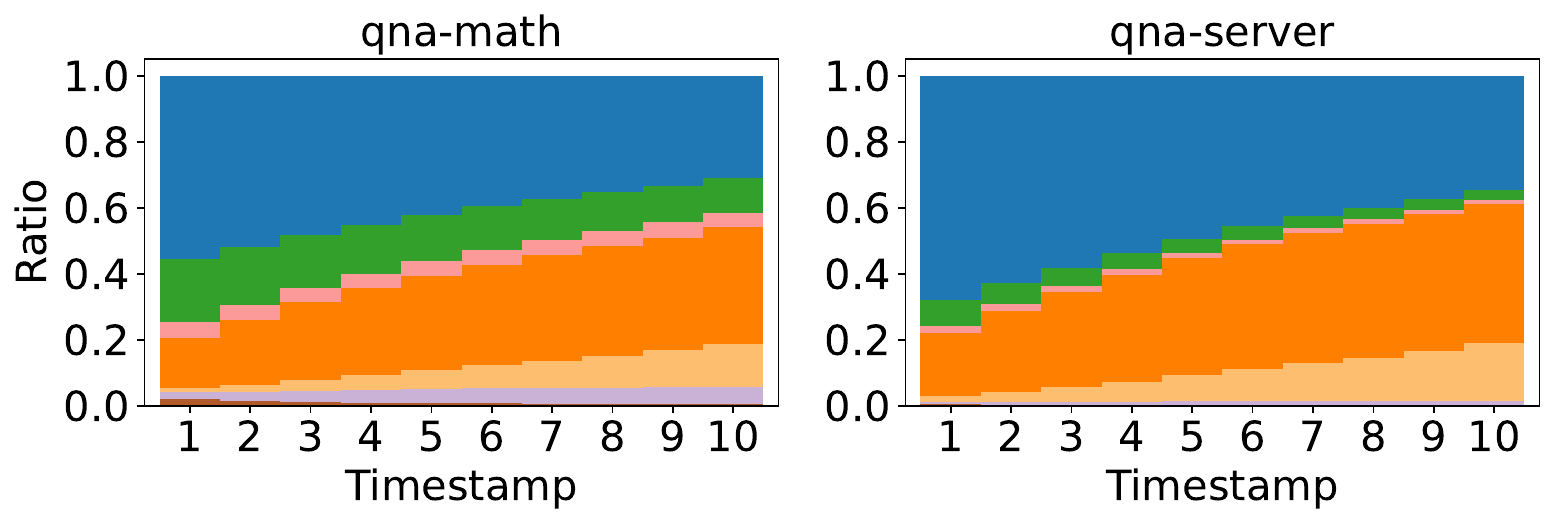}
    \vspace{1mm}
    \includegraphics[width=\linewidth]{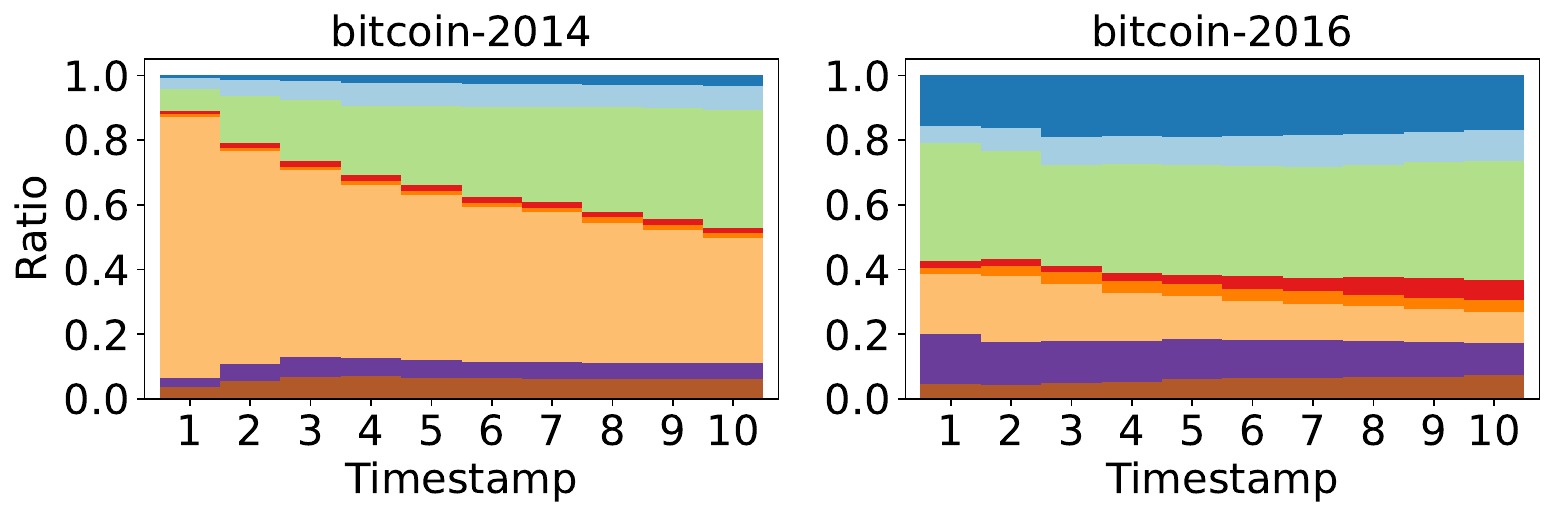}
    \vspace{-5mm}
    \caption{
    Directed hypergraphs from the same domain not only share the same set of frequent \ours but also exhibit similar time-evolving tendencies. \ours
with ratios below certain thresholds (spec., $0.03$ for the \textt{citation} datasets,  $0.1$ for the \textt{email} datasets, $0.01$ for the \textt{qna} datasets, and 0.03 for the \textt{bitcoin} datasets) are grouped as \textt{Others}.
    }
    \label{fig:time_all}
\end{figure}

\begin{figure}[t!]
    \vspace{-2mm}
    \centering
    \includegraphics[width=\linewidth]{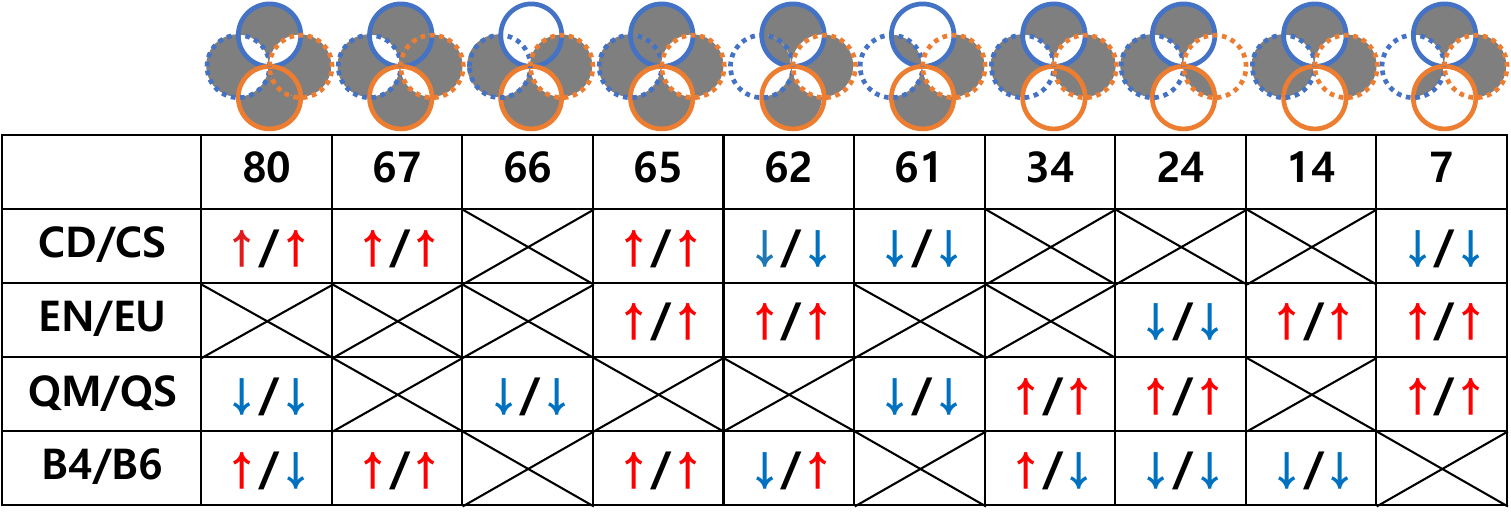}
    \vspace{-4mm}
     \caption{Time-evolving trend.
     Each row represents a domain of datasets and each column represents a \our. %
     \ours marked with  `X' have ratios less than a predefined threshold.
     The ratios of \ours marked with \red{$\uparrow$} or \blue{$\downarrow$} tend to increase or decrease, respectively. 
     }
     \label{fig:time_tendency_all}
\end{figure}

\subsection{Q4. Temporal Analysis}\label{sec:exp_discover}

We analyze the time-evolving patterns in real-world directed hypergraphs (DHs) using \ours.
Specifically, we examine snapshots of the same interval for each time-evolving DH\footnote{The \textt{metabolic} datasets do not include timestamps.} %
and track the occurrence ratio of each \our in each snapshot. Refer to Appendix F \cite{appendix} for deatils of the snapshots.

In Figure~\ref{fig:time_all}, we visualize \ours whose ratio is greater than specific thresholds while aggregating the rest as \textt{Others}.
The threshold values for the \textt{citation}, \textt{email}, \textt{qna}, and \textt{bitcoin} datasets are 0.03, 0.1, 0.01, and 0.03, respectively.
In addition, we summarize the time-evolving patterns of the top ten most frequent \ours in Figure~\ref{fig:time_tendency_all}.
Notably, datasets from the same domain exhibit not only similar sets of frequent \ours but also similar time-evolving patterns for each \our.

\smallsection{Citation datasets}
The ratios of \our-65, -67, and -80 increase, while those of \our-7, -61, and -62 decrease.
In \our-65, -67, and -80, all non-intersecting regions (i.e., regions 1, 4, 5, and 8 in Figure~1(c)) contain at least one node, while in \our-7, -61, and -62, some non-intersecting regions are empty.
That is, the number of non-empty non-intersecting regions increases over time.

\smallsection{Email datasets}
Among the frequent \ours, only the ratio of \our-24 decreases, whereas the ratios of \our-7, -14, -62, and -65 increase.
\our-24 is distinct from the other \ours in that two tail sets intersect in it.
Given that the size of tail sets in email datasets is always 1, the probability of two hyperarcs sharing the same tail set decreases over time, resulting in the decrease of the \our-24 ratio.

\smallsection{Qna datasets}
There is a dramatic decrease in the ratio of \our-80 and a dramatic increase in the ratio of \our-34.
The ratios of \our-61, -66, and -80  decrease, while the ratios of \our-7, -24, and -34 increase.
The increasing \ours have empty non-intersecting areas in their tail sets, while the declining \ours have no such areas.
This indicates that the number of users who frequently answer questions increases.

\smallsection{Bitcoin datasets}
Although the two \textt{bitcoin} datasets share the same set of frequent \ours, their tendencies to increase or decrease over time differ in them.
For example, \our-65 becomes dominant in the bitcoin-2016 dataset, with a ratio of 0.37, while it has a ratio of only 0.07 in the other dataset. Conversely, \our-24 becomes dominant in the bitcoin-2014 dataset, with a ratio of 0.81, but only has a ratio of 0.10 in the other dataset. The main difference between \our-65 and \our-24 is the presence or absence of an intersection between tail sets.
\our-24 has only an intersection between tail sets, but \our-65 does not.
This indicates that the diversity of accounts participating in transactions increases over time.

\section{Conclusions}
\label{sec:conclusion}

In this work, we focus on measurements, findings, and applications related to local structures of directed hypergraphs (DHs).
Specifically, we present (a) how to characterize the local structures of DHs (b) what the building blocks of real-world DHs are, (c) how they are used for applications, and (d) how we can rapidly but accurately characterize the local structures of a large-scale DH, based on a small fraction of it.

In summary, our contributions are as follows:
\begin{itemize}[leftmargin=*]
    \item \textbf{New Concepts.} To the best of our knowledge, we are the first to define equivalent classes for describing the local structures of DHs, i.e., 91 directed hypergraphlets (\ours).
    \item \textbf{Discoveries.} 
    Using \ours, we unveil striking domain-based patterns in the local structures of $11$ real-world DHs.
    \item \textbf{Applications.} We present how \ours can be used for vector representations of DHs and hyperarcs,
    and we show that they lead to better performance than other features on hypergraph clustering (Figure~\ref{fig:CPs}) and hyperarc prediction (Table~\ref{tab:results_one_domain}).
    \item \textbf{Fast Algorithms.} We develop fast algorithms for counting \our instances and prove their theoretical properties (Propositions~\ref{prop:adv:unb}-\ref{prop:adv:con}).
    We show that \adv, our proposed approximate algorithm, achieves a remarkable balance between speed and accuracy (Figures~\ref{app:fig:error} and  \ref{fig:cosine}) while requiring little memory (Figure~\ref{app:fig:memory_and_scalability}).
\end{itemize}
For \textbf{reproducibility}, we provide the code and datasets at~\cite{appendix}.

\appendix
\label{sec:appendix_offline}

\subsection{\ata: Baseline Algorithm}\label{app:ata}

In this section, we provide a detailed description of the baseline algorithm, \ata, including its unbiasedness, variance, and complexity. Its pseudocode is presented in Algorithm~\ref{algo:ata}.
Here, $p(e, e')=\left(\frac{1}{|N_{e}|}+\frac{1}{|N_{e'}|}\right)\cdot \frac{1}{|E_{\geq 1}|}$ where $E_{\geq 1}=\{e: N_{e}\geq 1\}$. 
Assume $E_{\geq 1}$ is given at first. 
Also, for space efficiency, we assume $N_{e}$ is maintained (Line~\ref{algo:ata:nei}).
$HM(A, B)$ on Line~\ref{algo:ata:hm} denotes the harmonic mean of $A$ and $B$. %

\begin{algorithm}[t]
    \small
    \caption{\ata}\label{algo:ata}
    \SetKwInput{KwInput}{Input}
    \SetKwInput{KwOutput}{Output}
    \KwInput{(1) a directed hypergraph: $G=(V, E)$ \\
            \quad\quad\quad(2) \# of samples $n=q\cdot |E|$ for a given ratio $q$\\} %
    \KwOutput{$C[i]$ for every $i\in [m]$}
     $C[i] \leftarrow 0, \forall i \in [m]$ \\ %
    \For{$1:n$}{
        Choose $e\in E_{\geq 1}$ uniformly at random \\   
        $N_{e}\leftarrow \{e'\in E\setminus \{e\}: e\cap e'\neq \emptyset\}$\label{algo:ata:nei}\\
        $C[f(e, e')]\leftarrow C[f(e, e')]+\frac{|E_{\geq 1}|}{2\cdot n}\cdot HM(|N_{e}|, |N_{e'}|)$ \label{algo:ata:hm}\\
    }
     \Return{$C$} 
\end{algorithm}

\begin{proposition}[Unbiasedness of \ata]
Algorithm~\ref{algo:ata} is unbiased, i.e., $\mathbb{E}[C[i]]=|\Omega_i|$, $\forall i\in [m]$.
\end{proposition}
\begin{proof}
    Follow the flow of the proof of Proposition~\ref{prop:naive:unb}.
\end{proof}

\begin{proposition}[Variance of \ata]
   For each $i\in [m]$, the variance of $C[i]$ obtained by Algorithm~\ref{algo:ata} is 
  {\small\begin{align*}Var[C[i]]&=\sum_{(e, e')\in \Omega_i}\frac{1}{n}\left(\frac{1}{p(e, e')}-1\right)\\&=\sum_{(e, e')\in \Omega_i}\frac{1}{n}\left(\frac{|E_{\geq 1}|}{2}\cdot HM(|N_{e}|, |N_{e'}|)-1\right).\end{align*}} %
\end{proposition} 
\begin{proof}
    Follow the flow of the proof of Proposition~\ref{prop:naive:var}.
\end{proof}
\begin{proposition}[Time and space complexity of \ata]\label{prop:a2a}
    The time complexity of Algorithm~\ref{algo:ata} is $O(n\cdot (\max_{e\in E}|\bar{e}|\cdot \max_{e\in E}|N_{e}|))$.
    Its space complexity is $O(\sum_{e\in E}|\bar{e}|)$.
\end{proposition}
\begin{proof}
    The information of a given directed graph is stored in $O(\sum_{e\in E}|\bar{e}|)$ space at first. 
    For time complexity, $O(\max_{e\in E} |\bar{e}|\cdot \max_{e\in E}|N_{e}|)$ time is required assuming $O(p\cdot q)$ time is taken for set union when there are $p$ sets, of which size bounded by $q$.
    For space complexity, $O(|N_{e}|)\in O(\sum_{e\in E} |\bar{e}|)$ space is needed. Checking $f(e, e')$ requires $O(\max_{(e, e')\in \Omega}\min(|\bar{e}|, |\bar{e}'|))$-time, which is bounded by $O(\max_{e\in E} |\bar{e}|\cdot \max_{e\in E} |N_e|)$.
\end{proof}

\subsection{Randomization of Directed Hypergraphs (DHs)}\label{app:random}
To randomize DHs, we extend the configuration model, which is widely-used for (hyper)graphs \cite{newman2018networks,chodrow2020configuration}, to DHs. Algorithm~\ref{algo:configuration} outlines the process of obtaining a randomized DH $G'$ from the input DH $G$. It ensures that $G'$ and $G$ have the same distributions of hyperarc sizes and node degrees.
Hyperarcs are paired and the nodes in the head sets of each pair are shuffled to obtain a shuffled set $E_{\text{temp}}$ of hyperarcs (Lines~\ref{algo:conf:head1}-\ref{algo:conf:head2}).
The same process is applied to the tail sets on $E_{\text{temp}}$ (Lines~\ref{algo:conf:tail1}-\ref{algo:conf:tail2}) to obtain the final set $E'$ of hyperarcs.

\subsection{Generation of Directed Hypergraphs (DHs)}\label{app:generate}
Algorithm~\ref{algo:naive_random} describes the process of generating a random DH of a desired size. The size of each hyperedge is determined uniformly at random (Line~\ref{algo:random:size}), and it consists of randomly selected nodes (Line~\ref{algo:random:nodes}), which are then randomly divided into head and tail sets of almost equal size (Line~\ref{algo:random:split}).

\begin{algorithm}[h!]
    \small
    \caption{Randomization of Directed Hypergraphs}\label{algo:configuration}
    \SetKwInput{KwInput}{Input}
    \SetKwInput{KwOutput}{Output}
    \KwInput{a directed hypergraph: $G=(V, E)$}
    \KwOutput{a randomized directed hypergraph: $G'=(V', E')$}
    
     $V' \leftarrow V$, $n_E\leftarrow|E|$ \\
     $E_{\text{temp}}, E' \leftarrow \emptyset, \emptyset$ \\
     \tcp{Shuffle the two head sets.}
     \For{$1:\lfloor n_E/2 \rfloor$}{\label{algo:conf:head1} %
        Choose $e_1, e_2 \in \binom{E}{2}$ uniformly at random \\
        $E \leftarrow E \setminus \{e_1, e_2\}$ \\ 
        $H'_1, H'_2 \leftarrow \textsc{Shuffle}(H_1, H_2 | T_1, T_2)$ \\
        $E_{\text{temp}} \leftarrow E_{\text{temp}} \cup \{\langle T_1, H'_1 \rangle, \langle T_2, H'_2 \rangle\}$
    }
    \If{$|E|=1$}{
        $e_1 \in E$ \\
        Choose $e_2\in E_{\text{temp}}$ uniformly at random
        \\$E_{\text{temp}} \leftarrow E_{\text{temp}} \setminus \{e_2\}$ \\
        $H'_1, H'_2 \leftarrow \textsc{Shuffle}(H_1, H_2 | T_1, T_2)$ \\
        $E_{\text{temp}} \leftarrow E_{\text{temp}} \cup \{\langle T_1, H'_1 \rangle, \langle T_2, H'_2 \rangle\}$\label{algo:conf:head2}
    }
    \tcp{Shuffle the two tail sets.}
    \For{$1:\lfloor n_E/2 \rfloor$}{ \label{algo:conf:tail1}%
        Choose  $e_1, e_2 \in \binom{E_{\text{temp}}}{2}$ uniformly at random \\ $E_{\text{temp}} \leftarrow E_{\text{temp}} \setminus \{e_1, e_2\}$ \\
        $T'_1, T'_2 \leftarrow \textsc{Shuffle}(T_1, T_2 | H_1, H_2)$ \\
        $E' \leftarrow E' \cup \{\langle T'_1, H_1 \rangle, \langle T'_2, H_2 \rangle$\}
    }
    \If{$|E_{\text{temp}}|=1$}{
        $e_1 \in E_{\text{temp}}$ \\
        Choose $e_2\in E'$ uniformly at random
        \\$E_{\text{temp}} \leftarrow E_{\text{temp}} \setminus \{e_2\}$ \\
        $T'_1, T'_2 \leftarrow \textsc{Shuffle}(T_1, T_2 | H_1, H_2)$ \\
        $E' \leftarrow E' \cup \{\langle T'_1, H_1 \rangle, \langle T'_2, H_2 \rangle\}$\label{algo:conf:tail2}
    }
     \Return{$G'=(V',E')$} 

    \tcp{Shuffle the two sets of nodes}
    \SetKwProg{Fn}{Function}{}{}
    \Fn{\textsc{Shuffle}($S_1, S_2 | F_1, F_2$)}{
        $I \leftarrow S_1 \cap S_2$ \\
        $R \leftarrow (S_1 \cup S_2) \setminus I$ \\
        $R' \leftarrow R \setminus F_1 \setminus F_2$\\
        $S_{1}' \leftarrow$ Choose $(|S_1\setminus I\setminus F_2|)$ elements in $R'$ uniformly at random \\
        $S_{2}' \leftarrow R' \setminus S_{1}'$ \\
        \KwRet{$I \cup (S_1\cap F_2) \cup S_{1}',  I \cup  (S_2\cap F_1) \cup S_{2}'$}
    }
\end{algorithm}

\begin{algorithm}[h!]
    \small
    \caption{Generation of Directed Hypergraphs}\label{algo:naive_random}
    \SetKwInput{KwInput}{Input}
    \SetKwInput{KwOutput}{Output}
    \KwInput{(1) the number $n$ of nodes, \\ \quad\quad\quad (2) a ratio $r$ of hyperedges to nodes, \\ \quad\quad\quad (3) the maximum size  $k$ of a hyperedge}
    \KwOutput{a directed hypergraph: $G=(V, E)$}
     $V \leftarrow [n]$, $E \leftarrow \emptyset$ \\
     \For{$1:r\cdot n$}{
        Choose the hyperarc size $d\in \{2,\dots,k\}$ uniformly at random\label{algo:random:size}\\ 
        Choose $U\subseteq V$ such that $|U|=d$ uniformly at random \label{algo:random:nodes}\\
        $T, H \leftarrow$ Split $U$ into two groups of sizes $\lfloor \frac{d}{2} \rfloor$ and $\lceil \frac{d}{2} \rceil$\label{algo:random:split}\\
        $E\leftarrow E\cup \{\langle T, H \rangle\}$ \\
    }
    \Return{$G=(V,E)$}
\end{algorithm}

\clearpage
\bibliographystyle{IEEEtran}
\bibliography{bib.bib}

\begin{thebibliography}{10}
\providecommand{\url}[1]{#1}
\csname url@samestyle\endcsname
\providecommand{\newblock}{\relax}
\providecommand{\bibinfo}[2]{#2}
\providecommand{\BIBentrySTDinterwordspacing}{\spaceskip=0pt\relax}
\providecommand{\BIBentryALTinterwordstretchfactor}{4}
\providecommand{\BIBentryALTinterwordspacing}{\spaceskip=\fontdimen2\font plus
\BIBentryALTinterwordstretchfactor\fontdimen3\font minus
  \fontdimen4\font\relax}
\providecommand{\BIBforeignlanguage}[2]{{%
\expandafter\ifx\csname l@#1\endcsname\relax
\typeout{** WARNING: IEEEtran.bst: No hyphenation pattern has been}%
\typeout{** loaded for the language `#1'. Using the pattern for}%
\typeout{** the default language instead.}%
\else
\language=\csname l@#1\endcsname
\fi
#2}}
\providecommand{\BIBdecl}{\relax}
\BIBdecl

\bibitem{benson2018simplicial}
A.~R. Benson, R.~Abebe, M.~T. Schaub, A.~Jadbabaie, and J.~Kleinberg,
  ``Simplicial closure and higher-order link prediction,'' \emph{PNAS}, vol.
  115, no.~48, pp. E11\,221--E11\,230, 2018.

\bibitem{kim2022reciprocity}
S.~Kim, M.~Choe, J.~Yoo, and K.~Shin, ``Reciprocity in directed hypergraphs:
  Measures, findings, and generators,'' in \emph{ICDM}, 2022.

\bibitem{Tang:08KDD}
J.~Tang, J.~Zhang, L.~Yao, J.~Li, L.~Zhang, and Z.~Su, ``Arnetminer: Extraction
  and mining of academic social networks,'' in \emph{KDD}, 2008.

\bibitem{stackexchange}
S.~Exchange, ``Stack exchange data dump,''
  \url{https://archive.org/details/stackexchange}, 2020.

\bibitem{wu2021detecting}
J.~Wu, J.~Liu, W.~Chen, H.~Huang, Z.~Zheng, and Y.~Zhang, ``Detecting mixing
  services via mining bitcoin transaction network with hybrid motifs,''
  \emph{IEEE Trans. Syst. Man Cybern.: Syst.}, vol.~52, no.~4, pp. 2237--2249,
  2021.

\bibitem{moyano2016strong}
F.~J. M.-R. Moyano and F.~Jose, ``Strong connectivity in directed hypergraphs
  and its application to the atomic decomposition of ontologies,'' \emph{ETSI
  Informatica}, 2016.

\bibitem{tran2020directed}
L.~H. Tran and L.~H. Tran, ``Directed hypergraph neural network,'' \emph{arXiv
  preprint arXiv:2008.03626}, 2020.

\bibitem{luo2022directed}
X.~Luo, J.~Peng, and J.~Liang, ``Directed hypergraph attention network for
  traffic forecasting,'' \emph{IET Intelligent Transport Systems}, vol.~16,
  no.~1, pp. 85--98, 2022.

\bibitem{gracious2023neural}
T.~Gracious, A.~Gupta, and A.~Dukkipati, ``Neural temporal point process for
  forecasting higher order and directional interactions,'' \emph{arXiv preprint
  arXiv:2301.12210}, 2023.

\bibitem{yadati2020nhp}
N.~Yadati, V.~Nitin, M.~Nimishakavi, P.~Yadav, A.~Louis, and P.~Talukdar,
  ``Nhp: Neural hypergraph link prediction,'' in \emph{CIKM}, 2020.

\bibitem{yadati2021knowledge}
N.~Yadati, R.~Dayanidhi, S.~Vaishnavi, K.~Indira, and G.~Srinidhi, ``Knowledge
  base question answering through recursive hypergraphs,'' in \emph{EACL},
  2021.

\bibitem{milo2004superfamilies}
R.~Milo, S.~Itzkovitz, N.~Kashtan, R.~Levitt, S.~Shen-Orr, I.~Ayzenshtat,
  M.~Sheffer, and U.~Alon, ``Superfamilies of evolved and designed networks,''
  \emph{Science}, vol. 303, no. 5663, pp. 1538--1542, 2004.

\bibitem{milenkovic2008uncovering}
T.~Milenkovi{\'c} and N.~Pr{\v{z}}ulj, ``Uncovering biological network function
  via graphlet degree signatures,'' \emph{Cancer informatics}, vol.~6, pp.
  CIN--S680, 2008.

\bibitem{sarajlic2016graphlet}
A.~Sarajli{\'c}, N.~Malod-Dognin, {\"O}.~N. Yavero{\u{g}}lu, and
  N.~Pr{\v{z}}ulj, ``Graphlet-based characterization of directed networks,''
  \emph{Scientific Reports}, vol.~6, no.~1, pp. 1--14, 2016.

\bibitem{lee2020hypergraph}
G.~Lee, J.~Ko, and K.~Shin, ``Hypergraph motifs: Concepts, algorithms, and
  discoveries,'' \emph{PVLDB}, vol.~13, no.~11, pp. 2256--2269, 2020.

\bibitem{lotito2022higher}
Q.~F. Lotito, F.~Musciotto, A.~Montresor, and F.~Battiston, ``Higher-order
  motif analysis in hypergraphs,'' \emph{Communications Physics}, vol.~5,
  no.~1, pp. 1--8, 2022.

\bibitem{lee2021thyme+}
G.~Lee and K.~Shin, ``Thyme+: Temporal hypergraph motifs and fast algorithms
  for exact counting,'' in \emph{ICDM}, 2021.

\bibitem{benson2016higher}
A.~R. Benson, D.~F. Gleich, and J.~Leskovec, ``Higher-order organization of
  complex networks,'' \emph{Science}, vol. 353, no. 6295, pp. 163--166, 2016.

\bibitem{abuoda2020link}
G.~AbuOda, G.~De~Francisci~Morales, and A.~Aboulnaga, ``Link prediction via
  higher-order motif features,'' in \emph{ECML/PKDD}, 2020.

\bibitem{feng2020link}
J.~Feng and S.~Chen, ``Link prediction based on orbit counting and graph
  auto-encoder,'' \emph{IEEE Access}, vol.~8, pp. 226\,773--226\,783, 2020.

\bibitem{appendix}
H.~Moon, H.~Kim, S.~Kim, and K.~Shin, ``Code, datasets, and appendix,''
  \url{https://github.com/hhyy0401/CODA}, 2023.

\bibitem{ajemni2014modeling}
H.~Ajemni, R.~El~Harrabi, and M.~N. Abdelkrim, ``Modeling of chemical reaction
  kinetics using hypergraph tools,'' in \emph{CISTEM}, 2014.

\bibitem{pearcy2014hypergraph}
N.~Pearcy, J.~J. Crofts, and N.~Chuzhanova, ``Hypergraph models of
  metabolism,'' \emph{IJABE}, vol.~8, no.~8, pp. 752--756, 2014.

\bibitem{ranshous2017exchange}
S.~Ranshous, C.~A. Joslyn, S.~Kreyling, K.~Nowak, N.~F. Samatova, C.~L. West,
  and S.~Winters, ``Exchange pattern mining in the bitcoin transaction directed
  hypergraph,'' in \emph{FC}, 2017.

\bibitem{yadati2021graph}
N.~Yadati, T.~Gao, S.~Asoodeh, P.~Talukdar, and A.~Louis, ``Graph neural
  networks for soft semi-supervised learning on hypergraphs,'' in \emph{PAKDD},
  2021.

\bibitem{comrie2021hypergraph}
C.~Comrie and J.~Kleinberg, ``Hypergraph ego-networks and their temporal
  evolution,'' in \emph{ICDM}, 2021.

\bibitem{benson2018sequences}
A.~R. Benson, R.~Kumar, and A.~Tomkins, ``Sequences of sets,'' in \emph{KDD},
  2018.

\bibitem{lee2021hyperedges}
G.~Lee, M.~Choe, and K.~Shin, ``How do hyperedges overlap in real-world
  hypergraphs?-patterns, measures, and generators,'' in \emph{WWW}, 2021.

\bibitem{do2020structural}
M.~T. Do, S.-e. Yoon, B.~Hooi, and K.~Shin, ``Structural patterns and
  generative models of real-world hypergraphs,'' in \emph{KDD}, 2020.

\bibitem{gao2018android}
T.~Gao, W.~Peng, D.~Sisodia, T.~K. Saha, F.~Li, and M.~Al~Hasan, ``Android
  malware detection via graphlet sampling,'' \emph{IEEE Trans. Mob.}, vol.~18,
  no.~12, pp. 2754--2767, 2018.

\bibitem{hovcevar2014combinatorial}
T.~Ho{\v{c}}evar and J.~Dem{\v{s}}ar, ``A combinatorial approach to graphlet
  counting,'' \emph{Bioinformatics}, vol.~30, no.~4, pp. 559--565, 2014.

\bibitem{ahmed2015efficient}
N.~K. Ahmed, J.~Neville, R.~A. Rossi, and N.~Duffield, ``Efficient graphlet
  counting for large networks,'' in \emph{ICDM}, 2015.

\bibitem{pinar2017escape}
A.~Pinar, C.~Seshadhri, and V.~Vishal, ``Escape: Efficiently counting all
  5-vertex subgraphs,'' in \emph{WWW}, 2017.

\bibitem{bressan2017counting}
M.~Bressan, F.~Chierichetti, R.~Kumar, S.~Leucci, and A.~Panconesi, ``Counting
  graphlets: Space vs time,'' in \emph{WSDM}, 2017.

\bibitem{bressan2018motif}
------, ``Motif counting beyond five nodes,'' \emph{ACM Transactions on
  Knowledge Discovery from Data}, vol.~12, no.~4, pp. 1--25, 2018.

\bibitem{bressan2019motivo}
M.~Bressan, S.~Leucci, and A.~Panconesi, ``Motivo: Fast motif counting via
  succinct color coding and adaptive sampling,'' \emph{PVLDB}, vol.~12, no.~11,
  pp. 1651--1663, 2019.

\bibitem{lee2023hypergraph}
G.~Lee, S.~Yoon, J.~Ko, H.~Kim, and K.~Shin, ``Hypergraph motifs and their
  extensions beyond binary,'' \emph{arXiv preprint arXiv:2310.15668}, 2023.

\bibitem{chodrow2020configuration}
P.~S. Chodrow, ``Configuration models of random hypergraphs,'' \emph{Journal of
  Complex Networks}, vol.~8, no.~3, p. cnaa018, 2020.

\bibitem{hoeffding1994probability}
W.~Hoeffding, ``Probability inequalities for sums of bounded random
  variables,'' \emph{The collected works of Wassily Hoeffding}, pp. 409--426,
  1994.

\bibitem{holland1977method}
P.~W. Holland and S.~Leinhardt, ``A method for detecting structure in
  sociometric data,'' in \emph{Social Networks}.\hskip 1em plus 0.5em minus
  0.4em\relax Elsevier, 1977, pp. 411--432.

\bibitem{batagelj2001subquadratic}
V.~Batagelj and A.~Mrvar, ``A subquadratic triad census algorithm for large
  sparse networks with small maximum degree,'' \emph{Social Networks}, vol.~23,
  no.~3, pp. 237--243, 2001.

\bibitem{patil2020negative}
P.~Patil, G.~Sharma, and M.~N. Murty, ``Negative sampling for hyperlink
  prediction in networks,'' in \emph{PAKDD}, 2020.

\bibitem{grover2016node2vec}
A.~Grover and J.~Leskovec, ``node2vec: Scalable feature learning for
  networks,'' in \emph{KDD}, 2016.

\bibitem{huang2019hyper2vec}
J.~Huang, C.~Chen, F.~Ye, J.~Wu, Z.~Zheng, and G.~Ling, ``Hyper2vec: Biased
  random walk for hyper-network embedding,'' in \emph{DASFAA}, 2019.

\bibitem{payne2019deep}
J.~Payne, ``Deep hyperedges: a framework for transductive and inductive
  learning on hypergraphs,'' \emph{arXiv preprint arXiv:1910.02633}, 2019.

\bibitem{scikit-learn}
F.~Pedregosa, G.~Varoquaux, A.~Gramfort, V.~Michel, B.~Thirion, O.~Grisel,
  M.~Blondel, P.~Prettenhofer, R.~Weiss, V.~Dubourg, J.~Vanderplas, A.~Passos,
  D.~Cournapeau, M.~Brucher, M.~Perrot, and E.~Duchesnay, ``Scikit-learn:
  Machine learning in {P}ython,'' \emph{JMLR}, vol.~12, pp. 2825--2830, 2011.

\bibitem{chen2016xgboost}
T.~Chen and C.~Guestrin, ``Xgboost: A scalable tree boosting system,'' in
  \emph{KDD}, 2016.

\bibitem{ke2017lightgbm}
G.~Ke, Q.~Meng, T.~Finley, T.~Wang, W.~Chen, W.~Ma, Q.~Ye, and T.-Y. Liu,
  ``Lightgbm: A highly efficient gradient boosting decision tree,'' in
  \emph{NeurIPS}, 2017.

\bibitem{feng2019hypergraph}
Y.~Feng, H.~You, Z.~Zhang, R.~Ji, and Y.~Gao, ``Hypergraph neural networks,''
  in \emph{AAAI}, 2019.

\bibitem{NEURIPS2019_1efa39bc}
N.~Yadati, M.~Nimishakavi, P.~Yadav, V.~Nitin, A.~Louis, and P.~Talukdar,
  ``Hypergcn: A new method for training graph convolutional networks on
  hypergraphs,'' in \emph{NeurIPS}, 2019.

\bibitem{ijcai21-UniGNN}
J.~Huang and J.~Yang, ``Unignn: a unified framework for graph and hypergraph
  neural networks,'' in \emph{IJCAI}, 2021.

\bibitem{newman2018networks}
M.~Newman, \emph{Networks}.\hskip 1em plus 0.5em minus 0.4em\relax Oxford
  university press, 2018.

\end{thebibliography}

\clearpage

\red{If the preview is not legible, please download the PDF file. Appendices A-C are included at the end of the main paper.}

\section*{APPENDIX (online)}
\label{sec:appendix_online}

\subsection{Datasets (Table~\ref{tab:datasets} of the main paper)}\label{app:data}
We describe the representation, sources, and preprocessing steps of the datasets used in this work. As a default preprocessing step, we remove all duplicate hyperarcs and self-loops.

\begin{itemize}[leftmargin = 0.5cm]
    \item \textbf{Metabolic datasets:} We use two metabolic datasets, \textt{metabolic-iAF1260b}, and \textt{metabolic-iJO1366}.
    Each node represents a gene, and each hyperarc represents a metabolic reaction, where each head and tail set indicates a set of genes. 
    When the genes in the tail set participate in a metabolic reaction, they become the genes in the head set of the corresponding hyperarc.
They are provided in the complete form of directed hypergraphs %
which do not require any preprocessing step.
    \item \textbf{Email datasets:} We use two email datasets, \textt{email-enron}%
    , and \textt{email-eu}%
    .
    Each node represents an account, and each hyperarc represents an email from a sender to one or more recipients, where the tail set consists of a node representing the sender, and the head set consists of nodes representing the recipients. We transformed the original pairwise graph into a directed hypergraph by considering all edges occurring at the same timestamp from the same sender as a single email (hyperarc).
Note that the size of tail sets is always 1 in these datasets. (i.e., $|T_i|=1, \forall i=\{1,\dots,|E|\}$.)
    \item \textbf{Citation datasets:} We use two citation datasets from DBLP: \textt{citation-data-science}, and \textt{citation-software}%
    . 
    We extracted papers in the fields of data science or software from the dataset. Nodes represent authors and the (head and tail) sets indicate co-authors of each publication. Hyperarcs indicate citation relationships, with the tail set representing the paper that cites the head set paper. %
    \item \textbf{Question \& Answering datasets:} We use two question \& answering datasets, \textt{qna-math} and \textt{qna-server}. 
    Following~\cite{kim2022reciprocity}, we created a directed hypergraph from the log data of the two question-answering sites%
    : Math Exchange and Server Fault. 
    Each node represents a user, and each hyperarc represents a post, with the tail set consisting of the answerers and the head set consisting of the questioner.
Note that the size of head sets is always 1 in these datasets. (i.e., $|H_i|=1, \forall i=\{1,\dots,|E|\}$.)
    \item \textbf{Bitcoin transaction datasets:} We use three bitcoin transaction datasets, \textt{bitcoin-2014}, \textt{bitcoin-2015}, and \textt{bitcoin-2016}, created from the original datasets%
    , as suggested in \cite{kim2022reciprocity}.
They contain the first 1.5 million transactions in Nov 2014, Jun 2015, and Jan 2016, respectively.
Each node represents an individual account, and each hyperarc represents a cryptocurrency transaction. The tail set of a hyperarc corresponds to the accounts selling the cryptocurrency, while the head set corresponds to the accounts buying the corresponding cryptocurrency.
\end{itemize}

\subsection{Count Distributions (Section~\ref{exp:domain} of the main paper)}\label{app:count_distribution}

We analyze the occurrence distributions of \ours in real-world and randomized directed hypergraphs (DHs). To ensure statistical significance, we generate ten randomized DHs and report the average counts. As shown in Figure~\ref{fig:count_all}, the counts of \ours in real-world directed hypergraphs are distinct from those in randomized directed hypergraphs.

\begin{figure}[t!]
    \centering
    \includegraphics[width=\linewidth]{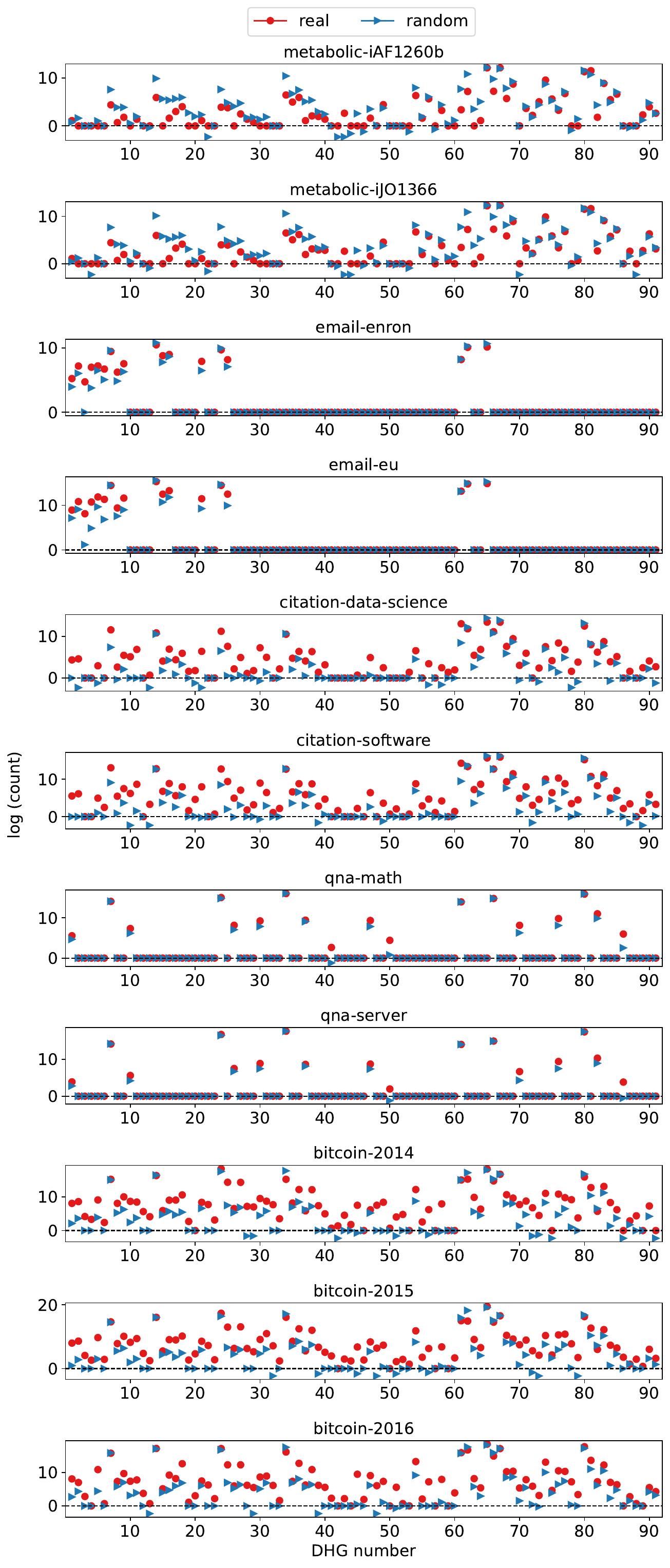}
    \caption{Log counts of \ours in real-world and randomized directed hypergraphs (DHs). The counts of \ours are clearly distinguished in real-world and randomized DHs.}
    \label{fig:count_all}
\end{figure}

\subsection{Temporal Analysis (Section~\ref{sec:exp_discover} of the main paper)}\label{app:exp_discover}
We analyze time-evolving DHs (all considered DHs except for
the \textt{metabolic} datasets, which do not contain timestamps). A time-evolving DH $G=(V, E)$ has timestamp $\tau_e$ for each $e\in E$, i.e., $e=\langle H, T, \tau_e\rangle$. With regard to the \textt{citation} datasets, \textt{citation-data-science} consists of 41 timestamps, while \textt{citation-software} includes 49 timestamps, with each publication year assigned as a timestamp. For the \textt{email}, \textt{qna}, and \textt{bitcoin} datasets, we consider $10$ timestamps $\{t_1, t_2, \cdots, t_{10}\}$ of the same interval, where $t_1<\dots <t_{10}=\max_{e\in E} \tau_e$ and
$t_{2}-t_{1}=t_{1}-
\min_{e\in E} \tau_e$.
For each timestamp $t_i$ above, we create a snapshot (i.e., sub-DH) where  the edge set is $E_i=\{e: \tau_e \leq t_i\}$ and the node set $V_i=\bigcup_{e\in E_i}\bar{e}$. Then, we compute the occurrence ratio of each \our in each sub-DH.

\subsection{Experimental Settings for Hyperarc Prediction (Section~\ref{sec:exp_apply} of the main paper)}\label{app:application_hyperparameter}
In this section, we list the hyperparameter settings of the feature vectors and classifiers used for the hyperarc prediction and report the detailed experimental setups.

\smallsection{Hyperparameter settings of feature vectors}
The embedding dimensions of node2vec, hyper2vec, and deep hyperedges are all fixed to 91.
Other hyperparameters of these methods are fixed to their default settings at the following links:
\begin{itemize}[leftmargin=*]
    \item \textbf{node2vec (n2v)}: \url{https://github.com/aditya-grover/node2vec}
    \item \textbf{hyper2vec (h2v)}: \url{https://github.com/jeffhj/NHNE}
    \item \textbf{deep hyperedges (deep-h)}: \url{https://github.com/0xpayne/deep-hyperedges}
\end{itemize}
Note that h-motif and triad do not have any hyperparameters.

\smallsection{Details of classifiers}
The hyperparameters of the tree-based classifiers (Decision Tree, Random Forest, XGBoost, and LightGBM), Logistic Regressor, KNN, and MLP are fixed to their default settings at the following links: 
\begin{itemize}[leftmargin=*]
    \item \textbf{Decision Tree (DT)}: \url{https://scikit-learn.org/stable/modules/generated/sklearn.tree.DecisionTreeClassifier}
    \item \textbf{Random Forest (RF)}: \url{https://scikit-learn.org/stable/modules/generated/sklearn.ensemble.RandomForestClassifier}
    \item \textbf{XGBoost (XGB)}: \url{https://xgboost.readthedocs.io/en/stable/}
    \item \textbf{LightGBM (LGBM)}: \url{https://lightgbm.readthedocs.io/en/latest/pythonapi/lightgbm.LGBMClassifier}
    \item \textbf{Logistic Regressor (LR)}: \url{https://scikit-learn.org/stable/modules/generated/sklearn.linear_model.LogisticRegression}
    \item \textbf{KNN}: \url{https://scikit-learn.org/stable/modules/generated/sklearn.neighbors.KNeighborsClassifier}
    \item \textbf{MLP}: \url{https://scikit-learn.org/stable/modules/generated/sklearn.neural_network.MLPClassifier}
\end{itemize}

To utilize the hyperarc-level feature vectors for hypergraph-neural-network-based (HNN-based) classifiers (HGNN, FastHyperGCN, and UniGCNII), which assume that the input is an undirected hypergraph with node features, we use the ``dual'' hypergraph of a given directed hypergraph (DH) as the input of the classifiers.
In the dual hypergraph $G^{*}=(V^{*}, E^{*}$) of a DH $G=(V,E)$, each node is a hyperarc in $G$ (i.e., $V^{*}=E$) and each hyperedge is the set of hyperarcs containing a node in $G$ (i.e., $E^{*}=\{E_v:v \in V\}$).

The hyperparameters of these HNN-based classifiers  are set as follows:
the number of layers and hidden dimension are all fixed to 2 and 128, respectively. 
We train HGNN and UniGCNII for 500 epochs using Adam
with a learning rate of 0.001 and a weight decay of $10^{-6}$, and FastHyperGCN for 200 epochs using Adam with a learning rate of 0.01, a weight decay of $5\times 10^{-4}$, and a dropout rate of 0.5.

For these HNN-based classifiers, we employ early stopping, and to this end, we divide the fake hyperarcs into the train, validation, and test sets using a 6:2:2 ratio. In each set, we uniformly sample the same number of real hyperarcs as the number of fake hyperarcs. For every 50 epochs, we measure the validation accuracy and save the model parameters. 
Then, we use the checkpoint (i.e., saved model parameters) with the highest validation accuracy to measure test performance.

\subsection{Application Results (Section~\ref{sec:exp_apply} of the main paper)}\label{app:application}
In this section, we report the full results of the hyperarc prediction problem.
Table~\ref{tab:results_all_in_one} reports the accuracy and AUROC and results (average over 100 trials).
The best performances are highlighted in bold, and the second-best performances are underlined.
Notably, in terms of average ranking, using \our vectors, including the dimension reduced versions, performs best in %
most settings, achieving up to $33\%$ higher AUROC on the \textt{bitcoin-2016} dataset and a $47\%$ higher accuracy on the \textt{bitcoin-2014} dataset than the second best features.

\begin{table*}
  \centering
  \caption{Hyperarc prediction performance. We compare nine hyperarc feature vectors using ten classifiers. The best performances are highlighted in bold, and the second-best performances are underlined. Notably, using \our vectors leads to the best performance (up to $47\%$ and $33\%$ better in terms of accuracy and AUROC, respectively) in most settings, indicating that \ours extract highly informative hyperarc features.}
  \label{tab:results_all_in_one}
  \vspace{-3mm}
  \subfigure[\textt{metabolic-iAF1260b}]{
    \scalebox{0.85}{%
\begin{tabular}{c|l|cc|cc|cccc|c}
        \toprule
        \multicolumn{2}{c|}{Dimension} & \multicolumn{2}{c|}{13} & \multicolumn{2}{c|}{26} & \multicolumn{4}{c|}{91} & 431 \\
        \midrule
        Measure & Model & DHG-13 & triad & DHG-26 & h-motif & ~DHG~ & ~n2v~ & ~h2v~ & deep-h & 3h-motif \\
        \midrule
        \multirow{10}{*}{ACC} & LR& 0.664$\pm$0.050 & 0.549$\pm$0.061 & 0.664$\pm$0.047 & 0.649$\pm$0.049 & 0.656$\pm$0.057 & 0.504$\pm$0.050 & 0.511$\pm$0.044 & \underline{0.701$\pm$0.055} & \textbf{0.705$\pm$0.044} \\
        &RF & \textbf{0.734$\pm$0.046} & 0.650$\pm$0.067 & \underline{0.733$\pm$0.048} & 0.681$\pm$0.059 & 0.690$\pm$0.063 & 0.531$\pm$0.054 & 0.518$\pm$0.047 & 0.714$\pm$0.050 & 0.717$\pm$0.060 \\
        &DT & \textbf{0.695$\pm$0.055} & 0.597$\pm$0.063 & \underline{0.683$\pm$0.051} & 0.632$\pm$0.065 & 0.651$\pm$0.053 & 0.512$\pm$0.057 & 0.506$\pm$0.056 & 0.583$\pm$0.055 & 0.615$\pm$0.065 \\
        &KNN & 0.702$\pm$0.047 & 0.625$\pm$0.067 & 0.703$\pm$0.044 & 0.696$\pm$0.058 & 0.696$\pm$0.049 & 0.537$\pm$0.050 & 0.534$\pm$0.056 & \underline{0.704$\pm$0.046} & \textbf{0.720$\pm$0.050} \\
        &MLP & \underline{0.709$\pm$0.049} & 0.603$\pm$0.062 & \textbf{0.716$\pm$0.053} & 0.646$\pm$0.054 & 0.654$\pm$0.063 & 0.533$\pm$0.051 & 0.537$\pm$0.047 & 0.705$\pm$0.054 & 0.667$\pm$0.054 \\
        &XGB & \underline{0.729$\pm$0.048} & 0.625$\pm$0.066 & \textbf{0.741$\pm$0.050} & 0.666$\pm$0.059 & 0.708$\pm$0.062 & 0.514$\pm$0.048 & 0.519$\pm$0.045 & 0.695$\pm$0.054 & 0.717$\pm$0.055 \\
\bottomrule
        \end{tabular}
    }
  }
  \subfigure[\textt{metabolic-iJO1366}]{
    \scalebox{0.85}{%
        \begin{tabular}{c|l|cc|cc|cccc|c}
        \toprule
        \multicolumn{2}{c|}{Dimension} & \multicolumn{2}{c|}{13} & \multicolumn{2}{c|}{26} & \multicolumn{4}{c|}{91} & 431 \\
        \midrule
        Measure & Model & DHG-13 & triad & DHG-26 & h-motif & ~DHG~ & ~n2v~ & ~h2v~ & deep-h & 3h-motif \\
        \midrule
        \multirow{10}{*}{ACC} & LR & 0.662$\pm$0.051 & 0.584$\pm$0.055 & 0.658$\pm$0.050 & 0.648$\pm$0.051 & 0.656$\pm$0.048 & 0.506$\pm$0.052 & 0.505$\pm$0.042 & \textbf{0.713$\pm$0.057} & \underline{0.699$\pm$0.051} \\
        & RF & \textbf{0.742$\pm$0.050} & 0.666$\pm$0.065 & \underline{0.730$\pm$0.045} & 0.704$\pm$0.059 & 0.698$\pm$0.068 & 0.532$\pm$0.045 & 0.527$\pm$0.045 & 0.723$\pm$0.048 & 0.710$\pm$0.066 \\
        & DT & \textbf{0.694$\pm$0.051} & 0.612$\pm$0.064 & \underline{0.680$\pm$0.053} & 0.641$\pm$0.057 & 0.656$\pm$0.071 & 0.509$\pm$0.053 & 0.517$\pm$0.049 & 0.582$\pm$0.061 & 0.612$\pm$0.055 \\
        & KNN & 0.705$\pm$0.049 & 0.628$\pm$0.063 & 0.689$\pm$0.045 & 0.691$\pm$0.055 & 0.682$\pm$0.051 & 0.520$\pm$0.053 & 0.539$\pm$0.049 & \textbf{0.725$\pm$0.054} & \underline{0.715$\pm$0.061} \\
        & MLP & \textbf{0.718$\pm$0.049} & 0.618$\pm$0.064 & 0.705$\pm$0.042 & 0.653$\pm$0.055 & 0.656$\pm$0.059 & 0.537$\pm$0.048 & 0.539$\pm$0.052 & \underline{0.710$\pm$0.055} & 0.644$\pm$0.051 \\
        & XGB & \textbf{0.741$\pm$0.050} & 0.657$\pm$0.062 & \underline{0.732$\pm$0.045} & 0.681$\pm$0.054 & 0.709$\pm$0.053 & 0.519$\pm$0.051 & 0.523$\pm$0.048 & 0.703$\pm$0.050 & 0.719$\pm$0.059 \\
\bottomrule
        \end{tabular}
    }
  }
\end{table*}

\begin{table*}
  \centering
  \vspace{-3mm}
  \subfigure[\textt{email-enron}]{
    \scalebox{0.85}{%
        \begin{tabular}{c|l|cc|cc|cccc|c}
        \toprule
        \multicolumn{2}{c|}{Dimension} & \multicolumn{2}{c|}{13} & \multicolumn{2}{c|}{26} & \multicolumn{4}{c|}{91} & 431 \\
        \midrule
        Measure & Model & DHG-13 & triad & DHG-26 & h-motif & ~DHG~ & ~n2v~ & ~h2v~ & deep-h & 3h-motif \\
        \midrule
        \multirow{10}{*}{ACC} & LR & 0.772$\pm$0.058 & 0.732$\pm$0.058 & \underline{0.795$\pm$0.065} & 0.752$\pm$0.053 & \textbf{0.804$\pm$0.054} & 0.578$\pm$0.069 & 0.492$\pm$0.058 & 0.590$\pm$0.062 & 0.749$\pm$0.068 \\
        & RF & 0.775$\pm$0.056 & 0.712$\pm$0.058 & \underline{0.780$\pm$0.071} & 0.773$\pm$0.056 & \textbf{0.796$\pm$0.053} & 0.626$\pm$0.069 & 0.562$\pm$0.073 & 0.592$\pm$0.063 & 0.754$\pm$0.066 \\
        & DT & 0.696$\pm$0.064 & 0.654$\pm$0.057 & \textbf{0.705$\pm$0.054} & 0.689$\pm$0.071 & \underline{0.705$\pm$0.069} & 0.551$\pm$0.069 & 0.528$\pm$0.073 & 0.542$\pm$0.068 & 0.650$\pm$0.076 \\
        & KNN & 0.761$\pm$0.068 & 0.694$\pm$0.058 & \underline{0.769$\pm$0.064} & 0.737$\pm$0.056 & \textbf{0.777$\pm$0.055} & 0.636$\pm$0.058 & 0.571$\pm$0.070 & 0.567$\pm$0.061 & 0.720$\pm$0.066 \\
        & MLP & \textbf{0.810$\pm$0.050} & 0.731$\pm$0.054 & \underline{0.808$\pm$0.062} & 0.751$\pm$0.053 & 0.805$\pm$0.055 & 0.639$\pm$0.073 & 0.551$\pm$0.080 & 0.588$\pm$0.064 & 0.721$\pm$0.070 \\
        & XGB & 0.763$\pm$0.059 & 0.709$\pm$0.060 & \textbf{0.780$\pm$0.064} & 0.763$\pm$0.064 & \underline{0.775$\pm$0.060} & 0.614$\pm$0.072 & 0.579$\pm$0.074 & 0.577$\pm$0.066 & 0.750$\pm$0.074 \\
\bottomrule
        \end{tabular}
    }
  }
  \subfigure[\textt{email-eu}]{
    \scalebox{0.85}{%
        \begin{tabular}{c|l|cc|cc|cccc|c}
        \toprule
        \multicolumn{2}{c|}{Dimension} & \multicolumn{2}{c|}{13} & \multicolumn{2}{c|}{26} & \multicolumn{4}{c|}{91} & 431 \\
        \midrule
        Measure & Model & DHG-13 & triad & DHG-26 & h-motif & ~DHG~ & ~n2v~ & ~h2v~ & deep-h & 3h-motif \\
        \midrule
        \multirow{10}{*}{ACC} & LR & 0.854$\pm$0.009 & 0.837$\pm$0.013 & \underline{0.855$\pm$0.009} & 0.776$\pm$0.019 & \textbf{0.869$\pm$0.010} & 0.618$\pm$0.028 & 0.496$\pm$0.012 & 0.659$\pm$0.019 & 0.799$\pm$0.017 \\
        & RF & 0.883$\pm$0.009 & 0.839$\pm$0.011 & \underline{0.890$\pm$0.008} & 0.838$\pm$0.011 & \textbf{0.907$\pm$0.010} & 0.652$\pm$0.017 & 0.515$\pm$0.021 & 0.668$\pm$0.020 & 0.840$\pm$0.011 \\
        & DT & 0.825$\pm$0.010 & 0.787$\pm$0.016 & \underline{0.828$\pm$0.010} & 0.761$\pm$0.016 & \textbf{0.849$\pm$0.012} & 0.546$\pm$0.015 & 0.504$\pm$0.019 & 0.564$\pm$0.015 & 0.753$\pm$0.018 \\
        & KNN & 0.867$\pm$0.009 & 0.838$\pm$0.011 & \underline{0.868$\pm$0.008} & 0.780$\pm$0.014 & \textbf{0.875$\pm$0.010} & 0.573$\pm$0.009 & 0.556$\pm$0.040 & 0.677$\pm$0.013 & 0.795$\pm$0.013 \\
        & MLP & 0.904$\pm$0.008 & 0.857$\pm$0.012 & \textbf{0.911$\pm$0.007} & 0.821$\pm$0.023 & \underline{0.906$\pm$0.011} & 0.660$\pm$0.047 & 0.507$\pm$0.018 & 0.675$\pm$0.048 & 0.808$\pm$0.027 \\
        & XGB & 0.885$\pm$0.008 & 0.854$\pm$0.011 & \underline{0.891$\pm$0.008} & 0.831$\pm$0.012 & \textbf{0.903$\pm$0.009} & 0.654$\pm$0.018 & 0.522$\pm$0.017 & 0.656$\pm$0.022 & 0.840$\pm$0.011 \\
\bottomrule
        \end{tabular}
    }
  }
\end{table*}

\begin{table*}
  \centering
  \vspace{-3mm}
  \subfigure[\textt{citation-data-science}]{
    \scalebox{0.85}{%
        \begin{tabular}{c|l|cc|cc|cccc|c}
        \toprule
        \multicolumn{2}{c|}{Dimension} & \multicolumn{2}{c|}{13} & \multicolumn{2}{c|}{26} & \multicolumn{4}{c|}{91} & 431 \\
        \midrule
        Measure & Model & DHG-13 & triad & DHG-26 & h-motif & ~DHG~ & ~n2v~ & ~h2v~ & deep-h & 3h-motif \\
        \midrule
        \multirow{10}{*}{ACC} & LR & 0.907$\pm$0.008 & 0.602$\pm$0.018 & \underline{0.918$\pm$0.007} & 0.751$\pm$0.039 & \textbf{0.921$\pm$0.011} & 0.527$\pm$0.017 & 0.504$\pm$0.008 & 0.593$\pm$0.022 & 0.837$\pm$0.027 \\
        & RF & \underline{0.952$\pm$0.006} & 0.644$\pm$0.022 & 0.952$\pm$0.008 & 0.855$\pm$0.017 & \textbf{0.977$\pm$0.004} & 0.548$\pm$0.015 & 0.500$\pm$0.012 & 0.599$\pm$0.023 & 0.923$\pm$0.007 \\
        & DT & \underline{0.914$\pm$0.011} & 0.583$\pm$0.018 & 0.906$\pm$0.015 & 0.777$\pm$0.026 & \textbf{0.963$\pm$0.005} & 0.511$\pm$0.016 & 0.497$\pm$0.015 & 0.539$\pm$0.015 & 0.848$\pm$0.015 \\
        & KNN & \underline{0.922$\pm$0.007} & 0.594$\pm$0.016 & \textbf{0.925$\pm$0.007} & 0.787$\pm$0.025 & 0.917$\pm$0.010 & 0.558$\pm$0.013 & 0.514$\pm$0.021 & 0.592$\pm$0.014 & 0.844$\pm$0.012 \\
        & MLP & 0.967$\pm$0.005 & 0.637$\pm$0.026 & \textbf{0.971$\pm$0.005} & 0.822$\pm$0.035 & \underline{0.969$\pm$0.005} & 0.545$\pm$0.024 & 0.502$\pm$0.010 & 0.633$\pm$0.037 & 0.888$\pm$0.018 \\
        & XGB & 0.958$\pm$0.005 & 0.639$\pm$0.025 & \underline{0.961$\pm$0.006} & 0.843$\pm$0.024 & \textbf{0.975$\pm$0.004} & 0.547$\pm$0.018 & 0.505$\pm$0.011 & 0.618$\pm$0.024 & 0.927$\pm$0.011 \\
\bottomrule
        \end{tabular}
    }
  }
  \subfigure[\textt{citation-software}]{
    \scalebox{0.85}{%
        \begin{tabular}{c|l|cc|cc|cccc|c}
        \toprule
        \multicolumn{2}{c|}{Dimension} & \multicolumn{2}{c|}{13} & \multicolumn{2}{c|}{26} & \multicolumn{4}{c|}{91} & 431 \\
        \midrule
        Measure & Model & DHG-13 & triad & DHG-26 & h-motif & ~DHG~ & ~n2v~ & ~h2v~ & deep-h & 3h-motif \\
        \midrule
        \multirow{10}{*}{ACC} & LR & 0.908$\pm$0.005 & 0.662$\pm$0.015 & \underline{0.912$\pm$0.004} & 0.767$\pm$0.007 & \textbf{0.919$\pm$0.007} & 0.541$\pm$0.018 & 0.508$\pm$0.013 & 0.625$\pm$0.019 & 0.829$\pm$0.010 \\
        & RF & 0.966$\pm$0.003 & 0.702$\pm$0.019 & \underline{0.968$\pm$0.003} & 0.866$\pm$0.022 & \textbf{0.984$\pm$0.002} & 0.568$\pm$0.018 & 0.501$\pm$0.012 & 0.621$\pm$0.022 & 0.941$\pm$0.004 \\
        & DT & \underline{0.945$\pm$0.004} & 0.623$\pm$0.012 & 0.944$\pm$0.004 & 0.783$\pm$0.045 & \textbf{0.974$\pm$0.002} & 0.519$\pm$0.010 & 0.498$\pm$0.016 & 0.548$\pm$0.009 & 0.875$\pm$0.012 \\
        & KNN & \textbf{0.947$\pm$0.004} & 0.652$\pm$0.013 & \underline{0.945$\pm$0.003} & 0.834$\pm$0.013 & 0.941$\pm$0.005 & 0.578$\pm$0.010 & 0.521$\pm$0.026 & 0.606$\pm$0.006 & 0.884$\pm$0.008 \\
        & MLP & 0.978$\pm$0.002 & 0.699$\pm$0.018 & \underline{0.979$\pm$0.002} & 0.844$\pm$0.029 & \textbf{0.980$\pm$0.002} & 0.575$\pm$0.032 & 0.512$\pm$0.022 & 0.652$\pm$0.041 & 0.919$\pm$0.011 \\
        & XGB & 0.971$\pm$0.002 & 0.704$\pm$0.018 & \underline{0.973$\pm$0.002} & 0.850$\pm$0.037 & \textbf{0.984$\pm$0.002} & 0.562$\pm$0.023 & 0.506$\pm$0.008 & 0.636$\pm$0.024 & 0.943$\pm$0.007 \\
\bottomrule
        \end{tabular}
    }
  }
\end{table*}

\begin{table*}
  \centering
  \vspace{-3mm}
  \subfigure[\textt{qna-math}]{
    \scalebox{0.85}{%
        \begin{tabular}{c|l|cc|cc|cccc|c}
        \toprule
        \multicolumn{2}{c|}{Dimension} & \multicolumn{2}{c|}{13} & \multicolumn{2}{c|}{26} & \multicolumn{4}{c|}{91} & 431 \\
        \midrule
        Measure & Model & DHG-13 & triad & DHG-26 & h-motif & ~DHG~ & ~n2v~ & ~h2v~ & deep-h & 3h-motif \\
        \midrule
        \multirow{10}{*}{ACC} & LR & 0.606$\pm$0.010 & 0.553$\pm$0.013 & \underline{0.607$\pm$0.010} & 0.579$\pm$0.011 & 0.604$\pm$0.009 & 0.500$\pm$0.005 & 0.504$\pm$0.006 & 0.566$\pm$0.009 & \textbf{0.615$\pm$0.012} \\
        & RF & 0.640$\pm$0.008 & 0.620$\pm$0.012 & 0.641$\pm$0.008 & 0.613$\pm$0.017 & \textbf{0.673$\pm$0.010} & 0.503$\pm$0.010 & 0.502$\pm$0.010 & 0.581$\pm$0.024 & \underline{0.656$\pm$0.015} \\
        & DT & \underline{0.598$\pm$0.009} & 0.572$\pm$0.010 & 0.595$\pm$0.009 & 0.547$\pm$0.016 & \textbf{0.617$\pm$0.011} & 0.502$\pm$0.008 & 0.502$\pm$0.011 & 0.547$\pm$0.013 & 0.574$\pm$0.017 \\
        & KNN & \textbf{0.607$\pm$0.008} & 0.553$\pm$0.011 & \underline{0.605$\pm$0.007} & 0.576$\pm$0.011 & 0.601$\pm$0.010 & 0.504$\pm$0.009 & 0.506$\pm$0.015 & 0.523$\pm$0.008 & 0.603$\pm$0.010 \\
        & MLP & 0.639$\pm$0.017 & 0.590$\pm$0.014 & \underline{0.640$\pm$0.016} & 0.598$\pm$0.015 & \textbf{0.679$\pm$0.010} & 0.505$\pm$0.009 & 0.503$\pm$0.005 & 0.611$\pm$0.043 & 0.626$\pm$0.020 \\
        & XGB & 0.643$\pm$0.008 & 0.629$\pm$0.011 & 0.644$\pm$0.009 & 0.607$\pm$0.019 & \textbf{0.681$\pm$0.010} & 0.502$\pm$0.010 & 0.508$\pm$0.008 & 0.601$\pm$0.027 & \underline{0.668$\pm$0.016} \\
\bottomrule
        \end{tabular}
    }
  }
  \subfigure[\textt{qna-server}]{
    \scalebox{0.85}{%
        \begin{tabular}{c|l|cc|cc|cccc|c}
        \toprule
        \multicolumn{2}{c|}{Dimension} & \multicolumn{2}{c|}{13} & \multicolumn{2}{c|}{26} & \multicolumn{4}{c|}{91} & 431 \\
        \midrule
        Measure & Model & DHG-13 & triad & DHG-26 & h-motif & ~DHG~ & ~n2v~ & ~h2v~ & deep-h & 3h-motif \\
        \midrule
        \multirow{10}{*}{ACC} & LR & 0.556$\pm$0.007 & 0.530$\pm$0.011 & \underline{0.557$\pm$0.007} & 0.533$\pm$0.009 & \textbf{0.561$\pm$0.007} & 0.502$\pm$0.004 & 0.509$\pm$0.008 & 0.556$\pm$0.007 & 0.553$\pm$0.007 \\
        & RF & \underline{0.643$\pm$0.005} & 0.590$\pm$0.012 & 0.640$\pm$0.004 & 0.565$\pm$0.012 & \textbf{0.661$\pm$0.004} & 0.500$\pm$0.006 & 0.503$\pm$0.005 & 0.565$\pm$0.018 & 0.631$\pm$0.010 \\
        & DT & \underline{0.615$\pm$0.005} & 0.572$\pm$0.012 & 0.614$\pm$0.005 & 0.528$\pm$0.008 & \textbf{0.626$\pm$0.006} & 0.500$\pm$0.007 & 0.501$\pm$0.007 & 0.545$\pm$0.009 & 0.571$\pm$0.011 \\
        & KNN & \textbf{0.628$\pm$0.005} & 0.566$\pm$0.010 & \underline{0.626$\pm$0.005} & 0.538$\pm$0.007 & 0.619$\pm$0.006 & 0.504$\pm$0.007 & 0.502$\pm$0.009 & 0.576$\pm$0.005 & 0.562$\pm$0.008 \\
        & MLP & \underline{\textbf{0.670$\pm$0.005}} & 0.603$\pm$0.014 & \underline{\textbf{0.670$\pm$0.005}} & 0.544$\pm$0.013 & 0.668$\pm$0.005 & 0.509$\pm$0.006 & 0.506$\pm$0.011 & 0.598$\pm$0.040 & 0.584$\pm$0.014 \\
        & XGB & \underline{0.658$\pm$0.005} & 0.612$\pm$0.016 & 0.656$\pm$0.005 & 0.575$\pm$0.013 & \textbf{0.679$\pm$0.005} & 0.503$\pm$0.005 & 0.503$\pm$0.007 & 0.585$\pm$0.019 & 0.647$\pm$0.010 \\
\bottomrule
        \end{tabular}
    }
  }
\end{table*}

\begin{table*}
  \centering
  \vspace{-3mm}
  \subfigure[\textt{bitcoin-2014}]{
    \scalebox{0.85}{%
        \begin{tabular}{c|l|cc|cc|cccc|c}
        \toprule
        \multicolumn{2}{c|}{Dimension} & \multicolumn{2}{c|}{13} & \multicolumn{2}{c|}{26} & \multicolumn{4}{c|}{91} & 431 \\
        \midrule
        Measure & Model & DHG-13 & triad & DHG-26 & h-motif & ~DHG~ & ~n2v~ & ~h2v~ & deep-h & 3h-motif \\
        \midrule
        \multirow{10}{*}{ACC} & LR & 0.552$\pm$0.003 & 0.524$\pm$0.002 & \underline{0.581$\pm$0.006} & \multirow{13}{*}{O.O.T.*} & \textbf{0.626$\pm$0.022} & 0.527$\pm$0.027 & 0.527$\pm$0.014 & 0.580$\pm$0.027 & \multirow{13}{*}{O.O.T.*} \\
        & RF & 0.914$\pm$0.003 & 0.731$\pm$0.029 & \underline{0.919$\pm$0.003} &  & \textbf{0.927$\pm$0.010} & 0.540$\pm$0.013 & 0.517$\pm$0.005 & 0.608$\pm$0.039 &  \\
        & DT & 0.886$\pm$0.003 & 0.703$\pm$0.027 & \underline{0.891$\pm$0.003} &  & \textbf{0.896$\pm$0.017} & 0.507$\pm$0.010 & 0.503$\pm$0.005 & 0.549$\pm$0.011 &  \\
        & KNN & \textbf{0.879$\pm$0.003} & 0.699$\pm$0.014 & \underline{0.878$\pm$0.003} &  & 0.838$\pm$0.016 & 0.577$\pm$0.011 & 0.537$\pm$0.010 & 0.639$\pm$0.004 &  \\
        & MLP & 0.763$\pm$0.011 & 0.612$\pm$0.049 & \underline{0.838$\pm$0.015} &  & \textbf{0.922$\pm$0.002} & 0.555$\pm$0.031 & 0.538$\pm$0.023 & 0.629$\pm$0.039 &  \\
        & XGB & 0.902$\pm$0.004 & 0.760$\pm$0.036 & \underline{0.912$\pm$0.003} &  & \textbf{0.935$\pm$0.003} & 0.554$\pm$0.014 & 0.520$\pm$0.012 & 0.615$\pm$0.032 &  \\
\end{tabular}
    }
  }
  \subfigure[\textt{bitcoin-2015}]{
    \scalebox{0.85}{%
        \begin{tabular}{c|l|cc|cc|cccc|c}
        \toprule
        \multicolumn{2}{c|}{Dimension} & \multicolumn{2}{c|}{13} & \multicolumn{2}{c|}{26} & \multicolumn{4}{c|}{91} & 431 \\
        \midrule
        Measure & Model & DHG-13 & triad & DHG-26 & h-motif & ~DHG~ & ~n2v~ & ~h2v~ & deep-h & 3h-motif \\
        \midrule
        \multirow{10}{*}{ACC} & LR & 0.539$\pm$0.016 & 0.539$\pm$0.006 & 0.561$\pm$0.003 & \multirow{13}{*}{O.O.T.*} & \underline{0.566$\pm$0.002} & 0.528$\pm$0.029 & 0.532$\pm$0.013 & \textbf{0.569$\pm$0.027} & \multirow{13}{*}{O.O.T.*} \\
        & RF & \underline{0.928$\pm$0.003} & 0.686$\pm$0.034 & \textbf{0.929$\pm$0.003} &  & 0.926$\pm$0.012 & 0.552$\pm$0.012 & 0.520$\pm$0.006 & 0.598$\pm$0.037 &  \\
        & DT & \underline{0.904$\pm$0.004} & 0.664$\pm$0.036 & \textbf{0.905$\pm$0.003} &  & 0.889$\pm$0.021 & 0.512$\pm$0.009 & 0.504$\pm$0.005 & 0.543$\pm$0.011 &  \\
        & KNN & \textbf{0.906$\pm$0.002} & 0.658$\pm$0.028 & \underline{0.895$\pm$0.003} &  & 0.839$\pm$0.025 & 0.592$\pm$0.008 & 0.539$\pm$0.009 & 0.628$\pm$0.004 &  \\
        & MLP & 0.825$\pm$0.018 & 0.635$\pm$0.027 & \underline{0.854$\pm$0.014} &  & \textbf{0.920$\pm$0.015} & 0.571$\pm$0.029 & 0.553$\pm$0.023 & 0.619$\pm$0.039 &  \\
        & XGB & 0.918$\pm$0.005 & 0.721$\pm$0.042 & \underline{0.923$\pm$0.004} &  & \textbf{0.939$\pm$0.005} & 0.571$\pm$0.012 & 0.519$\pm$0.011 & 0.606$\pm$0.028 &  \\
\end{tabular}
    }
  }
\end{table*}

\begin{table*}
  \centering
  \vspace{-3mm}
  \subfigure[\textt{bitcoin-2016}]{
    \scalebox{0.85}{%
        \begin{tabular}{c|l|cc|cc|cccc|c}
        \toprule
        \multicolumn{2}{c|}{Dimension} & \multicolumn{2}{c|}{13} & \multicolumn{2}{c|}{26} & \multicolumn{4}{c|}{91} & 431 \\
        \midrule
        Measure & Model & DHG-13 & triad & DHG-26 & h-motif & ~DHG~ & ~n2v~ & ~h2v~ & deep-h & 3h-motif \\
        \midrule
        \multirow{10}{*}{ACC} & LR & 0.545$\pm$0.002 & 0.549$\pm$0.022 & 0.551$\pm$0.002 & \multirow{13}{*}{O.O.T.*} & \underline{0.556$\pm$0.003} & 0.527$\pm$0.022 & 0.530$\pm$0.015 & \textbf{0.569$\pm$0.028} & \multirow{13}{*}{O.O.T.*} \\
        & RF & 0.915$\pm$0.002 & 0.694$\pm$0.030 & \underline{0.921$\pm$0.001} &  & \textbf{0.934$\pm$0.007} & 0.529$\pm$0.010 & 0.521$\pm$0.007 & 0.609$\pm$0.039 &  \\
        & DT & 0.888$\pm$0.002 & 0.671$\pm$0.029 & \underline{0.894$\pm$0.002} &  & \textbf{0.907$\pm$0.016} & 0.509$\pm$0.008 & 0.504$\pm$0.005 & 0.547$\pm$0.011 &  \\
        & KNN & \underline{0.888$\pm$0.001} & 0.668$\pm$0.028 & \textbf{0.890$\pm$0.001} &  & 0.866$\pm$0.020 & 0.567$\pm$0.010 & 0.536$\pm$0.009 & 0.626$\pm$0.003 &  \\
        & MLP & 0.825$\pm$0.010 & 0.591$\pm$0.030 & \underline{0.830$\pm$0.011} &  & \textbf{0.923$\pm$0.003} & 0.549$\pm$0.027 & 0.545$\pm$0.026 & 0.630$\pm$0.040 &  \\
        & XGB & 0.897$\pm$0.002 & 0.722$\pm$0.038 & \underline{0.913$\pm$0.002} &  & \textbf{0.938$\pm$0.005} & 0.544$\pm$0.012 & 0.522$\pm$0.013 & 0.621$\pm$0.030 &  \\
\end{tabular}
    }
  }
\end{table*}

\clearpage

\end{document}


\title{Online Appendix: Four-set Hypergraphlets for Characterization of Directed Hypergraphs\vspace{-10mm}}

\maketitle

\noindent\red{If the preview is not legible, please download the PDF file.}

\label{sec:appendix_online}

\subsection{Datasets (Table~II of the main paper)}\label{app:data}
We describe the representation, sources, and preprocessing steps of the datasets used in this work. As a default preprocessing step, we remove all duplicate hyperarcs and self-loops.

\begin{itemize}[leftmargin = 0.5cm]
    \item \textbf{Metabolic datasets:} We use two metabolic datasets, \textt{metabolic-iAF1260b}, and \textt{metabolic-iJO1366}.
    Each node represents a gene, and each hyperarc represents a metabolic reaction, where each head and tail set indicates a set of genes. 
    When the genes in the tail set participate in a metabolic reaction, they become the genes in the head set of the corresponding hyperarc.
They are provided in the complete form of directed hypergraphs \cite{yadati2020nhp} 
which do not require any preprocessing step.
    \item \textbf{Email datasets:} We use two email datasets, \textt{email-enron} \cite{chodrow2020annotated}
    , and \textt{email-eu} \cite{snapnets}
    .
    Each node represents an account, and each hyperarc represents an email from a sender to one or more recipients, where the tail set consists of a node representing the sender, and the head set consists of nodes representing the recipients. We transformed the original pairwise graph into a directed hypergraph by considering all edges occurring at the same timestamp from the same sender as a single email (hyperarc).
Note that the size of tail sets is always 1 in these datasets. (i.e., $|T_i|=1, \forall i=\{1,\dots,|E|\}$.)
    \item \textbf{Citation datasets:} We use two citation datasets from DBLP: \textt{citation-data-science}, and \textt{citation-software} \cite{Tang:08KDD, sinha2015overview}
    . 
    We extracted papers in the fields of data science or software from the dataset. Nodes represent authors and the (head and tail) sets indicate co-authors of each publication. Hyperarcs indicate citation relationships, with the tail set representing the paper that cites the head set paper. %
    \item \textbf{Question \& Answering datasets:} We use two question \& answering datasets, \textt{qna-math} and \textt{qna-server}. 
    Following~\cite{kim2022reciprocity}, we created a directed hypergraph from the log data of the two question-answering sites \cite{stackexchange}
    : Math Exchange and Server Fault. 
    Each node represents a user, and each hyperarc represents a post, with the tail set consisting of the answerers and the head set consisting of the questioner.
Note that the size of head sets is always 1 in these datasets. (i.e., $|H_i|=1, \forall i=\{1,\dots,|E|\}$.)
    \item \textbf{Bitcoin transaction datasets:} We use three bitcoin transaction datasets, \textt{bitcoin-2014}, \textt{bitcoin-2015}, and \textt{bitcoin-2016}, created from the original datasets \cite{wu2021detecting}
    , as suggested in \cite{kim2022reciprocity}.
They contain the first 1.5 million transactions in Nov 2014, Jun 2015, and Jan 2016, respectively.
Each node represents an individual account, and each hyperarc represents a cryptocurrency transaction. The tail set of a hyperarc corresponds to the accounts selling the cryptocurrency, while the head set corresponds to the accounts buying the corresponding cryptocurrency.
\end{itemize}

\subsection{Concentration Bounds (Section~IV of the main paper)}\label{app:con}

We present the sample concentration bounds of D-MoCHy and CODA-A, along with their corresponding proofs as below.

\begin{lemma}[Hoeffding's inequality~\cite{hoeffding1994probability}]
\label{lem:hoeff}
Let $X_1, X_2, \dots, X_n$ be independent random variables with $a_j\leq X_j\leq b_j$ for all $j\in [n]$. Consider the sum of random variables $X=X_1+\dots+X_n$. Then for any $t>0$, we have 
\[
    \Pr[|X-\mu|\geq t]\leq 2\exp\left(-\frac{2t^2}{\sum_{j=1}^n (b_j-a_j)^2}\right).
\]
\end{lemma}

\begin{proposition}[Sample Concentration Bound of \naive]\label{prop:naive:con}
   For any $\epsilon>0$, if $n\geq \frac{1}{2\epsilon^2}\ln (\frac{2}{\delta})$ and $|\Omega|>0$, $\Pr(|C[i]-|\Omega_i||\geq |\Omega|\cdot \epsilon)\leq \delta$, $\forall i\in [m]$.
\end{proposition} 

\begin{proof}
    Let $t:=|\Omega|\cdot\epsilon$. Since $\mathbb{E}[C[i]]=|\Omega_i|$ and $X^i_1$, $X^i_2$, $\cdots$, $X^i_n$ are independent random variables such that $0\leq X^i_j\leq \frac{1}{np(e, e')}=\frac{|\Omega|}{n}$ where $j\in [n]$, we can apply Hoeffding's inequality (Lemma~\ref{lem:hoeff}):
    \begin{align*}
        \Pr[|C[i]-|\Omega_i||\geq |\Omega|\cdot \epsilon]
        &\leq 2\exp\left(-\frac{2\epsilon^2|\Omega|^2}{n(|\Omega|/n)^2}\right)
        \\&\leq 2\exp(-2\epsilon^2n)\leq \delta. \qedhere
    \end{align*}
\end{proof}

\begin{proposition}[Sample Concentration Bound of \adv]\label{prop:adv:con}
   Let $W=\sum_{v\in V} w[v]$ and $\overline{h}=HM(|\bar{e}\cap \bar{e}'|^2)$ be a harmonic mean of $|\bar{e}\cap \bar{e}'|^2$ for all $(e, e')\in \Omega$, \textit{i.e.,} $HM(|\bar{e}\cap \bar{e}'|^2)=n/\sum_{(e, e')\in \Omega} \frac{1}{|\bar{e}\cap \bar{e}'|^2}$. Then for any $\epsilon>0$, if $n\geq \frac{1}{2\epsilon^2\overline{h}}\ln (\frac{2}{\delta})$ and $W>0$, $\Pr(|C[i]-|\Omega_i||\geq W\cdot \epsilon)\leq \delta$, $\forall i\in [m]$. \qedhere
\end{proposition} 

\begin{proof}
    Let $t:=W\cdot\epsilon$. Since $\mathbb{E}[C[i]]=|\Omega_i|$ and $X^i_1$, $X^i_2$, $\cdots$, $X^i_n$ are independent random variables such that $0\leq X^i_j\leq \frac{1}{np(e, e')}=\frac{W}{n\cdot |\bar{e}\cap \bar{e}'|}$ where $j\in [n]$, we can apply Hoeffding's inequality (Lemma~\ref{lem:hoeff}):
    \begin{align*}
        \Pr[|C[i]-|\Omega_i||\geq W\cdot \epsilon]
        &\leq 2\exp\left(-\frac{2\epsilon^2W^2}{\sum_{(e, e')\in \Omega}\left(\frac{W}{n|\bar{e}\cap \bar{e}'|}\right)^2}\right)
        \\& \leq 2\exp\left(-2\epsilon^2n\overline{h}\right)\leq \delta. \qedhere
    \end{align*}
\end{proof}

\subsection{\ata: Baseline Algorithm (Section~V-D of the main paper)}\label{app:ata}

In this section, we provide a detailed description of the baseline algorithm, \ata, including its unbiasedness, variance, and complexity. Its pseudocode is presented in Algorithm~\ref{algo:ata}.
Here, $p(e, e')=\left(\frac{1}{|N_{e}|}+\frac{1}{|N_{e'}|}\right)\cdot \frac{1}{|E_{\geq 1}|}$ where $E_{\geq 1}=\{e: N_{e}\geq 1\}$. 
Assume $E_{\geq 1}$ is given at first. 
Also, for space efficiency, we assume $N_{e}$ is maintained (Line~\ref{algo:ata:nei}).
$HM(A, B)$ on Line~\ref{algo:ata:hm} denotes the harmonic mean of $A$ and $B$. %

\begin{algorithm}[htb]
    \small
    \caption{\ata}\label{algo:ata}
    \SetKwInput{KwInput}{Input}
    \SetKwInput{KwOutput}{Output}
    \KwInput{(1) a directed hypergraph: $G=(V, E)$ \\
            \quad\quad\quad(2) \# of samples $n=q\cdot |E|$ for $q\in (0,1)$\\}
    \KwOutput{$C[i]$ for every $i\in [m]$}
     $C[i] \leftarrow 0, \forall i \in [m]$ \\ %
    \For{$1:n$}{
        Choose $e\in E_{\geq 1}$ uniformly at random \\   
        $N_{e}\leftarrow \{e'\in E\setminus \{e\}: e\cap e'\neq \emptyset\}$\label{algo:ata:nei}\\
        $C[f(e, e')]\leftarrow C[f(e, e')]+\frac{|E_{\geq 1}|}{2\cdot n}\cdot HM(|N_{e}|, |N_{e'}|)$ \label{algo:ata:hm}\\
    }
     \Return{$C$} 
\end{algorithm}

\begin{proposition}[Unbiasedness of \ata]
Algorithm~\ref{algo:ata} is unbiased, i.e., $\mathbb{E}[C[i]]=|\Omega_i|$. 
\end{proposition}
\begin{proof}
    Follow the flow of the proof of Proposition~2.%
\end{proof}

\begin{proposition}[Variance of \ata]
   The variance of $C[i]$ obtained by Algorithm~\ref{algo:ata} is \begin{align*}Var[C[i]]&=\sum_{(e, e')\in \Omega_i}\frac{1}{n}\left(\frac{1}{p(e, e')}-1\right)\\&=\sum_{(e, e')\in \Omega_i}\frac{1}{n}\left(\frac{|E_{\geq 1}|}{2}\cdot HM(|N_{e}|, |N_{e'}|)-1\right).\end{align*} %
\end{proposition} 
\begin{proof}
    Follow the flow of the proof of Proposition~3.%
\end{proof}
\begin{proposition}[Time \& Space complexity of \ata]\label{prop:a2a}
    The time complexity of Algorithm~\ref{algo:ata} is $O(n\cdot (\max_{e\in E}|\bar{e}|\cdot \max_{e\in E}|N_{e}|))$.
    Its space complexity is $O(\sum_{e\in E}|\bar{e}|)$.
\end{proposition}

\begin{proof}
    The information of a given directed graph is stored in $O(\sum_{e\in E}|\bar{e}|)$ space at first. 
    For time complexity, $O(\max_{e\in E} |\bar{e}|\cdot \max_{e\in E}|N_{e}|)$ time is required assuming $O(p\cdot q)$ time is taken for set union when there are $p$ sets, of which size bounded by $q$.
    For space complexity, $O(|N_{e}|)\in O(\sum_{e\in E} |\bar{e}|)$ space is needed. Checking $f(e, e')$ requires $O(\max_{(e, e')\in \Omega}\min(|\bar{e}|, |\bar{e}'|))$-time, which is bounded by $O(\max_{e\in E} |\bar{e}|\cdot \max_{e\in E} |N_e|)$.
\end{proof}

\subsection{Count Distributions (Section~V-B of the main paper)}\label{app:count_distribution}

We analyze the occurrence distributions of \ours in real-world and randomized directed hypergraphs (DHs). To ensure statistical significance, we generate ten randomized DHs and report the average counts. As shown in Figure~\ref{fig:count_all}, the counts of \ours in real-world directed hypergraphs are distinct from those in randomized directed hypergraphs.

\begin{figure}[t!]
    \centering
    \includegraphics[width=\linewidth]{figures/count_ratio_metabolic_compact 5.pdf}
    \caption{Log counts of \ours in real-world and randomized directed hypergraphs (DHs). The counts of \ours are clearly distinguished in real-world and randomized DHs.}
    \label{fig:count_all}
\end{figure}

\subsection{Temporal Analysis (Section~V-E of the main paper)}\label{app:exp_discover}
We analyze time-evolving DHs (all considered DHs except for
the \textt{metabolic} datasets, which do not contain timestamps). A time-evolving DH $G=(V, E)$ has timestamp $\tau_e$ for each $e\in E$, i.e., $e=\langle H, T, \tau_e\rangle$. With regard to the \textt{citation} datasets, \textt{citation-data-science} consists of 41 timestamps, while \textt{citation-software} includes 49 timestamps, with each publication year assigned as a timestamp. For the \textt{email}, \textt{qna}, and \textt{bitcoin} datasets, we consider $10$ timestamps $\{t_1, t_2, \cdots, t_{10}\}$ of the same interval, where $t_1<\dots <t_{10}=\max_{e\in E} \tau_e$ and
$t_{2}-t_{1}=t_{1}-
\min_{e\in E} \tau_e$.
For each timestamp $t_i$ above, we create a snapshot (i.e., sub-DH) where the edge set is $E_i=\{e: \tau_e \leq t_i\}$ and the node set $V_i=\bigcup_{e\in E_i}\bar{e}$. Then, we compute the occurrence ratio of each \our in each sub-DH.

In Figure~\ref{fig:time_all}, we visualize \ours whose ratio is greater than specific thresholds while aggregating the rest as \textt{Others}.
The threshold values for the \textt{citation}, \textt{email}, \textt{qna}, and \textt{bitcoin} datasets are 0.03, 0.1, 0.01, and 0.03, respectively.
In addition, we summarize the time-evolving patterns of the top 10 most frequent \ours in Figure~\ref{fig:time_tendency_all}.
Notably, the datasets from the same domain not only have the same set of frequent \ours but also exhibit similar time-evolving tendencies for each \our.

\begin{figure}[htb!]
    \centering
    \includegraphics[width=\linewidth]{figures/TVP_CITATION (3).pdf}
    \vspace{1mm}
    \includegraphics[width=\linewidth]{figures/TVP_EMAIL.pdf}
    \vspace{1mm}
    \includegraphics[width=\linewidth]{figures/TVP_QNA.pdf}
    \vspace{1mm}
    \includegraphics[width=\linewidth]{figures/TVP_BITCOIN (2).pdf}
    \caption{
    Directed hypergraphs from the same domain not only share the same set of frequent \ours but also exhibit similar time-evolving tendencies. \ours
with ratios below certain thresholds (spec., $0.03$ for the \textt{citation} datasets,  $0.1$ for the \textt{email} datasets, $0.01$ for the \textt{qna} datasets, and 0.03 for the \textt{bitcoin} datasets) are grouped as \textt{Others}.
    }
    \label{fig:time_all}
\end{figure}

\begin{figure}[htb!]
    \vspace{-2mm}
    \centering
    \includegraphics[width=\linewidth]{figures/tendency (2).pdf}
     \caption{Time-evolving trend.
     Each row represents a domain of datasets and each column represents a \our. %
     \ours marked with  `X' have ratios less than a predefined threshold.
     The ratios of \ours marked with \red{$\uparrow$} or \blue{$\downarrow$} tend to increase or decrease, respectively. 
     }
     \label{fig:time_tendency_all}
\end{figure}

\smallsection{Citation datasets}
The ratios of \our-65, -67, and -80 increase, while those of \our-7, -61, and -62 decrease.
In \our-65, -67, and -80, all non-intersecting regions (i.e., regions 1, 4, 5, and 8 in Figure~1(c)  contain at least one node, while in \our-7, -61, and -62, some non-intersecting regions are empty.
That is, the number of non-empty non-intersecting regions increases over time.

\smallsection{Email datasets}
Among the frequent \ours, only the ratio of \our-24 decreases, whereas the ratios of \our-7, -14, -62, and -65 increase.
\our-24 is distinct from the other \ours in that two tail sets intersect in it.
Given that the size of tail sets in email datasets is always 1, the probability of two hyperarcs sharing the same tail set decreases over time, resulting in the decrease of the \our-24 ratio.

\smallsection{Qna datasets}
There is a dramatic decrease in the ratio of \our-80 and a dramatic increase in the ratio of \our-34.
The ratios of \our-61, -66, and -80  decrease, while the ratios of \our-7, -24, and -34 increase.
The increasing \ours have empty non-intersecting areas in their tail sets, while the declining \ours have no such areas.
This indicates that the number of users who frequently answer questions increases.

\smallsection{Bitcoin datasets}
Although the two \textt{bitcoin} datasets share the same set of frequent \ours, their tendencies to increase or decrease over time differ in them.
For example, \our-65 becomes dominant in the bitcoin-2016 dataset, with a ratio of 0.37, while it has a ratio of only 0.07 in the other dataset. Conversely, \our-24 becomes dominant in the bitcoin-2014 dataset, with a ratio of 0.81, but only has a ratio of 0.10 in the other dataset. The main difference between \our-65 and \our-24 is the presence or absence of an intersection between tail sets.
\our-24 has only an intersection between tail sets, but \our-65 does not.
This indicates that the diversity of accounts participating in transactions increases over time.

\subsection{Experimental Settings for Hyperarc Prediction (Section~V-C of the main paper)}\label{app:application_hyperparameter}
In this section, we list the hyperparameter settings of the feature vectors and classifiers used for the hyperarc prediction and report the detailed experimental setups.

\smallsection{Hyperparameter settings of feature vectors}
The embedding dimensions of node2vec, hyper2vec, and deep hyperedges are all fixed to 91.
Other hyperparameters of these methods are fixed to their default settings at the following links:
\begin{itemize}[leftmargin=*]
    \item \textbf{node2vec (n2v)}: \url{https://github.com/aditya-grover/node2vec}
    \item \textbf{hyper2vec (h2v)}: \url{https://github.com/jeffhj/NHNE}
    \item \textbf{deep hyperedges (deep-h)}: \url{https://github.com/0xpayne/deep-hyperedges}
\end{itemize}
Note that h-motif and triad do not have any hyperparameters.

\smallsection{Details of classifiers}
The hyperparameters of the tree-based classifiers (Decision Tree, Random Forest, XGBoost, and LightGBM), Logistic Regressor, KNN, and MLP are fixed to their default settings at the following links: 
\begin{itemize}[leftmargin=*]
    \item \textbf{Decision Tree (DT)}: \url{https://scikit-learn.org/stable/modules/generated/sklearn.tree.DecisionTreeClassifier}
    \item \textbf{Random Forest (RF)}: \url{https://scikit-learn.org/stable/modules/generated/sklearn.ensemble.RandomForestClassifier}
    \item \textbf{XGBoost (XGB)}: \url{https://xgboost.readthedocs.io/en/stable/}
    \item \textbf{LightGBM (LGBM)}: \url{https://lightgbm.readthedocs.io/en/latest/pythonapi/lightgbm.LGBMClassifier}
    \item \textbf{Logistic Regressor (LR)}: \url{https://scikit-learn.org/stable/modules/generated/sklearn.linear_model.LogisticRegression}
    \item \textbf{KNN}: \url{https://scikit-learn.org/stable/modules/generated/sklearn.neighbors.KNeighborsClassifier}
    \item \textbf{MLP}: \url{https://scikit-learn.org/stable/modules/generated/sklearn.neural_network.MLPClassifier}
\end{itemize}

To utilize the hyperarc-level feature vectors for hypergraph-neural-network-based (HNN-based) classifiers (HGNN, FastHyperGCN, and UniGCNII), which assume that the input is an undirected hypergraph with node features, we use the ``dual'' hypergraph of a given directed hypergraph (DH) as the input of the classifiers.
In the dual hypergraph $G^{*}=(V^{*}, E^{*}$) of a DH $G=(V,E)$, each node is a hyperarc in $G$ (i.e., $V^{*}=E$) and each hyperedge is the set of hyperarcs containing a node in $G$ (i.e., $E^{*}=\{E_v:v \in V\}$).

The hyperparameters of these HNN-based classifiers  are set as follows:
the number of layers and hidden dimension are all fixed to 2 and 128, respectively. 
We train HGNN and UniGCNII for 500 epochs using Adam
~\cite{kingma2014adam}
with a learning rate of 0.001 and a weight decay of $10^{-6}$, and FastHyperGCN for 200 epochs using Adam with a learning rate of 0.01, a weight decay of $5\times 10^{-4}$, and a dropout rate of 0.5.

For these HNN-based classifiers, we employ early stopping, and to this end, we divide the fake hyperarcs into the train, validation, and test sets using a 6:2:2 ratio. In each set, we uniformly sample the same number of real hyperarcs as the number of fake hyperarcs. For every 50 epochs, we measure the validation accuracy and save the model parameters. 
Then, we use the checkpoint (i.e., saved model parameters) with the highest validation accuracy to measure test performance.

\subsection{Application Results (Section~V-C of the main paper)}\label{app:application}
In this section, we report the full results of the hyperarc prediction problem.
Tables~\ref{tab:results_all_auc} and  \ref{tab:results_all_acc} report the AUROC and accuracy results (average over 100 trials).
The best performances are highlighted in bold.
Notably, in terms of average ranking, using \our vectors performs best in all settings, achieving up to $33\%$ higher AUROC on the \textt{bitcoin-2016} dataset and a $47\%$ higher accuracy on the \textt{bitcoin-2014} dataset than the second best features.

\begin{table*}[t!]
  \setlength\tabcolsep{4pt}
  \caption{
  Hyperarc prediction performance (AUROC). We compare six hyperarc feature vectors using ten classifiers, and the best performances are highlighted in bold. Notably, using \our vectors  leads to the best (up to $33\%$ better) performance in most settings, indicating that \ours extract highly informative hyperarc features.
  }
  
  \label{tab:results_all_auc}
  \centering
  \scalebox{0.7}{%
\begin{tabular}{l|c|c|c|c|c|c|c|c|c|c|c|c|c|c}
\toprule
Model & Dataset & \our & h-motif & triad & n2v & h2v & deep-h & Dataset & \our & h-motif & triad & n2v & h2v & deep-h\\
\midrule
LR & \multirow{13}{*}{\textt{MB}} & 0.728$\pm$0.022 & 0.724$\pm$0.015 & 0.580$\pm$0.017 & 0.506$\pm$0.021 & 0.509$\pm$0.013 & \textbf{0.785$\pm$0.017}
& \multirow{13}{*}{\textt{M6}} & 0.721$\pm$0.016 & 0.727$\pm$0.020 & 0.622$\pm$0.012 & 0.501$\pm$0.015 & 0.504$\pm$0.013 & \textbf{0.795$\pm$0.010} \\
RF && \textbf{0.826$\pm$0.015} & 0.784$\pm$0.020 & 0.703$\pm$0.011 & 0.539$\pm$0.018 & 0.533$\pm$0.013 & 0.793$\pm$0.018
&& \textbf{0.834$\pm$0.011} & 0.794$\pm$0.018 & 0.730$\pm$0.013 & 0.541$\pm$0.010 & 0.546$\pm$0.027 & 0.803$\pm$0.012 \\
DT && \textbf{0.651$\pm$0.009} & 0.632$\pm$0.019 & 0.599$\pm$0.015 & 0.512$\pm$0.014 & 0.506$\pm$0.011 & 0.583$\pm$0.016
&& \textbf{0.656$\pm$0.018} & 0.641$\pm$0.015 & 0.611$\pm$0.009 & 0.509$\pm$0.014 & 0.517$\pm$0.016 & 0.582$\pm$0.014 \\
KNN && \textbf{0.762$\pm$0.013} & 0.755$\pm$0.018 & 0.673$\pm$0.011 & 0.550$\pm$0.018 & 0.552$\pm$0.020 & 0.760$\pm$0.015
&& 0.746$\pm$0.012 & 0.752$\pm$0.017 & 0.673$\pm$0.018 & 0.535$\pm$0.015 & 0.560$\pm$0.018 & \textbf{0.778$\pm$0.013} \\
MLP && 0.696$\pm$0.016 & 0.694$\pm$0.021 & 0.622$\pm$0.021 & 0.543$\pm$0.016 & 0.559$\pm$0.020 & \textbf{0.800$\pm$0.009}
&& 0.687$\pm$0.016 & 0.705$\pm$0.011 & 0.646$\pm$0.018 & 0.548$\pm$0.012 & 0.559$\pm$0.017 & \textbf{0.804$\pm$0.011 }\\
XGB && \textbf{0.812$\pm$0.013} & 0.752$\pm$0.025 & 0.684$\pm$0.019 & 0.518$\pm$0.018 & 0.526$\pm$0.018 & 0.775$\pm$0.010
&& \textbf{0.819$\pm$0.020} & 0.765$\pm$0.016 & 0.710$\pm$0.014 & 0.530$\pm$0.014 & 0.537$\pm$0.019 & 0.792$\pm$0.011 \\
LGBM && \textbf{0.799$\pm$0.057} & 0.736$\pm$0.062 & 0.680$\pm$0.072 & 0.540$\pm$0.059 & 0.529$\pm$0.021 & 0.785$\pm$0.012
&& \textbf{0.822$\pm$0.047} & 0.747$\pm$0.055 & 0.687$\pm$0.063 & 0.515$\pm$0.057 & 0.541$\pm$0.016 & 0.794$\pm$0.015 \\
HGNN && \textbf{0.594$\pm$0.060} & 0.554$\pm$0.059 & 0.541$\pm$0.067 & 0.479$\pm$0.058 & 0.506$\pm$0.086 & 0.522$\pm$0.060
&& \textbf{0.601$\pm$0.062} & 0.551$\pm$0.049 & 0.552$\pm$0.052 & 0.509$\pm$0.051 & 0.526$\pm$0.027 & 0.534$\pm$0.043 \\
FHGCN && \textbf{0.729$\pm$0.065} & 0.567$\pm$0.072 & 0.612$\pm$0.056 & 0.506$\pm$0.059 & 0.490$\pm$0.067 & 0.561$\pm$0.060
&& \textbf{0.720$\pm$0.072} & 0.553$\pm$0.070 & 0.615$\pm$0.065 & 0.506$\pm$0.059 & 0.484$\pm$0.059 & 0.566$\pm$0.065 \\
UGCN\uppercase\expandafter{\romannumeral2} && \textbf{0.680$\pm$0.060} & 0.674$\pm$0.048 & 0.641$\pm$0.097 & 0.488$\pm$0.049 & 0.489$\pm$0.068 & 0.607$\pm$0.092
&& \textbf{0.663$\pm$0.058} & 0.663$\pm$0.044 & 0.647$\pm$0.056 & 0.502$\pm$0.045 & 0.503$\pm$0.035 & 0.642$\pm$0.043 \\
\rowcolor{Gray} Max && \textbf{0.826$\pm$0.015} & 0.784$\pm$0.020 & 0.703$\pm$0.011 & 0.550$\pm$0.018 & 0.559$\pm$0.020 & 0.800$\pm$0.009
&& \textbf{0.834$\pm$0.011} & 0.794$\pm$0.018 & 0.730$\pm$0.013 & 0.548$\pm$0.012 & 0.560$\pm$0.018 & 0.804$\pm$0.011 \\
\rowcolor{Gray} Avg. && \textbf{0.728$\pm$0.075} & 0.687$\pm$0.080 & 0.634$\pm$0.052 & 0.518$\pm$0.024 & 0.520$\pm$0.024 & 0.697$\pm$0.113
&& \textbf{0.727$\pm$0.079} & 0.690$\pm$0.086 & 0.649$\pm$0.053 & 0.520$\pm$0.017 & 0.528$\pm$0.025 & 0.709$\pm$0.113 \\
\rowcolor{Gray} Rank Avg. && \textbf{1.200$\pm$0.422} & 2.700$\pm$0.483 & 3.500$\pm$0.707 & 5.600$\pm$0.516 & 5.400$\pm$0.516 & 2.600$\pm$1.265
&& \textbf{1.600$\pm$0.966} & 2.600$\pm$0.699 & 3.400$\pm$0.843 & 5.900$\pm$0.316 & 5.100$\pm$0.316 & 2.400$\pm$1.265 \\
\midrule
\midrule
LR &\multirow{13}{*}{\textt{EN}}& \textbf{0.883$\pm$0.014} & 0.826$\pm$0.015 & 0.783$\pm$0.016 & 0.627$\pm$0.017 & 0.480$\pm$0.020 & 0.634$\pm$0.023 
& \multirow{13}{*}{\textt{EU}}& \textbf{0.933$\pm$0.002} & 0.838$\pm$0.004 & 0.876$\pm$0.003 & 0.691$\pm$0.004 & 0.494$\pm$0.003 & 0.722$\pm$0.002\\
RF && \textbf{0.880$\pm$0.016} & 0.856$\pm$0.010 & 0.773$\pm$0.021 & 0.684$\pm$0.024 & 0.624$\pm$0.029 & 0.623$\pm$0.024
&& \textbf{0.960$\pm$0.002} & 0.921$\pm$0.003 & 0.901$\pm$0.003 & 0.737$\pm$0.003 & 0.529$\pm$0.004 & 0.770$\pm$0.004 \\
DT && \textbf{0.707$\pm$0.010} & 0.690$\pm$0.019 & 0.652$\pm$0.019 & 0.551$\pm$0.022 & 0.529$\pm$0.017 & 0.542$\pm$0.021
&& \textbf{0.852$\pm$0.003} & 0.762$\pm$0.004 & 0.785$\pm$0.005 & 0.546$\pm$0.004 & 0.504$\pm$0.007 & 0.564$\pm$0.007 \\
KNN && \textbf{0.846$\pm$0.013} & 0.810$\pm$0.015 & 0.745$\pm$0.020 & 0.685$\pm$0.019 & 0.597$\pm$0.026 & 0.591$\pm$0.024
&& \textbf{0.920$\pm$0.002} & 0.854$\pm$0.005 & 0.881$\pm$0.003 & 0.701$\pm$0.003 & 0.586$\pm$0.019 & 0.749$\pm$0.003 \\
MLP && \textbf{0.883$\pm$0.013} & 0.825$\pm$0.017 & 0.780$\pm$0.013 & 0.697$\pm$0.019 & 0.618$\pm$0.024 & 0.636$\pm$0.022 
&& \textbf{0.962$\pm$0.001} & 0.909$\pm$0.004 & 0.911$\pm$0.004 & 0.790$\pm$0.005 & 0.539$\pm$0.015 & 0.802$\pm$0.003 \\
XGB && \textbf{0.863$\pm$0.017} & 0.847$\pm$0.014 & 0.765$\pm$0.022 & 0.666$\pm$0.022 & 0.632$\pm$0.025 & 0.610$\pm$0.019
&& \textbf{0.961$\pm$0.001} & 0.917$\pm$0.002 & 0.905$\pm$0.003 & 0.747$\pm$0.005 & 0.557$\pm$0.011 & 0.777$\pm$0.002 \\
LGBM && 0.842$\pm$0.056 & \textbf{0.847$\pm$0.060} & 0.769$\pm$0.063 & 0.655$\pm$0.081 & 0.635$\pm$0.027 & 0.612$\pm$0.027
&& \textbf{0.963$\pm$0.005} & 0.923$\pm$0.006 & 0.908$\pm$0.008 & 0.761$\pm$0.014 & 0.555$\pm$0.007 & 0.796$\pm$0.002 \\
HGNN && 0.505$\pm$0.052 & 0.543$\pm$0.070 & \textbf{0.554$\pm$0.070} & 0.532$\pm$0.070 & 0.543$\pm$0.062 & 0.532$\pm$0.080  
&& \textbf{0.529$\pm$0.028} & 0.520$\pm$0.020 & 0.516$\pm$0.017 & 0.504$\pm$0.018 & 0.502$\pm$0.014 & 0.500$\pm$0.017 \\
FHGCN && \textbf{0.804$\pm$0.102} & 0.738$\pm$0.099 & 0.773$\pm$0.073 & 0.548$\pm$0.075 & 0.619$\pm$0.084 & 0.570$\pm$0.073
&& 0.849$\pm$0.054 & 0.724$\pm$0.051 & \textbf{0.888$\pm$0.041} & 0.520$\pm$0.029 & 0.550$\pm$0.057 & 0.614$\pm$0.070 \\
UGCN\uppercase\expandafter{\romannumeral2} && 0.787$\pm$0.048 & 0.767$\pm$0.065 & \textbf{0.788$\pm$0.052} & 0.706$\pm$0.045 & 0.739$\pm$0.056 & 0.606$\pm$0.059 
&& 0.874$\pm$0.010 & 0.805$\pm$0.010 & \textbf{0.912$\pm$0.007} & 0.793$\pm$0.014 & 0.812$\pm$0.020 & 0.784$\pm$0.012 \\
\rowcolor{Gray} Max && \textbf{0.883$\pm$0.014} & 0.856$\pm$0.010 & 0.788$\pm$0.052 & 0.706$\pm$0.045 & 0.739$\pm$0.056 & 0.636$\pm$0.022 
&& \textbf{0.963$\pm$0.005} & 0.923$\pm$0.006 & 0.788$\pm$0.052 & 0.706$\pm$0.045 & 0.739$\pm$0.056 & 0.802$\pm$0.003 \\
\rowcolor{Gray} Avg. && \textbf{0.800$\pm$0.117} & 0.775$\pm$0.098 & 0.738$\pm$0.076 & 0.635$\pm$0.067 & 0.602$\pm$0.071 & 0.596$\pm$0.037 
&& \textbf{0.880$\pm$0.132} & 0.817$\pm$0.126 & 0.848$\pm$0.123 & 0.679$\pm$0.113 & 0.563$\pm$0.092 & 0.708$\pm$0.108 \\
\rowcolor{Gray} Rank Avg. && \textbf{1.700$\pm$1.567} & 2.200$\pm$0.632 & 2.500$\pm$0.850 & 4.500$\pm$0.707 & 4.800$\pm$1.229 & 5.300$\pm$0.823
&& \textbf{1.200$\pm$0.422} & 2.700$\pm$0.675 & 2.200$\pm$0.789 & 5.000$\pm$0.471 & 5.500$\pm$0.972 & 4.400$\pm$0.843 \\
\midrule
\midrule
LR &\multirow{13}{*}{\textt{CD}}& \textbf{0.969$\pm$0.002} & 0.857$\pm$0.003 & 0.644$\pm$0.003 & 0.564$\pm$0.004 & 0.512$\pm$0.003 & 0.653$\pm$0.004
& \multirow{13}{*}{\textt{CS}} & \textbf{0.980$\pm$0.001} & 0.890$\pm$0.002 & 0.722$\pm$0.004 & 0.584$\pm$0.002 & 0.516$\pm$0.002 & 0.688$\pm$0.002 \\
RF && \textbf{0.997$\pm$0.000} & 0.939$\pm$0.001 & 0.703$\pm$0.007 & 0.573$\pm$0.006 & 0.498$\pm$0.004 & 0.707$\pm$0.005
&& \textbf{0.999$\pm$0.000} & 0.945$\pm$0.003 & 0.777$\pm$0.004 & 0.611$\pm$0.004 & 0.502$\pm$0.005 & 0.739$\pm$0.004 \\
DT && \textbf{0.963$\pm$0.001} & 0.777$\pm$0.007 & 0.583$\pm$0.006 & 0.511$\pm$0.004 & 0.497$\pm$0.003 & 0.539$\pm$0.004
&& \textbf{0.974$\pm$0.001} & 0.783$\pm$0.011 & 0.623$\pm$0.004 & 0.519$\pm$0.003 & 0.498$\pm$0.004 & 0.548$\pm$0.002 \\
KNN && \textbf{0.962$\pm$0.002} & 0.857$\pm$0.004 & 0.629$\pm$0.008 & 0.595$\pm$0.005 & 0.520$\pm$0.010 & 0.633$\pm$0.005
&& \textbf{0.974$\pm$0.001} & 0.899$\pm$0.002 & 0.702$\pm$0.004 & 0.632$\pm$0.003 & 0.531$\pm$0.011 & 0.659$\pm$0.002 \\
MLP && \textbf{0.990$\pm$0.001} & 0.914$\pm$0.007 & 0.693$\pm$0.009 & 0.597$\pm$0.006 & 0.510$\pm$0.005 & 0.776$\pm$0.003
&& \textbf{0.996$\pm$0.000} & 0.930$\pm$0.004 & 0.775$\pm$0.005 & 0.671$\pm$0.006 & 0.544$\pm$0.014 & 0.806$\pm$0.002 \\
XGB && \textbf{0.997$\pm$0.000} & 0.937$\pm$0.002 & 0.703$\pm$0.008 & 0.579$\pm$0.004 & 0.507$\pm$0.007 & 0.747$\pm$0.004 
&& \textbf{0.999$\pm$0.000} & 0.939$\pm$0.005 & 0.781$\pm$0.004 & 0.620$\pm$0.004 & 0.522$\pm$0.007 & 0.786$\pm$0.002  \\
LGBM && \textbf{0.998$\pm$0.001} & 0.941$\pm$0.008 & 0.719$\pm$0.021 & 0.587$\pm$0.014 & 0.511$\pm$0.004 & 0.782$\pm$0.003
&& \textbf{0.999$\pm$0.000} & 0.946$\pm$0.017 & 0.792$\pm$0.012 & 0.626$\pm$0.019 & 0.525$\pm$0.005 & 0.807$\pm$0.001 \\
HGNN && \textbf{0.629$\pm$0.008} & 0.528$\pm$0.017 & 0.523$\pm$0.013 & 0.537$\pm$0.014 & 0.540$\pm$0.022 & 0.510$\pm$0.013
&& 0.587$\pm$0.006 & 0.509$\pm$0.023 & 0.510$\pm$0.012 & 0.562$\pm$0.009 & \textbf{0.590$\pm$0.012} & 0.508$\pm$0.007 \\
FHGCN && \textbf{0.861$\pm$0.059} & 0.700$\pm$0.053 & 0.637$\pm$0.060 & 0.516$\pm$0.028 & 0.513$\pm$0.021 & 0.505$\pm$0.011
&& \textbf{0.852$\pm$0.044} & 0.718$\pm$0.050 & 0.714$\pm$0.054 & 0.513$\pm$0.025 & 0.511$\pm$0.024 & 0.505$\pm$0.013 \\
UGCN\uppercase\expandafter{\romannumeral2} && \textbf{0.975$\pm$0.005} & 0.874$\pm$0.010 & 0.714$\pm$0.013 & 0.851$\pm$0.016 & 0.672$\pm$0.042 & 0.547$\pm$0.012
&& \textbf{0.971$\pm$0.003} & 0.812$\pm$0.010 & 0.792$\pm$0.011 & 0.899$\pm$0.016 & 0.895$\pm$0.007 & 0.609$\pm$0.015 \\
\rowcolor{Gray} Max && \textbf{0.998$\pm$0.001} & 0.941$\pm$0.008 & 0.719$\pm$0.021 & 0.851$\pm$0.016 & 0.672$\pm$0.042 &  0.782$\pm$0.003
&& \textbf{0.999$\pm$0.000} & 0.946$\pm$0.017 & 0.792$\pm$0.012 & 0.899$\pm$0.016 & 0.895$\pm$0.007 & 0.807$\pm$0.001 \\
\rowcolor{Gray} Avg. && \textbf{0.934$\pm$0.115} & 0.832$\pm$0.132 & 0.655$\pm$0.064 & 0.591$\pm$0.096 & 0.528$\pm$0.052 & 0.640$\pm$0.110
&& \textbf{0.933$\pm$0.129} & 0.837$\pm$0.139 & 0.719$\pm$0.091 & 0.624$\pm$0.109 & 0.563$\pm$0.119 & 0.666$\pm$0.119 \\
\rowcolor{Gray} Rank Avg. && \textbf{1.000$\pm$0.000} & 2.200$\pm$0.632 & 3.900$\pm$0.568 & 4.500$\pm$0.850 & 5.400$\pm$1.265 & 4.000$\pm$1.414
&& \textbf{1.100$\pm$0.316} & 2.500$\pm$1.080 & 3.600$\pm$0.699 & 4.400$\pm$1.075 & 5.100$\pm$1.729 & 4.300$\pm$1.252 \\
\midrule
\midrule
LR &\multirow{13}{*}{\textt{QM}}& \textbf{0.652$\pm$0.003} & 0.620$\pm$0.004 & 0.580$\pm$0.004 & 0.499$\pm$0.003 & 0.514$\pm$0.003 & 0.600$\pm$0.002
& \multirow{13}{*}{\textt{QS}} & \textbf{0.598$\pm$0.001} & 0.553$\pm$0.002 & 0.558$\pm$0.003 & 0.512$\pm$0.001 & 0.528$\pm$0.002 & 0.586$\pm$0.002 \\
RF && 0.734$\pm$0.003 & 0.657$\pm$0.006 & 0.663$\pm$0.006 & 0.505$\pm$0.003 & 0.504$\pm$0.004 & \textbf{0.767$\pm$0.005}
&& 0.728$\pm$0.001 & 0.595$\pm$0.002 & 0.637$\pm$0.003 & 0.501$\pm$0.003 & 0.504$\pm$0.003 & \textbf{0.748$\pm$0.002} \\
DT && \textbf{0.621$\pm$0.003} & 0.547$\pm$0.006 & 0.569$\pm$0.004 & 0.502$\pm$0.001 & 0.502$\pm$0.004 & 0.547$\pm$0.004
&& \textbf{0.630$\pm$0.002} & 0.524$\pm$0.002 & 0.571$\pm$0.003 & 0.500$\pm$0.002 & 0.501$\pm$0.002 & 0.545$\pm$0.002 \\
KNN && \textbf{0.638$\pm$0.002} & 0.601$\pm$0.004 & 0.570$\pm$0.003 & 0.506$\pm$0.003 & 0.511$\pm$0.006 & 0.561$\pm$0.002
&& \textbf{0.669$\pm$0.002} & 0.552$\pm$0.003 & 0.599$\pm$0.002 & 0.505$\pm$0.003 & 0.503$\pm$0.004 & 0.610$\pm$0.002 \\
MLP && 0.737$\pm$0.004 & 0.649$\pm$0.005 & 0.634$\pm$0.003 & 0.514$\pm$0.004 & 0.509$\pm$0.002 & \textbf{0.834$\pm$0.004}
&& 0.717$\pm$0.002 & 0.573$\pm$0.002 & 0.630$\pm$0.002 & 0.511$\pm$0.002 & 0.525$\pm$0.004 & \textbf{0.815$\pm$0.001} \\
XGB && 0.744$\pm$0.003 & 0.660$\pm$0.005 & 0.677$\pm$0.005 & 0.504$\pm$0.003 & 0.513$\pm$0.004 & \textbf{0.823$\pm$0.002}
&& 0.745$\pm$0.001 & 0.612$\pm$0.002 & 0.653$\pm$0.003 & 0.504$\pm$0.002 & 0.507$\pm$0.002 & \textbf{0.801$\pm$0.002} \\
LGBM && 0.755$\pm$0.010 & 0.679$\pm$0.016 & 0.694$\pm$0.015 & 0.505$\pm$0.014 & 0.513$\pm$0.003 & \textbf{0.844$\pm$0.002}
&& 0.753$\pm$0.005 & 0.628$\pm$0.014 & 0.669$\pm$0.018 & 0.506$\pm$0.009 & 0.513$\pm$0.003 & \textbf{0.813$\pm$0.002}\\
HGNN && \textbf{0.570$\pm$0.013} & 0.520$\pm$0.011 & 0.535$\pm$0.012 & 0.519$\pm$0.012 & 0.566$\pm$0.009 & 0.551$\pm$0.010
&& \textbf{0.619$\pm$0.008} & 0.543$\pm$0.006 & 0.563$\pm$0.006 & 0.534$\pm$0.008 & 0.599$\pm$0.007 & 0.585$\pm$0.009 \\
FHGCN && 0.538$\pm$0.031 & 0.507$\pm$0.010 & \textbf{0.545$\pm$0.044} & 0.501$\pm$0.005 & 0.507$\pm$0.020 & 0.521$\pm$0.028
&& \textbf{0.593$\pm$0.041} & 0.528$\pm$0.025 & 0.546$\pm$0.049 & 0.502$\pm$0.007 & 0.503$\pm$0.012 & 0.511$\pm$0.018\\
UGCN\uppercase\expandafter{\romannumeral2} && 0.658$\pm$0.007 & 0.620$\pm$0.011 & 0.642$\pm$0.011 & 0.661$\pm$0.012 & 0.689$\pm$0.007 & \textbf{0.817$\pm$0.014}
&& 0.713$\pm$0.005 & 0.636$\pm$0.008 & 0.649$\pm$0.007 & 0.594$\pm$0.010 & 0.713$\pm$0.006 & \textbf{0.824$\pm$0.008}  \\
\rowcolor{Gray} Max && 0.755$\pm$0.010 & 0.679$\pm$0.016 & 0.694$\pm$0.015 & 0.661$\pm$0.012 & 0.689$\pm$0.007 & \textbf{0.844$\pm$0.002}
&& 0.753$\pm$0.005 & 0.636$\pm$0.008 & 0.669$\pm$0.018 & 0.594$\pm$0.010 & 0.713$\pm$0.006 & \textbf{0.824$\pm$0.008} \\
\rowcolor{Gray} Avg. && 0.665$\pm$0.076 & 0.606$\pm$0.061 & 0.611$\pm$0.058 & 0.522$\pm$0.049 & 0.533$\pm$0.058 & \textbf{0.687$\pm$0.140}
&& 0.677$\pm$0.062 & 0.574$\pm$0.041 & 0.608$\pm$0.045 & 0.517$\pm$0.029 & 0.540$\pm$0.068 & \textbf{0.684$\pm$0.127} \\
\rowcolor{Gray} Rank Avg. && \textbf{1.800$\pm$0.919} & 3.800$\pm$1.317 & 3.200$\pm$1.135 & 5.500$\pm$0.972 & 4.500$\pm$1.434 & 2.200$\pm$1.317
&& \textbf{1.700$\pm$0.949} & 4.100$\pm$0.568 & 2.900$\pm$0.568 & 5.900$\pm$0.316 & 4.500$\pm$1.354 & 1.900$\pm$1.101 \\
\midrule
\midrule
LR &\multirow{13}{*}{\textt{B4}}& \textbf{0.693$\pm$0.003} & \multirow{13}{*}{O.O.T.*} & 0.616$\pm$0.001 & 0.569$\pm$0.001 & 0.559$\pm$0.000 & 0.640$\pm$0.000
&\multirow{13}{*}{\textt{B5}} & \textbf{0.696$\pm$0.002} & \multirow{13}{*}{O.O.T.*} & 0.612$\pm$0.003 & 0.592$\pm$0.001 & 0.563$\pm$0.001 & 0.627$\pm$0.001 \\
RF && \textbf{0.976$\pm$0.001} && 0.799$\pm$0.009 & 0.562$\pm$0.003 & 0.524$\pm$0.003 & 0.729$\pm$0.008 
&& \textbf{0.977$\pm$0.001} && 0.755$\pm$0.013 & 0.578$\pm$0.004 & 0.528$\pm$0.003 & 0.715$\pm$0.005 \\
DT && \textbf{0.900$\pm$0.002} && 0.704$\pm$0.008 & 0.507$\pm$0.003 & 0.503$\pm$0.001 & 0.549$\pm$0.004
&& \textbf{0.893$\pm$0.006} && 0.672$\pm$0.012 & 0.512$\pm$0.002 & 0.504$\pm$0.001 & 0.543$\pm$0.003 \\
KNN && \textbf{0.905$\pm$0.002} && 0.744$\pm$0.002 & 0.606$\pm$0.005 & 0.556$\pm$0.002 & 0.694$\pm$0.002
&& \textbf{0.906$\pm$0.007} && 0.705$\pm$0.012 & 0.627$\pm$0.004 & 0.558$\pm$0.003 & 0.681$\pm$0.002 \\
MLP && \textbf{0.968$\pm$0.000} && 0.679$\pm$0.014 & 0.622$\pm$0.008 & 0.577$\pm$0.009 & 0.771$\pm$0.002
&& \textbf{0.967$\pm$0.001} && 0.690$\pm$0.009 & 0.639$\pm$0.008 & 0.596$\pm$0.010 & 0.761$\pm$0.003 \\
XGB && \textbf{0.980$\pm$0.000} && 0.841$\pm$0.008 & 0.575$\pm$0.003 & 0.549$\pm$0.003 & 0.784$\pm$0.001
&& \textbf{0.981$\pm$0.001} && 0.797$\pm$0.014 & 0.596$\pm$0.004 & 0.553$\pm$0.002 & 0.776$\pm$0.001 \\
LGBM && \textbf{0.980$\pm$0.002} && 0.847$\pm$0.030 & 0.583$\pm$0.018 & 0.558$\pm$0.004 & 0.789$\pm$0.001
&& \textbf{0.983$\pm$0.002} && 0.806$\pm$0.040 & 0.597$\pm$0.012 & 0.561$\pm$0.003 & 0.782$\pm$0.002 \\
HGNN && \textbf{0.808$\pm$0.002} && 0.663$\pm$0.004 & 0.664$\pm$0.004 & 0.664$\pm$0.003 & 0.607$\pm$0.002
&& \textbf{0.828$\pm$0.002} && 0.663$\pm$0.003 & 0.681$\pm$0.006 & 0.666$\pm$0.005 & 0.605$\pm$0.002 \\
FHGCN && \textbf{0.822$\pm$0.028} && 0.714$\pm$0.055 & 0.510$\pm$0.021 & 0.506$\pm$0.015 & 0.506$\pm$0.012
&& \textbf{0.823$\pm$0.040} && 0.712$\pm$0.070 & 0.510$\pm$0.018 & 0.504$\pm$0.011 & 0.503$\pm$0.008 \\
UGCN\uppercase\expandafter{\romannumeral2} && \textbf{0.926$\pm$0.003} & ~ & 0.854$\pm$0.003 & 0.855$\pm$0.008 & 0.790$\pm$0.005 & 0.713$\pm$0.003
&& \textbf{0.935$\pm$0.002} && 0.857$\pm$0.006 & 0.863$\pm$0.005 & 0.795$\pm$0.004 & 0.713$\pm$0.004 \\
\rowcolor{Gray} Max && \textbf{0.980$\pm$0.000} && 0.854$\pm$0.003 & 0.855$\pm$0.008 & 0.790$\pm$0.005 &  0.789$\pm$0.001
&& \textbf{0.983$\pm$0.002} && 0.857$\pm$0.006 & 0.863$\pm$0.005 & 0.795$\pm$0.004 &  0.782$\pm$0.002 \\
\rowcolor{Gray} Avg. && \textbf{0.896$\pm$0.095} && 0.746$\pm$0.085 & 0.605$\pm$0.100 & 0.598$\pm$0.083 & 0.678$\pm$0.099
&& \textbf{0.899$\pm$0.093} && 0.727$\pm$0.075 & 0.620$\pm$0.100 & 0.583$\pm$0.088 & 0.671$\pm$0.098 \\
\rowcolor{Gray} Rank Avg. && \textbf{1.000$\pm$0.000} && 2.500$\pm$0.707 & 3.600$\pm$0.699 & 4.500$\pm$0.972 & 3.400$\pm$1.174 
&& \textbf{1.000$\pm$0.000} && 2.500$\pm$0.707 & 3.500$\pm$0.850 & 4.600$\pm$0.699 & 3.400$\pm$1.174 \\
\midrule
LR &\multirow{10}{*}{\textt{B6}}& \textbf{0.654$\pm$0.002} & \multirow{13}{*}{O.O.T.*} & 0.612$\pm$0.002 & 0.564$\pm$0.001 & 0.567$\pm$0.001 & 0.635$\pm$0.001 \\
RF && \textbf{0.978$\pm$0.001} && 0.766$\pm$0.007 & 0.550$\pm$0.003 & 0.529$\pm$0.003 & 0.725$\pm$0.007 \\ 
DT && \textbf{0.909$\pm$0.005} && 0.681$\pm$0.007 & 0.509$\pm$0.002 & 0.504$\pm$0.001 & 0.547$\pm$0.003 \\ 
KNN && \textbf{0.927$\pm$0.005} && 0.716$\pm$0.007 & 0.592$\pm$0.004 & 0.553$\pm$0.004 & 0.678$\pm$0.002 \\ 
MLP && \textbf{0.968$\pm$0.001} && 0.655$\pm$0.006 & 0.600$\pm$0.006 & 0.587$\pm$0.011 & 0.770$\pm$0.002 \\
XGB && \textbf{0.982$\pm$0.001} && 0.801$\pm$0.006 & 0.562$\pm$0.004 & 0.550$\pm$0.002 & 0.783$\pm$0.001 \\ 
LGBM && \textbf{0.982$\pm$0.002} && 0.814$\pm$0.044 & 0.563$\pm$0.011 & 0.558$\pm$0.002 & 0.788$\pm$0.001 \\ 
HGNN && \textbf{0.827$\pm$0.002} && 0.659$\pm$0.003 & 0.659$\pm$0.006 & 0.663$\pm$0.004 & 0.596$\pm$0.003 \\
FHGCN && \textbf{0.825$\pm$0.049} && 0.712$\pm$0.062 & 0.508$\pm$0.020 & 0.507$\pm$0.017 & 0.505$\pm$0.015 \\
UGCN\uppercase\expandafter{\romannumeral2} && \textbf{0.931$\pm$0.002} & ~ & 0.847$\pm$0.006 & 0.856$\pm$0.010 & 0.801$\pm$0.004 & 0.704$\pm$0.005 \\
\rowcolor{Gray} Max && \textbf{0.982$\pm$0.001} && 0.847$\pm$0.006 & 0.856$\pm$0.010 & 0.801$\pm$0.004 & 0.788$\pm$0.001 \\
\rowcolor{Gray} Avg. && \textbf{0.898$\pm$0.104} && 0.726$\pm$0.078 & 0.596$\pm$0.101 & 0.582$\pm$0.089 & 0.673$\pm$0.100 \\
\rowcolor{Gray} Rank Avg. && \textbf{1.000$\pm$0.000} && 2.500$\pm$0.707 & 3.700$\pm$0.823 & 4.400$\pm$0.966 & 3.400$\pm$1.174 \\
\bottomrule
\multicolumn{15}{l}{* O.O.T: out-of-time ($>1$ day).}
\end{tabular}
}%
\end{table*}

\begin{table*}[t!]
  \setlength\tabcolsep{4pt}
  \caption{
  Hyperarc prediction performance (accuracy). We compare six hyperarc feature vectors using ten classifiers, and the best performances are highlighted in bold. Notably, using \our vectors  leads to the best (up to $47\%$ better) performance in most settings, indicating that \ours extract highly informative hyperarc features.
  }
  
  \label{tab:results_all_acc}
  \centering
  \scalebox{0.7}{%
\begin{tabular}{l|c|c|c|c|c|c|c|c|c|c|c|c|c|c}
\toprule
Model & Dataset & \our & h-motif & triad & n2v & h2v & deep-h & Dataset & \our & h-motif & triad & h2v & h2v & deep-h\\
\midrule
LR & \multirow{13}{*}{\textt{MB}} & 0.656$\pm$0.023 & 0.649$\pm$0.015 & 0.549$\pm$0.018 & 0.504$\pm$0.017 & 0.511$\pm$0.012 & \textbf{0.701$\pm$0.014}
& \multirow{13}{*}{\textt{M6}} & 0.656$\pm$0.016 & 0.648$\pm$0.019 & 0.584$\pm$0.011 & 0.506$\pm$0.011 & 0.505$\pm$0.011 & \textbf{0.713$\pm$0.012} \\
RF && 0.690$\pm$0.011 & 0.681$\pm$0.021 & 0.650$\pm$0.009 & 0.531$\pm$0.014 & 0.518$\pm$0.012 & \textbf{0.714$\pm$0.015}
&& 0.698$\pm$0.014 & 0.704$\pm$0.017 & 0.666$\pm$0.010 & 0.532$\pm$0.010 & 0.527$\pm$0.015 & \textbf{0.723$\pm$0.017} \\
DT && \textbf{0.651$\pm$0.009} & 0.632$\pm$0.019 & 0.597$\pm$0.016 & 0.512$\pm$0.014 & 0.506$\pm$0.011 & 0.583$\pm$0.016
&& \textbf{0.656$\pm$0.018} & 0.641$\pm$0.015 & 0.612$\pm$0.009 & 0.509$\pm$0.014 & 0.517$\pm$0.016 & 0.582$\pm$0.014 \\
KNN && 0.696$\pm$0.014 & 0.696$\pm$0.018 & 0.625$\pm$0.014 & 0.537$\pm$0.014 & 0.534$\pm$0.016 & \textbf{0.704$\pm$0.014}
&& 0.682$\pm$0.014 & 0.691$\pm$0.020 & 0.628$\pm$0.019 & 0.520$\pm$0.016 & 0.539$\pm$0.011 & \textbf{0.725$\pm$0.013} \\
MLP && 0.654$\pm$0.011 & 0.646$\pm$0.014 & 0.603$\pm$0.019 & 0.533$\pm$0.008 & 0.537$\pm$0.014 & \textbf{0.705$\pm$0.014}
&& 0.656$\pm$0.012 & 0.653$\pm$0.008 & 0.618$\pm$0.015 & 0.537$\pm$0.009 & 0.539$\pm$0.011 & \textbf{0.710$\pm$0.012} \\
XGB && \textbf{0.708$\pm$0.016} & 0.666$\pm$0.024 & 0.625$\pm$0.015 & 0.514$\pm$0.016 & 0.519$\pm$0.015 & 0.695$\pm$0.015
&& \textbf{0.709$\pm$0.022} & 0.681$\pm$0.017 & 0.657$\pm$0.013 & 0.519$\pm$0.017 & 0.522$\pm$0.013 & 0.703$\pm$0.014 \\
LGBM && 0.697$\pm$0.062 & 0.658$\pm$0.053 & 0.626$\pm$0.063 & 0.527$\pm$0.054 & 0.516$\pm$0.012 & \textbf{0.698$\pm$0.012}
&& \textbf{0.720$\pm$0.057} & 0.663$\pm$0.056 & 0.632$\pm$0.058 & 0.509$\pm$0.050 & 0.530$\pm$0.011 & 0.702$\pm$0.014 \\
HGNN && \textbf{0.549$\pm$0.055} & 0.543$\pm$0.040 & 0.538$\pm$0.063 & 0.483$\pm$0.045 & 0.498$\pm$0.067 & 0.526$\pm$0.057
&& \textbf{0.569$\pm$0.055} & 0.553$\pm$0.036 & 0.550$\pm$0.043 & 0.499$\pm$0.056 & 0.522$\pm$0.037 & 0.545$\pm$0.036 \\
FHGCN && \textbf{0.666$\pm$0.067} & 0.535$\pm$0.051 & 0.555$\pm$0.055 & 0.507$\pm$0.045 & 0.491$\pm$0.059 & 0.538$\pm$0.050
&& \textbf{0.653$\pm$0.071} & 0.532$\pm$0.053 & 0.566$\pm$0.059 & 0.504$\pm$0.045 & 0.498$\pm$0.048 & 0.548$\pm$0.060 \\
UGCN\uppercase\expandafter{\romannumeral2} && 0.618$\pm$0.061 & \textbf{0.625$\pm$0.038} & 0.596$\pm$0.069 & 0.491$\pm$0.050 & 0.494$\pm$0.054 & 0.597$\pm$0.070
&& 0.621$\pm$0.051 & 0.616$\pm$0.048 & 0.610$\pm$0.050 & 0.518$\pm$0.047 & 0.514$\pm$0.034 & \textbf{0.633$\pm$0.035} \\
\rowcolor{Gray} Max && 0.708$\pm$0.016 & 0.696$\pm$0.018 & 0.650$\pm$0.009 & 0.537$\pm$0.014 & 0.537$\pm$0.014 & \textbf{0.714$\pm$0.015}
&& 0.720$\pm$0.057 & 0.704$\pm$0.017 & 0.666$\pm$0.010 & 0.537$\pm$0.009 & 0.539$\pm$0.011 & \textbf{0.725$\pm$0.013} \\
\rowcolor{Gray} Avg. && \textbf{0.659$\pm$0.047} & 0.633$\pm$0.054 & 0.596$\pm$0.038 & 0.514$\pm$0.018 & 0.512$\pm$0.016 & 0.646$\pm$0.076 
&& \textbf{0.662$\pm$0.044} & 0.638$\pm$0.057 & 0.612$\pm$0.037 & 0.515$\pm$0.012 & 0.521$\pm$0.013 & 0.658$\pm$0.074\\
\rowcolor{Gray} Rank Avg. && \textbf{1.600$\pm$0.516} & 2.700$\pm$0.823 & 3.600$\pm$0.699 & 5.500$\pm$0.527 & 5.500$\pm$0.527 & 2.100$\pm$1.287  
&& \textbf{1.700$\pm$0.823} & 2.700$\pm$0.675 & 3.600$\pm$0.699 & 5.600$\pm$0.516 & 5.400$\pm$0.516 & 2.000$\pm$1.247  \\
\midrule
\midrule 
LR &\multirow{13}{*}{\textt{EN}}& \textbf{0.804$\pm$0.014} & 0.752$\pm$0.011 & 0.732$\pm$0.017 & 0.578$\pm$0.017 & 0.492$\pm$0.012 & 0.590$\pm$0.018 
& \multirow{13}{*}{\textt{EU}}& \textbf{0.869$\pm$0.001} & 0.776$\pm$0.005 & 0.837$\pm$0.004 & 0.618$\pm$0.008 & 0.496$\pm$0.003 & 0.659$\pm$0.006 \\
RF && \textbf{0.796$\pm$0.013} & 0.773$\pm$0.017 & 0.712$\pm$0.023 & 0.626$\pm$0.024 & 0.562$\pm$0.024 & 0.592$\pm$0.022
&& \textbf{0.907$\pm$0.003} & 0.838$\pm$0.004 & 0.839$\pm$0.003 & 0.652$\pm$0.003 & 0.515$\pm$0.002 & 0.668$\pm$0.005 \\
DT && \textbf{0.705$\pm$0.011} & 0.689$\pm$0.018 & 0.654$\pm$0.020 & 0.551$\pm$0.022 & 0.528$\pm$0.018 & 0.542$\pm$0.021
&& \textbf{0.849$\pm$0.003} & 0.761$\pm$0.005 & 0.787$\pm$0.005 & 0.546$\pm$0.004 & 0.504$\pm$0.007 & 0.564$\pm$0.007 \\
KNN && \textbf{0.778$\pm$0.014} & 0.737$\pm$0.016 & 0.694$\pm$0.019 & 0.636$\pm$0.020 & 0.571$\pm$0.017 & 0.567$\pm$0.022
&& \textbf{0.875$\pm$0.002} & 0.780$\pm$0.005 & 0.838$\pm$0.005 & 0.573$\pm$0.002 & 0.556$\pm$0.012 & 0.677$\pm$0.003  \\
MLP && \textbf{0.805$\pm$0.014} & 0.751$\pm$0.011 & 0.731$\pm$0.013 & 0.639$\pm$0.021 & 0.551$\pm$0.018 & 0.588$\pm$0.023
&& \textbf{0.906$\pm$0.003} & 0.821$\pm$0.007 & 0.857$\pm$0.005 & 0.660$\pm$0.016 & 0.507$\pm$0.005 & 0.675$\pm$0.011 \\
XGB && \textbf{0.775$\pm$0.018} & 0.763$\pm$0.014 & 0.709$\pm$0.026 & 0.614$\pm$0.020 & 0.579$\pm$0.018 & 0.577$\pm$0.020
&& \textbf{0.903$\pm$0.003} & 0.831$\pm$0.005 & 0.854$\pm$0.005 & 0.654$\pm$0.005 & 0.522$\pm$0.005 & 0.656$\pm$0.003\\
LGBM && 0.756$\pm$0.059 & \textbf{0.763$\pm$0.056} & 0.709$\pm$0.060 & 0.609$\pm$0.064 & 0.580$\pm$0.018 & 0.581$\pm$0.019
&& \textbf{0.906$\pm$0.010} & 0.839$\pm$0.010 & 0.856$\pm$0.011 & 0.645$\pm$0.027 & 0.512$\pm$0.003 & 0.645$\pm$0.005 \\
HGNN && 0.499$\pm$0.049 & \textbf{0.543$\pm$0.063} & 0.538$\pm$0.074 & 0.513$\pm$0.055 & 0.526$\pm$0.048 & 0.512$\pm$0.058
&& \textbf{0.529$\pm$0.020} & 0.523$\pm$0.020 & 0.520$\pm$0.015 & 0.512$\pm$0.018 & 0.505$\pm$0.014 & 0.513$\pm$0.017 \\
FHGCN && 0.693$\pm$0.117 & 0.651$\pm$0.090 & \textbf{0.703$\pm$0.101} & 0.536$\pm$0.061 & 0.566$\pm$0.076 & 0.550$\pm$0.069
&& 0.742$\pm$0.072 & 0.638$\pm$0.060 & \textbf{0.790$\pm$0.136} & 0.512$\pm$0.020 & 0.519$\pm$0.038 & 0.547$\pm$0.054 \\
UGCN\uppercase\expandafter{\romannumeral2} && 0.710$\pm$0.050 & 0.708$\pm$0.065 & \textbf{0.727$\pm$0.046} & 0.673$\pm$0.045 & 0.689$\pm$0.055 & 0.582$\pm$0.051
&& 0.783$\pm$0.013 & 0.726$\pm$0.009 & \textbf{0.859$\pm$0.008} & 0.724$\pm$0.014 & 0.740$\pm$0.020 & 0.706$\pm$0.013 \\
\rowcolor{Gray} Max && \textbf{0.805$\pm$0.014} & 0.773$\pm$0.017 & 0.732$\pm$0.017 & 0.673$\pm$0.045 & 0.689$\pm$0.055 &  0.592$\pm$0.022
&& \textbf{0.907$\pm$0.003} & 0.839$\pm$0.010 & 0.859$\pm$0.008 & 0.724$\pm$0.014 & 0.740$\pm$0.020 & 0.706$\pm$0.013 \\
\rowcolor{Gray} Avg. && \textbf{0.732$\pm$0.092} & 0.713$\pm$0.071 & 0.691$\pm$0.058 & 0.598$\pm$0.051 & 0.564$\pm$0.052 & 0.568$\pm$0.026 
&& \textbf{0.827$\pm$0.119} & 0.753$\pm$0.102 & 0.804$\pm$0.103 & 0.610$\pm$0.071 & 0.538$\pm$0.073 & 0.631$\pm$0.065\\
\rowcolor{Gray} Rank Avg. && \textbf{1.800$\pm$1.549} & 2.000$\pm$0.667 & 2.500$\pm$0.850 & 4.400$\pm$0.699 & 5.100$\pm$1.101 & 5.200$\pm$0.632  
&& \textbf{1.200$\pm$0.422} & 3.000$\pm$0.471 & 1.900$\pm$0.568 & 5.100$\pm$0.316 & 5.600$\pm$0.966 & 4.200$\pm$0.632  \\
\midrule
\midrule
LR &\multirow{13}{*}{\textt{CD}}& \textbf{0.921$\pm$0.004} & 0.751$\pm$0.009 & 0.602$\pm$0.004 & 0.527$\pm$0.004 & 0.504$\pm$0.002 & 0.593$\pm$0.007 & 
\multirow{13}{*}{\textt{CS}} & \textbf{0.919$\pm$0.002} & 0.767$\pm$0.002 & 0.662$\pm$0.005 & 0.541$\pm$0.005 & 0.508$\pm$0.005 & 0.625$\pm$0.007 \\
RF && \textbf{0.977$\pm$0.001} & 0.855$\pm$0.004 & 0.644$\pm$0.007 & 0.548$\pm$0.003 & 0.500$\pm$0.003 & 0.599$\pm$0.005
&& \textbf{0.984$\pm$0.001} & 0.866$\pm$0.006 & 0.702$\pm$0.006 & 0.568$\pm$0.005 & 0.501$\pm$0.003 & 0.621$\pm$0.005 \\
DT && \textbf{0.963$\pm$0.001} & 0.777$\pm$0.007 & 0.583$\pm$0.006 & 0.511$\pm$0.004 & 0.497$\pm$0.003 & 0.539$\pm$0.004 
&& \textbf{0.974$\pm$0.001} & 0.783$\pm$0.011 & 0.623$\pm$0.004 & 0.519$\pm$0.003 & 0.498$\pm$0.004 & 0.548$\pm$0.002  \\
KNN && \textbf{0.917$\pm$0.003} & 0.787$\pm$0.005 & 0.594$\pm$0.007 & 0.558$\pm$0.004 & 0.514$\pm$0.006 & 0.592$\pm$0.003
&& \textbf{0.941$\pm$0.001} & 0.834$\pm$0.002 & 0.652$\pm$0.004 & 0.578$\pm$0.003 & 0.521$\pm$0.008 & 0.606$\pm$0.001  \\
MLP && \textbf{0.969$\pm$0.001} & 0.822$\pm$0.008 & 0.637$\pm$0.008 & 0.545$\pm$0.006 & 0.502$\pm$0.004 & 0.633$\pm$0.008
&& \textbf{0.980$\pm$0.001} & 0.844$\pm$0.008 & 0.699$\pm$0.006 & 0.575$\pm$0.009 & 0.512$\pm$0.006 & 0.652$\pm$0.012 \\
XGB && \textbf{0.975$\pm$0.001} & 0.843$\pm$0.005 & 0.639$\pm$0.007 & 0.547$\pm$0.004 & 0.505$\pm$0.004 & 0.618$\pm$0.006
&& \textbf{0.984$\pm$0.001} & 0.850$\pm$0.009 & 0.704$\pm$0.006 & 0.562$\pm$0.006  & 0.506$\pm$0.003 & 0.636$\pm$0.006 \\
LGBM && \textbf{0.977$\pm$0.004} & 0.849$\pm$0.020 & 0.652$\pm$0.022 & 0.546$\pm$0.020 & 0.504$\pm$0.003 & 0.601$\pm$0.006
&& \textbf{0.984$\pm$0.002} & 0.860$\pm$0.032 & 0.713$\pm$0.016 & 0.560$\pm$0.028 & 0.502$\pm$0.002 & 0.622$\pm$0.007 \\
HGNN && \textbf{0.595$\pm$0.009} & 0.543$\pm$0.015 & 0.534$\pm$0.013 & 0.542$\pm$0.012 & 0.535$\pm$0.017 & 0.519$\pm$0.013
&& 0.555$\pm$0.007 & 0.534$\pm$0.011 & 0.529$\pm$0.008 & 0.553$\pm$0.009 & \textbf{0.568$\pm$0.009} & 0.521$\pm$0.006 \\
FHGCN && \textbf{0.754$\pm$0.091} & 0.597$\pm$0.087 & 0.556$\pm$0.067 & 0.505$\pm$0.015 & 0.504$\pm$0.013 & 0.502$\pm$0.008
&& \textbf{0.738$\pm$0.098} & 0.609$\pm$0.086 & 0.594$\pm$0.089 & 0.503$\pm$0.013 & 0.502$\pm$0.009 & 0.503$\pm$0.009 \\
UGCN\uppercase\expandafter{\romannumeral2} && \textbf{0.932$\pm$0.005} & 0.798$\pm$0.013 & 0.657$\pm$0.012 & 0.769$\pm$0.016 & 0.630$\pm$0.028 & 0.541$\pm$0.011
&& \textbf{0.917$\pm$0.006} & 0.739$\pm$0.010 & 0.718$\pm$0.010 & 0.827$\pm$0.016 & 0.823$\pm$0.009 & 0.578$\pm$0.012 \\
\rowcolor{Gray} Max && \textbf{0.977$\pm$0.001} & 0.855$\pm$0.004 & 0.657$\pm$0.012 & 0.769$\pm$0.016 & 0.630$\pm$0.028 & 0.633$\pm$0.008
&& \textbf{0.984$\pm$0.002} & 0.866$\pm$0.006 & 0.718$\pm$0.010 & 0.827$\pm$0.016 & 0.823$\pm$0.009 & 0.652$\pm$0.012 \\
\rowcolor{Gray} Avg. && \textbf{0.898$\pm$0.126} & 0.762$\pm$0.107 & 0.610$\pm$0.043 & 0.560$\pm$0.075 & 0.520$\pm$0.040 & 0.574$\pm$0.045
&& \textbf{0.898$\pm$0.142} & 0.769$\pm$0.114 & 0.660$\pm$0.062 & 0.579$\pm$0.091 & 0.544$\pm$0.100 & 0.591$\pm$0.051 \\
\rowcolor{Gray} Rank Avg. && \textbf{1.000$\pm$0.000} & 2.000$\pm$0.000 & 3.300$\pm$0.675 & 4.500$\pm$0.850 & 5.600$\pm$0.699 & 4.600$\pm$0.966  
&& \textbf{1.100$\pm$0.316} & 2.400$\pm$0.843 & 3.400$\pm$0.843 & 4.400$\pm$1.075 & 5.200$\pm$1.751 & 4.500$\pm$0.850  \\
\midrule
\midrule
LR &\multirow{13}{*}{\textt{QM}}& \textbf{0.604$\pm$0.003} & 0.579$\pm$0.004 & 0.553$\pm$0.003 & 0.500$\pm$0.001 & 0.504$\pm$0.001 & 0.566$\pm$0.002 & 
\multirow{13}{*}{\textt{QS}} & \textbf{0.561$\pm$0.002} & 0.533$\pm$0.003 & 0.530$\pm$0.003 & 0.502$\pm$0.001 & 0.509$\pm$0.002 & 0.556$\pm$0.002 \\
RF && \textbf{0.673$\pm$0.003} & 0.613$\pm$0.005 & 0.620$\pm$0.004 & 0.503$\pm$0.003 & 0.502$\pm$0.003 & 0.581$\pm$0.005
&& \textbf{0.661$\pm$0.001} & 0.565$\pm$0.002 & 0.590$\pm$0.002 & 0.500$\pm$0.002 & 0.503$\pm$0.001 & 0.565$\pm$0.004\\
DT && \textbf{0.617$\pm$0.003} & 0.547$\pm$0.006 & 0.572$\pm$0.004 & 0.502$\pm$0.001 & 0.502$\pm$0.004 & 0.547$\pm$0.004
&& \textbf{0.626$\pm$0.002} & 0.528$\pm$0.002 & 0.572$\pm$0.003 & 0.500$\pm$0.002 & 0.501$\pm$0.002 & 0.545$\pm$0.002\\
KNN && \textbf{0.601$\pm$0.002} & 0.576$\pm$0.003 & 0.553$\pm$0.003 & 0.504$\pm$0.002 & 0.506$\pm$0.004 & 0.523$\pm$0.001
&& \textbf{0.619$\pm$0.002} & 0.538$\pm$0.002 & 0.566$\pm$0.002 & 0.504$\pm$0.002 & 0.502$\pm$0.003 & 0.576$\pm$0.001 \\
MLP && \textbf{0.679$\pm$0.002} & 0.598$\pm$0.004 & 0.590$\pm$0.002 & 0.505$\pm$0.002 & 0.503$\pm$0.002 & 0.611$\pm$0.011
&& \textbf{0.668$\pm$0.001} & 0.544$\pm$0.003 & 0.603$\pm$0.003 & 0.509$\pm$0.002 & 0.506$\pm$0.003 & 0.598$\pm$0.013 \\
XGB && \textbf{0.681$\pm$0.002} & 0.607$\pm$0.005 & 0.629$\pm$0.003 & 0.502$\pm$0.002 & 0.508$\pm$0.002 & 0.601$\pm$0.004
&& \textbf{0.679$\pm$0.001} & 0.575$\pm$0.002 & 0.612$\pm$0.003 & 0.503$\pm$0.001 & 0.503$\pm$0.001 & 0.585$\pm$0.005 \\
LGBM && \textbf{0.690$\pm$0.009} & 0.618$\pm$0.019 & 0.641$\pm$0.012 & 0.504$\pm$0.009 & 0.506$\pm$0.002 & 0.577$\pm$0.005
&& \textbf{0.688$\pm$0.004} & 0.585$\pm$0.011 & 0.624$\pm$0.013 & 0.504$\pm$0.006 & 0.503$\pm$0.002 & 0.567$\pm$0.005 \\
HGNN && \textbf{0.549$\pm$0.011} & 0.525$\pm$0.011 & 0.529$\pm$0.008 & 0.520$\pm$0.011 & 0.546$\pm$0.007 & 0.545$\pm$0.008
&& \textbf{0.579$\pm$0.006} & 0.539$\pm$0.005 & 0.546$\pm$0.005 & 0.535$\pm$0.007 & 0.572$\pm$0.006 & 0.571$\pm$0.005 \\
FHGCN && 0.512$\pm$0.022 & 0.503$\pm$0.007 & \textbf{0.516$\pm$0.029} & 0.500$\pm$0.003 & 0.502$\pm$0.008 & 0.507$\pm$0.016
&& \textbf{0.530$\pm$0.041} & 0.505$\pm$0.014 & 0.515$\pm$0.027 & 0.501$\pm$0.003 & 0.501$\pm$0.005 & 0.502$\pm$0.009 \\
UGCN\uppercase\expandafter{\romannumeral2} && 0.607$\pm$0.006 & 0.583$\pm$0.009 & 0.599$\pm$0.010 & 0.615$\pm$0.010 & 0.637$\pm$0.008 & \textbf{0.743$\pm$0.013}
&& 0.645$\pm$0.005 & 0.595$\pm$0.007 & 0.605$\pm$0.005 & 0.563$\pm$0.005 & 0.653$\pm$0.006 & \textbf{0.752$\pm$0.007} \\
\rowcolor{Gray} Max && 0.690$\pm$0.009 & 0.618$\pm$0.019 & 0.641$\pm$0.012 & 0.615$\pm$0.010  & 0.637$\pm$0.008 & \textbf{0.743$\pm$0.013}
&& 0.645$\pm$0.005 & 0.595$\pm$0.007 & 0.624$\pm$0.013 & 0.563$\pm$0.005 & 0.653$\pm$0.006 & \textbf{0.752$\pm$0.007} \\
\rowcolor{Gray} Avg. && \textbf{0.621$\pm$0.060} & 0.575$\pm$0.039 & 0.580$\pm$0.043 & 0.516$\pm$0.035 & 0.522$\pm$0.043 & 0.580$\pm$0.066
&& \textbf{0.626$\pm$0.053} & 0.551$\pm$0.028 & 0.576$\pm$0.037 & 0.512$\pm$0.021 & 0.525$\pm$0.050 & 0.582$\pm$0.065 \\
\rowcolor{Gray} Rank Avg. && \textbf{1.400$\pm$0.966} & 3.400$\pm$1.265 & 2.900$\pm$1.287 & 5.500$\pm$0.972 & 4.600$\pm$1.430 & 3.200$\pm$1.033 
&& \textbf{1.100$\pm$0.316} & 3.900$\pm$0.738 & 2.700$\pm$0.949 & 5.600$\pm$0.516 & 4.900$\pm$1.370 & 2.800$\pm$0.919  \\
\midrule
\midrule
LR &\multirow{13}{*}{\textt{B4}}& \textbf{0.626$\pm$0.004} & \multirow{13}{*}{O.O.T.*} & 0.524$\pm$0.000 & 0.527$\pm$0.006 & 0.527$\pm$0.004 & 0.580$\pm$0.004 & 
\multirow{13}{*}{\textt{B5}} & 0.566$\pm$0.000 & \multirow{13}{*}{O.O.T.*} & 0.539$\pm$0.001 & 0.528$\pm$0.008 & 0.532$\pm$0.004 & \textbf{0.569$\pm$0.006} \\
RF && \textbf{0.927$\pm$0.002} && 0.731$\pm$0.008 & 0.540$\pm$0.002 & 0.517$\pm$0.002 & 0.608$\pm$0.015
&& \textbf{0.926$\pm$0.004} && 0.686$\pm$0.012 & 0.552$\pm$0.002 & 0.520$\pm$0.002 & 0.598$\pm$0.010 \\
DT && \textbf{0.896$\pm$0.002} && 0.703$\pm$0.007 & 0.507$\pm$0.003 & 0.503$\pm$0.001 & 0.549$\pm$0.004
&& \textbf{0.889$\pm$0.006} && 0.664$\pm$0.012 & 0.512$\pm$0.002 & 0.504$\pm$0.001 & 0.543$\pm$0.003\\
KNN && \textbf{0.838$\pm$0.003} && 0.699$\pm$0.003 & 0.577$\pm$0.003 & 0.537$\pm$0.002 & 0.639$\pm$0.001
&& \textbf{0.839$\pm$0.008} && 0.658$\pm$0.010 & 0.592$\pm$0.002 & 0.539$\pm$0.003 & 0.628$\pm$0.001 \\
MLP && \textbf{0.922$\pm$0.000} && 0.612$\pm$0.011 & 0.555$\pm$0.011 & 0.538$\pm$0.005 & 0.629$\pm$0.013
&& \textbf{0.920$\pm$0.003} && 0.635$\pm$0.009 & 0.571$\pm$0.007 & 0.553$\pm$0.007 & 0.619$\pm$0.009 \\
XGB && \textbf{0.935$\pm$0.001} && 0.760$\pm$0.010 & 0.554$\pm$0.003 & 0.520$\pm$0.004 & 0.615$\pm$0.012
&& \textbf{0.939$\pm$0.002} && 0.721$\pm$0.014 & 0.571$\pm$0.003 & 0.519$\pm$0.003 & 0.606$\pm$0.007 \\
LGBM && \textbf{0.935$\pm$0.002} && 0.765$\pm$0.036 & 0.560$\pm$0.018 & 0.526$\pm$0.005 & 0.606$\pm$0.013
&& \textbf{0.941$\pm$0.003} && 0.725$\pm$0.038 & 0.574$\pm$0.019 & 0.526$\pm$0.004 & 0.596$\pm$0.008 \\
HGNN && \textbf{0.727$\pm$0.003} && 0.633$\pm$0.003 & 0.632$\pm$0.004 & 0.628$\pm$0.003 & 0.600$\pm$0.003
&& \textbf{0.750$\pm$0.003} && 0.636$\pm$0.003 & 0.641$\pm$0.006 & 0.630$\pm$0.004 & 0.603$\pm$0.002 \\
FHGCN && \textbf{0.705$\pm$0.103} && 0.623$\pm$0.090 & 0.504$\pm$0.009 & 0.504$\pm$0.010 & 0.505$\pm$0.009
&& \textbf{0.699$\pm$0.112} && 0.599$\pm$0.098 & 0.506$\pm$0.013 & 0.503$\pm$0.008 & 0.503$\pm$0.006 \\
UGCN\uppercase\expandafter{\romannumeral2} && \textbf{0.845$\pm$0.003} && 0.771$\pm$0.003 & 0.773$\pm$0.008 & 0.711$\pm$0.005 & 0.652$\pm$0.003
&& \textbf{0.858$\pm$0.002} && 0.773$\pm$0.005 & 0.782$\pm$0.005 & 0.717$\pm$0.003 & 0.651$\pm$0.003 \\
\rowcolor{Gray} Max && \textbf{0.935$\pm$0.002} && 0.771$\pm$0.003 & 0.773$\pm$0.008 & 0.711$\pm$0.005 & 0.652$\pm$0.003
&& \textbf{0.941$\pm$0.003} && 0.773$\pm$0.005 & 0.782$\pm$0.005 & 0.717$\pm$0.003 & 0.651$\pm$0.003 \\
\rowcolor{Gray} Avg. && \textbf{0.836$\pm$0.112} && 0.682$\pm$0.081 & 0.573$\pm$0.079 & 0.551$\pm$0.066 & 0.598$\pm$0.044 
&& \textbf{0.833$\pm$0.124} && 0.664$\pm$0.067 & 0.583$\pm$0.081 & 0.554$\pm$0.068 & 0.592$\pm$0.043\\
\rowcolor{Gray} Rank Avg. && \textbf{1.000$\pm$0.000} && 2.500$\pm$0.972 & 3.700$\pm$0.675 & 4.600$\pm$0.699 & 3.200$\pm$1.033  
&& \textbf{1.100$\pm$0.316} && 2.300$\pm$0.483 & 3.600$\pm$0.966 & 4.700$\pm$0.483 & 3.300$\pm$1.160  \\
\midrule
LR &\multirow{13}{*}{\textt{B6}}& 0.556$\pm$0.001 & \multirow{13}{*}{O.O.T.*} & 0.549$\pm$0.005 & 0.527$\pm$0.006 & 0.530$\pm$0.005 & \textbf{0.569$\pm$0.008} \\
RF && \textbf{0.934$\pm$0.002} && 0.694$\pm$0.005 & 0.529$\pm$0.003 & 0.521$\pm$0.002 & 0.609$\pm$0.013  \\
DT && \textbf{0.907$\pm$0.005} && 0.671$\pm$0.006 & 0.509$\pm$0.002  & 0.504$\pm$0.001 & 0.547$\pm$0.003\\
KNN && \textbf{0.866$\pm$0.006} && 0.668$\pm$0.006 & 0.567$\pm$0.003  & 0.536$\pm$0.003 & 0.626$\pm$0.001 \\
MLP && \textbf{0.923$\pm$0.001} && 0.591$\pm$0.008 & 0.549$\pm$0.005 & 0.545$\pm$0.006 & 0.630$\pm$0.015 \\
XGB && \textbf{0.938$\pm$0.002} && 0.722$\pm$0.005 & 0.544$\pm$0.002 & 0.522$\pm$0.003 & 0.621$\pm$0.010 \\
LGBM && \textbf{0.937$\pm$0.003} && 0.733$\pm$0.041 & 0.551$\pm$0.016 & 0.529$\pm$0.004 & 0.610$\pm$0.010 \\
HGNN && \textbf{0.749$\pm$0.003} && 0.633$\pm$0.003 & 0.627$\pm$0.005 & 0.629$\pm$0.003 & 0.597$\pm$0.004 \\
FHGCN && \textbf{0.712$\pm$0.110} && 0.608$\pm$0.100 & 0.504$\pm$0.011 & 0.504$\pm$0.010 & 0.504$\pm$0.011 \\
UGCN\uppercase\expandafter{\romannumeral2} && \textbf{0.854$\pm$0.003} && 0.762$\pm$0.005 & 0.771$\pm$0.010 & 0.720$\pm$0.004 & 0.646$\pm$0.004 \\
\rowcolor{Gray} Max && \textbf{0.938$\pm$0.002} && 0.762$\pm$0.005 & 0.771$\pm$0.010 & 0.720$\pm$0.004 & 0.646$\pm$0.004 \\
\rowcolor{Gray} Avg. && \textbf{0.838$\pm$0.127} && 0.663$\pm$0.068 & 0.568$\pm$0.079 & 0.554$\pm$0.068 & 0.596$\pm$0.044 \\
\rowcolor{Gray} Rank Avg. && \textbf{1.100$\pm$0.316} && 2.300$\pm$0.483 & 4.000$\pm$0.816 & 4.400$\pm$0.843 & 3.200$\pm$1.229  \\
\bottomrule
\multicolumn{15}{l}{* O.O.T: out-of-time ($>1$ day).}
\end{tabular}
}%
\end{table*}

\bibliographystyle{IEEEtran}
\bibliography{bib.bib}
\clearpage